\documentclass[aps,pre,twocolumn]{revtex4-2}

\usepackage{amsmath,amssymb,amsthm}
\usepackage{graphicx}
\usepackage{xcolor}
\usepackage{newtxtext}
\usepackage[slantedGreek,cmintegrals]{newtxmath}
\usepackage{hyperref}
\usepackage{enumitem}
\usepackage{stmaryrd}
\usepackage{accents}

\hypersetup{colorlinks,linkcolor={blue!90!black},citecolor={blue!90!black},urlcolor={blue!90!black}}

\renewcommand{\vec}[1]{\boldsymbol{#1}}

\renewcommand{\pi}{\uppi}
\newcommand{\btilde}[1]{\boldsymbol{\tilde{#1}}}
\newcommand{\bhat}[1]{\boldsymbol{\hat{#1}}}
\newcommand{\bpm}{{\boldsymbol{\pm}}}

\DeclareMathAlphabet{\mathcal}{OMS}{cmsy}{m}{n}
\DeclareMathAlphabet{\mathbfsf}{\encodingdefault}{\sfdefault}{b}{n}
\newcommand{\mat}[1]{\mathsf{#1}}
\newcommand{\tens}[1]{\mathbfsf{#1}}

\DeclareFontEncoding{LGR}{}{}
\DeclareSymbolFont{sfgreek}{LGR}{cmss}{m}{n}
\DeclareMathSymbol{\svarsigma}{\mathord}{sfgreek}{`c}
\DeclareMathSymbol{\skappa}{\mathord}{sfgreek}{`k}
\DeclareMathSymbol{\slambdasym}{\mathord}{sfgreek}{`l}
\newcommand{\slambda}{\scalebox{0.91}{$\slambdasym$}}
\DeclareSymbolFont{sfgreek2}{LGR}{cmss}{bx}{n}
\DeclareMathSymbol{\skappabsym}{\mathord}{sfgreek2}{`k}
\newcommand{\skappab}{\scalebox{0.95}{$\skappabsym$}}
\DeclareMathSymbol{\slambdabsym}{\mathord}{sfgreek2}{`l}
\newcommand{\slambdab}{\scalebox{0.91}{$\slambdabsym$}}
\newcommand{\figref}[2]{[Fig.~\hyperref[#1]{\ref*{#1}(#2)}]}
\newcommand{\figrefi}[2]{[Fig.~\hyperref[#1]{\ref*{#1}(#2)}, inset]}
\newcommand{\textfigref}[2]{Fig.~\hyperref[#1]{\ref*{#1}(#2)}}
\newcommand{\wholefigref}[1]{(Fig.~\ref{#1})}

\newcommand{\textwholefigref}[1]{Fig.~\ref{#1}}

\newcommand{\figrefp}[2]{\hyperref[#1]{\ref*{#1}(#2)}}

\renewcommand{\leq}{\leqslant}
\renewcommand{\geq}{\geqslant}

\renewcommand{\Gamma}{\upGamma}
\newtheorem{lemma}{Lemma}
\newtheorem{corollary}{Corollary}
\newtheorem{proposition}{Proposition}
\DeclareMathOperator{\tr}{tr}
\DeclareMathOperator{\adj}{adj}
\begin{document}
\title{Morphoelasticity of Large Bending Deformations of Cell Sheets during Development}
\author{Pierre A. Haas}
\email{haas@maths.ox.ac.uk}
\altaffiliation[current address: ]{Mathematical Institute, University of Oxford, Woodstock Road, Oxford OX2 6GG, United Kingdom}
\author{Raymond E. Goldstein}
\email{r.e.goldstein@damtp.cam.ac.uk}
\affiliation{Department of Applied Mathematics and Theoretical Physics, Centre for Mathematical Sciences, \\ University of Cambridge, 
Wilberforce Road, Cambridge CB3 0WA, United Kingdom}
\date{\today}%
\begin{abstract}
Deformations of cell sheets during morphogenesis are driven by developmental processes such as cell division and cell shape changes. In morphoelastic shell theories of development, these processes appear as variations of the intrinsic geometry of a thin elastic shell. However, morphogenesis often involves large bending deformations that are outside the formal range of validity of these shell theories. Here, by asymptotic expansion of three-dimensional incompressible morphoelasticity in the limit of a thin shell, we derive a shell theory for large intrinsic bending deformations and emphasise the resulting geometric material anisotropy and the elastic role of cell constriction. Taking the invagination of the green alga \emph{Volvox} as a model developmental event, we show how results for this theory differ from those for a classical shell theory that is not formally valid for these large bending deformations and reveal how these geometric effects stabilise invagination.
\end{abstract}
\maketitle

\section{Introduction}
Cell division, cell shape changes, and related processes can drive deformations of cell sheets during animal and plant development~\cite{keller03,leptin05,lecuit07,keller11,lecuit11,tada12}. In elastic continuum theories of the development of the green alga \emph{Volvox}~\cite{hohn15,haas15,haas18a,haas18b}, of tissue folding in \emph{Drosophila}~\cite{heer17,yevick19}, or of more abstract active surfaces~\cite{miller18}, these driving processes appear as changes of the reference or intrinsic geometry of thin elastic shells.

Just as classical thin shell theories arise from an asymptotic expansion of bulk elasticity in the small thickness of the shell~\cite{ciarlet05,* [] [{ Chap.~6, pp.~159--213, Chap.~12, pp.~435--453, and App.~D, pp.~571--581.}] audoly,steigmann13}, these ``morphoelastic'' shell theories should be asymptotic limits of a bulk theory. While there is now a well-established framework of three-dimensional morphoelasticity~\cite{* [] [{ Chap.~11, pp.~261--344 and Chap.~12, pp.~345--373.}] goriely,ambrosi19}, based on a multiplicative decomposition of the deformation gradient tensor into intrinsic and elastic deformations~\cite{rodriguez94}, studies of this asymptotic limit have mostly been restricted to the case of flat morphoelastic plates. Extensions of the classical F\"oppl--von K\'arm\'an equations~\cite{dervaux08,dervaux09} have been derived and residual stresses in Kirchhoff plate theories~\cite{mcmahon11} have been studied in this case. A theory of non-Euclidean plates~\cite{efrati09} has been developed in parallel. Apart from a general geometric theory of morphoelastic surfaces~\cite{sadik16}, studies of morphoelastic shells have remained more phenomenological, however: some models~\cite{hohn15,haas15,heer17,miller18,yevick19} simply replaced the elastic strains in classical shell theories~\cite{ventsel,libai,audoly} with measures of the difference of the intrinsic and deformed geometries.  Other studies~\cite{haas18a,haas18b} took a more geometric approach, mirroring geometric derivations of classical shell theories~\cite{ventsel} based on the so-called Kirchhoff ``hypothesis''. This is the asymptotic result~\cite{audoly} that the normals of the midsurface of the undeformed shell remain, at leading order, normal to the deformed midsurface.

\begin{figure}[b]
\centering\includegraphics{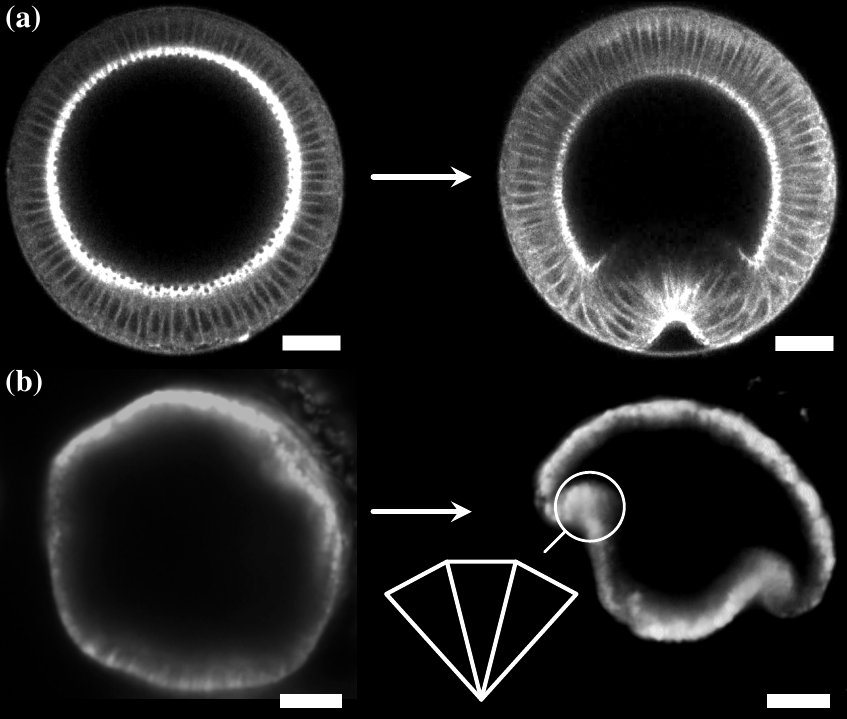}
\caption{Large bending deformations during morphogenesis: even if the thickness of the cell sheet is small compared to the undeformed radius of curvature, the local radius of curvature need not remain large compared to the cell sheet thickness as the sheet deforms. (a)~Cross section of ventral furrow formation in \emph{Drosophila}, reproduced from Ref.~\cite{conte12}. (b) Midsagittal cross section of invagination in the spherical alga \emph{Volvox globator}, reproduced from Ref.~\cite{haas15}. Inset: cartoon of constricted triangular cells in the bend region. Scale bars: $20\,\text{\textmu m}$.}
\label{fig1}
\end{figure}

There is however one more serious limitation of these models: tissues in development undergo large bending deformations~\wholefigref{fig1} that are outside the formal range of validity of the underlying thin shell theories, which assume that the thickness of the shell is much smaller than all lengthscales of the midsurface of the shell~\cite{audoly,ventsel,libai}. However, even if the thickness of the cell sheet is much smaller than its undeformed radius of curvature, this radius of curvature may become comparable, locally, to the thickness of the cell sheet as it deforms~\wholefigref{fig1}. This is associated with cells contracting at one cell pole to splay and thereby bend the cell sheet~\cite{keller11}.

Here, we derive a theory of thin incompressible morphoelastic shells undergoing large bending deformations by asymptotic expansion of three-dimensional elasticity. We reveal how, even in a constitutively isotropic material, this new scaling limit of large bending deformations induces, in the thin shell limit, a geometric anisotropy absent from classical shell theories: different deformation directions exhibit different deformation responses. We stress how this geometric effect is associated with the geometric singularity of cell constriction, i.e. the limit of wedged triangular cells~\figrefi{fig1}{b} associated with these large bending deformations. Specialising to the invagination of the green alga \emph{Volvox}~\cite{hallmann06,hohn11}, we then show how results for this theory differ from those for a classical theory that is not formally valid in this large bending limit, and reveal how invagination is stabilised by the geometry of large bending deformations.

\section{Elastic Model}\label{sec:model}
In this section, we describe large bending deformations of a thin incompressible morphoelastic shell, starting from three-dimensional morphoelasticity. We shall have to distinguish between three configurations of the shell~\figref{fig2}{a}: (i) the undeformed configuration of the shell, (ii) the deformed configuration of the shell, and (iii) the intrinsic configuration of the shell that encodes the local, intrinsic deformations of the shell, i.e. the cell shape changes or cell division in the biological system. These intrinsic deformations are not in general compatible with the global geometry of the shell: in other words, this intrinsic configuration cannot in general be embedded into three-dimensional Euclidean space~\cite{goriely}. Elasticity must therefore intervene to ``glue'' the intrinsically deformed infinitesimal patches of cell sheet back together, as illustrated in~\textfigref{fig2}{a}. Configurations (i) and (ii) are related by the geometric deformation gradient $\btilde{\tens{F}}$. This tensor decomposes multiplicatively into an intrinsic contribution $\tens{F^0}$ that relates configurations (i) and~(iii), and an elastic contribution $\smash{\tens{F}=\btilde{\tens{F}}\bigl(\tens{F^0}\bigr)^{-1}}$. This is the multiplicative decomposition of morphoelasticity~\cite{goriely,ambrosi19}.

\begin{figure}
\centering\includegraphics{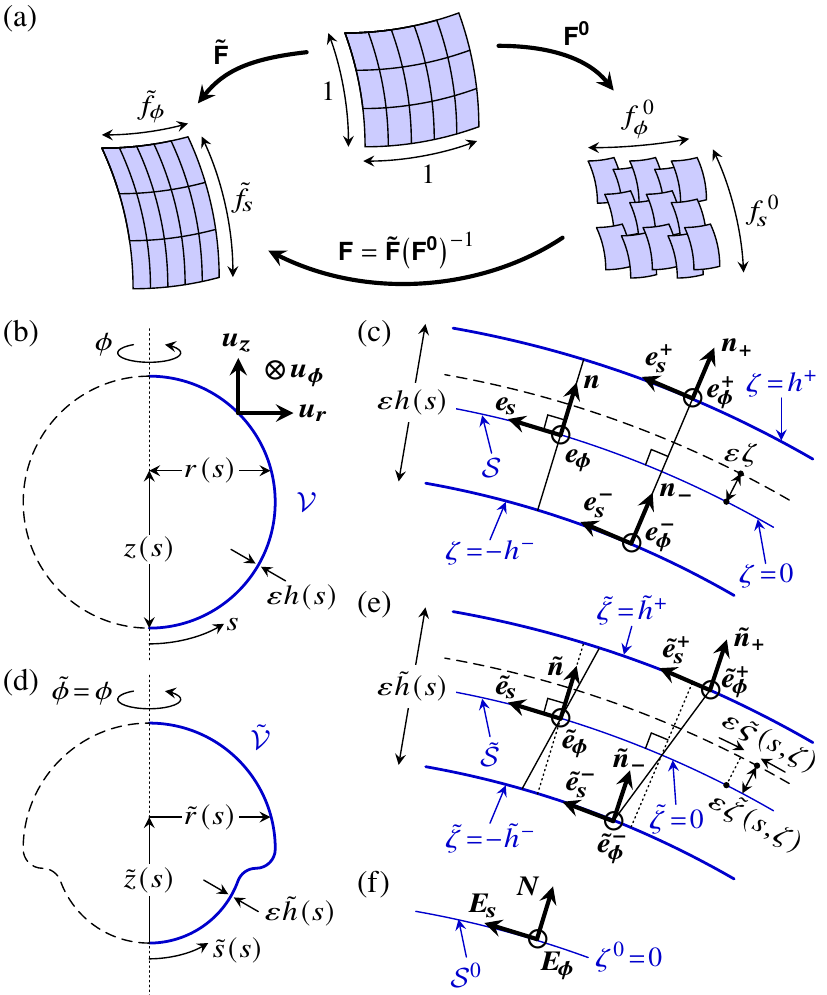}
\caption{Morphoelasticity of an axisymmetric shell. (a) The undeformed (top), deformed (left), and intrinsic (right) configurations of the shell are related by the three tensors $\btilde{\tens{F}}$, $\tens{F^0}$, and $\smash{\tens{F}=\tilde{\tens{F}}\bigl(\tens{F^0}\bigr)^{-1}}$. The geometric and intrinsic midsurface stretches are $\tilde{f}_s,\tilde{f}_\phi$ and $\smash{f_s^0,f_{\smash{\phi}}^0}$. (b)~Undeformed configuration $\mathcal{V}$ of an axisymmetric shell of thickness $\varepsilon h(s)$, described by coordinates $r(s),z(s)$, where $s$ is arclength, with respect to the basis $\{\vec{u_r},\vec{u_\phi},\vec{u_z}\}$ of cylindrical polars. (c)~Cross section of the undeformed shell, defining a basis $\mathcal{B}=\{\vec{e_s},\vec{e_\phi},\vec{n}\}$ and the transverse coordinate~$\zeta$. The surfaces of the undeformed shell are at $\zeta=\pm h^\pm(s)$, where the tangent vectors are $\vec{e_s^\pm},\vec{e_{\smash\phi}^\pm}$, and the normal is $\vec{n^\pm}$. (d)~Deformed configuration~$\smash{\tilde{\mathcal{V}}}$ of the shell: after a torsionless deformation, the shell has thickness $\varepsilon\tilde{h}(s)$, arclength $\tilde{s}$, and is described by coordinates $\tilde{r}(s),\tilde{z}(s)$ with respect to cylindrical polars. (e)~Cross section of the deformed shell, defining a basis $\tilde{\mathcal{B}}=\{\vec{\tilde{e}_s},\vec{\tilde{e}_\phi},\vec{\tilde{n}}\}$. Normals to the midsurface rotate so that a point at a distance $\varepsilon\zeta$ from the undeformed midsurface $\mathcal{S}$ is at a distance $\varepsilon\tilde{\zeta}(s,\zeta)$ from the deformed midsurface $\tilde{\mathcal{S}}$, and displaced by a distance $\varepsilon\tilde{\varsigma}(s,\zeta)$ parallel to $\tilde{\mathcal{S}}$. At the surfaces $\tilde{\zeta}=\pm\tilde{h}^\pm(s)$ of the deformed shell, the tangent vectors are $\vec{\tilde{e}_s^\pm},\vec{\tilde{e}_{\smash\phi}^\pm}$, and the normal is $\vec{\tilde{n}^\pm}$. (f)~The intrinsic midsurface $\mathcal{S}^0$, on which $\zeta^0=0$, embeds, locally, into three-dimensional space to define an intrinsic basis $\smash{\mathcal{B}^0}=\{\vec{E_s},\vec{E_\phi},\vec{N}\}$.}
\label{fig2} 
\end{figure}

In this section, we restrict to torsionless deformations of an axisymmetric shell. The analysis can be extended to more general deformations of the shell, and, for the sake of completeness, we do so in Appendix~\ref{appA}, but the restriction to axisymmetric deformations eschews the mire of tensorial notation that arises in the general case.

The derivation of the shell theory for large bending deformations divides, like derivations of classical shell theories, into two steps: first, in subsection~\ref{sec:modela}, we describe the kinematics of the deformation and derive expressions for the geometric, intrinsic, and elastic deformations gradients. Second, in subsection~\ref{sec:modelb}, we analyse the mechanics of the shell and expand the three-dimensional elastic energy and equilibrium conditions asymptotically. At the end of this section, in subsection~\ref{sec:modelc}, we discuss the limit of small bending deformations that gives rise to classical shell theories.
\subsection{Axisymmetric deformations of an elastic shell}\label{sec:modela}
We consider an elastic shell of undeformed thickness $\varepsilon h$, where $\varepsilon\ll 1$ is a small asymptotic parameter expressing the thinness of the shell compared to other lengthscales associated with its midsurface. Large bending deformations will be introduced in Section~\ref{sec:modelb} by allowing one of the intrinsic radii of curvature of the shell to be of order $O(\varepsilon)$. We begin by deriving an expression for the elastic deformation gradient $\tens{F}$ for torsionless deformations of an axisymmetric shell. 

\subsubsection{Undeformed configuration of the shell}
We will describe the undeformed configuration $\mathcal{V}$ of the shell with reference to a midsurface $\mathcal{S}$ that we will choose later. With respect to the basis $\{\vec{u_r},\vec{u_\phi},\vec{u_z}\}$ of cylindrical coordinates, we define the position vector of a point on $\mathcal{S}$,
\begin{align}
\vec{\rho}(s,\phi)=r(s)\vec{u_r}(\phi)+z(s)\vec{u_z},\label{eq:mundef}
\end{align}
with $s$ denoting arclength and $\phi$ being the azimuthal coordinate~\figref{fig2}{b}. The tangent angle $\psi(s)$ of $\mathcal{S}$ is defined by
\begin{align}
&r'(s)=\cos{\psi(s)},&& z'(s)=\sin{\psi(s)},  
\end{align}
in which dashes denote differentiation with respect to $s$. The vectors
\begin{align}
&\vec{e_s}(s,\phi)=\cos{\psi(s)}\vec{u_r}(\phi)+\sin{\psi(s)}\vec{u_z},&&\vec{e_\phi}(\phi)=\vec{u_\phi}(\phi) 
\end{align}
thus constitute a basis of the tangent space of $\mathcal{S}$~\figref{fig2}{c}, which we extend to a (right-handed) orthonormal basis ${\mathcal{B}=\{\vec{e_s},\vec{e_\phi},\vec{n}\}}$ for $\mathcal{V}$ by adjoining the normal to $\mathcal{S}$,
\begin{align}
\vec{n}(s,\phi)=\cos{\psi(s)}\vec{u_z}-\sin{\psi(s)}\vec{u_r}(\phi).
\end{align}
In particular, $\vec{n}=\vec{e_s}\times\vec{e_\phi}$. We complete the description of $\mathcal{S}$ by computing its curvatures,
\begin{align}
&\varkappa_s(s)=\psi'(s),&& \varkappa_\phi(s)=\dfrac{\sin{\psi(s)}}{r(s)}.\label{eq:curvundef}
\end{align}
Now, the position of a point in $\mathcal{V}$ is
\begin{align}
\vec{r}(s,\phi,\zeta)=\vec{\rho}(s,\phi)+\varepsilon\zeta\vec{n}(s,\phi), \label{eq:undef}
\end{align}
where we have introduced the transverse coordinate $\zeta$, which is such that the shell surfaces are at $\zeta=\pm h^\pm(s)$ \figref{fig2}{c}. Noting the derivatives $\partial\vec{n}/\partial s=-\varkappa_s\vec{e_s}$ and $\partial\vec{n}/\partial\phi=-\varkappa_\phi\vec{e_\phi}$, we obtain the tangent basis of $\mathcal{V}$,
\begin{align}
&\dfrac{\partial\vec{r}}{\partial s}=(1-\varepsilon\varkappa_s\zeta)\vec{e_s},&& \dfrac{\partial\vec{r}}{\partial \phi}=r(1-\varepsilon\varkappa_\phi\zeta)\vec{e_\phi},&&\dfrac{\partial\vec{r}}{\partial\zeta}=\varepsilon\vec{n},\label{eq:pdundef}
\end{align}
from which follows the expression for the Riemannian metric of the undeformed configuration,
\begin{subequations}\label{eq:undefmetric}
\begin{align}
\chi_s^2\,\mathrm{d}s^2+\chi_\phi^2\,\mathrm{d}\phi^2+\chi_\zeta^2\,\mathrm{d}\zeta^2
\end{align}
with associated scale factors
\begin{align}
&\chi_s=1-\varepsilon\varkappa_s\zeta,&& \chi_\phi=r(1-\varepsilon\varkappa_\phi\zeta),&&\chi_\zeta=\varepsilon.\label{eq:sfundef}
\end{align}
and hence volume element
\begin{align}
\mathrm{d}V=\chi_s\chi_\phi \chi_\zeta\,\mathrm{d}s\,\mathrm{d}\phi\,\mathrm{d}\zeta=\varepsilon(1-\varepsilon\varkappa_s\zeta)(1-\varepsilon\varkappa_\phi\zeta)\,r\,\mathrm{d}s\,\mathrm{d}\phi\,\mathrm{d}\zeta.\label{eq:dVundef}
\end{align}
\end{subequations}
The position vectors of the surfaces $\zeta=\pm h^\pm(s)$ of the undeformed shell are
\begin{subequations}
\begin{align}
\vec{r^\pm}(s,\phi,\zeta)= \vec{\rho}(s,\phi)\pm\varepsilon h^\pm(s)\vec{n}(s,\phi),
\end{align}
so that, using commata to denote partial differentiation,
\begin{align}
\dfrac{\partial\vec{r^\pm}}{\partial s}=\left(1\mp\varepsilon\varkappa_sh^\pm\right)\vec{e_s}\pm\varepsilon h^\pm_{,s}\vec{n},
\end{align}
\end{subequations}
in which commata denote partial differentiation. The unit tangent vectors to these shell surfaces are $\vec{e_s^\pm}\parallel\partial\vec{r^\pm}/\partial s$ and $\vec{e_{\smash\phi}^\pm}=\vec{e_\phi}$, in which the symbol $\parallel$ expresses parallelism and hides a normalisation factor for the unit vector on the left-hand side. By definition, the unit normals $\vec{n^\pm}$ to the deformed shell surfaces~\figref{fig2}{c} obey $\vec{n^\pm}\parallel\smash{\vec{e_{\smash{s}}^\pm}}\times\vec{e_{\smash{\phi}}^\pm}$. Now introducing the normalisation factor explicitly, we find
\begin{align}
\vec{n^\pm}=\dfrac{\vec{n}\mp\nu_\pm\vec{e_s}}{\sqrt{1+\nu_\pm^2}}\quad\text{with }\nu_\pm=\dfrac{\varepsilon h^\pm_{,s}}{1\mp\varepsilon\varkappa_s h^\pm}.\label{eq:nundef} 
\end{align}

\subsubsection{Deformed configuration of the shell}
As the shell deforms into its deformed configuration $\tilde{\mathcal{V}}$, the midsurface $\mathcal{S}$ maps to the deformed midsurface $\tilde{\mathcal{S}}$~\figref{fig2}{d}, with position vector
\begin{align}
\vec{\tilde{\rho}}(s,\phi)=\tilde{r}(s)\vec{u_r}(\phi)+\tilde{z}(s)\vec{u_z}, 
\end{align}
where, in particular, $s$ is again the undeformed arclength. Denoting by $\tilde{s}$ the deformed arclength, we define the stretches
\begin{align}
&\tilde{f}_s(s)=\dfrac{\mathrm{d}\tilde{s}}{\mathrm{d}s},&& \tilde{f}_\phi(s)=\dfrac{\tilde{r}(s)}{r(s)},\label{eq:fs}
\end{align}
which enable us to define the tangent angle $\tilde{\psi}(s)$ of $\tilde{\mathcal{S}}$ by
\begin{align}
&\tilde{r}'(s)=\tilde{f}_s\cos{\tilde{\psi}(s)},&& \tilde{z}'(s)=\tilde{f}_s\sin{\tilde{\psi}(s)},  \label{eq:drz}
\end{align}
where dashes still denote differentiation with respect to $s$. Similarly to the analysis of the undeformed configuration, we introduce the tangent vectors
\begin{align}
&\vec{\tilde{e}_s}(s,\phi)=\cos{\tilde{\psi}(s)}\vec{u_r}(\phi)+\sin{\tilde{\psi}(s)}\vec{u_z},&&\vec{\tilde{e}_\phi}(\phi)=\vec{u_\phi}(\phi), 
\end{align}
and the normal vector
\begin{align}
\vec{\tilde{n}}(s,\phi)=\cos{\tilde{\psi}(s)}\vec{u_z}-\sin{\tilde{\psi}(s)}\vec{u_r}(\phi), 
\end{align}
so that $\vec{\tilde{n}}=\vec{\tilde{e}_s}\times\vec{\tilde{e}_\phi}$. This defines a (right-handed) orthonormal basis ${\tilde{\mathcal{B}}=\{\vec{\tilde{e}_s},\vec{\tilde{e}_\phi},\vec{\tilde{n}}\}}$ describing \smash{$\tilde{\mathcal{V}}$}~\figref{fig2}{e}. The curvatures of the deformed shell are
\begin{align}
&\tilde{\kappa}_s(s)=\dfrac{\tilde{\psi}'(s)}{\tilde{f}_s(s)},&& \tilde{\kappa}_\phi(s)=\dfrac{\sin{\tilde{\psi}(s)}}{\tilde{r}(s)}.\label{eq:kappas}
\end{align}
As the shell deforms, the normals to $\mathcal{S}$ need not remain normal to $\tilde{\mathcal{S}}$, and so a point in $\mathcal{V}$ at a distance $\varepsilon\zeta$ from $\mathcal{S}$ will end up, in $\tilde{\mathcal{V}}$, at a distance $\varepsilon\tilde{\zeta}$ from $\tilde{\mathcal{S}}$, and displaced by a distance $\varepsilon\tilde{\varsigma}$ parallel to $\tilde{\mathcal{S}}$~\figref{fig2}{e}. By definition of the midsurface, $\tilde{\zeta}=\tilde{\varsigma}=0$ if $\zeta=0$. The position of a point in $\tilde{\mathcal{V}}$ is thus
\begin{align}
\vec{\tilde{r}}(s,\phi,\zeta)= \vec{\tilde{\rho}}(s,\phi)+\varepsilon \tilde{\zeta}(s,\zeta)\vec{\tilde{n}}(s,\phi)+\varepsilon \tilde{\varsigma}(s,\zeta)\vec{\tilde{e}_s}(s,\phi). \label{eq:def}
\end{align}
Continuing to use commata to denote partial differentiation, we find
\begin{subequations}\label{eq:pddef}
\begin{align}
\dfrac{\partial\vec{\tilde{r}}}{\partial s}&=\left[\tilde{f}_s\left(1-\varepsilon\tilde{\kappa}_s \tilde{\zeta}\right)+\varepsilon \tilde{\varsigma}_{,s}\right]\vec{\tilde{e}_s}+\varepsilon\left(\tilde{\zeta}_{,s}+\tilde{f}_s\tilde{\kappa}_s \tilde{\varsigma}\right)\vec{\tilde{n}},\label{eq:pddefs}
\end{align}
and
\begin{align}
\dfrac{\partial\vec{\tilde{r}}}{\partial\phi}&=\left[\tilde{r}\left(1-\varepsilon\tilde{\kappa}_\phi \tilde{\zeta}\right)+\varepsilon \tilde{\varsigma}\cos{\tilde{\psi}}\right]\vec{\tilde{e}_\phi},\,\dfrac{\partial\vec{\tilde{r}}}{\partial\zeta}=\varepsilon\left(\tilde{\zeta}_{,\zeta}\vec{\tilde{n}}+\tilde{\varsigma}_{,\zeta}\vec{\tilde{e}_s}\right).
\end{align}
\end{subequations}
Noting that $\tilde{r}=\tilde{f}_\phi r$ from definitions~\eqref{eq:fs}, the Riemannian metric of $\smash{\tilde{\mathcal{V}}}$ is therefore
\begin{subequations}
\begin{align}
&\!\left\{\left[\tilde{f}_s\left(1-\varepsilon\tilde{\kappa}_s \tilde{\zeta}\right)+\varepsilon \tilde{\varsigma}_{,s}\right]^2+\varepsilon^2\left(\tilde{\zeta}_{,s}+\tilde{f}_s\tilde{\kappa}_s \tilde{\varsigma}\right)^2\right\}\mathrm{d}s^2\nonumber\\
&\;+\!\left[\tilde{f}_\phi r\left(1\!-\!\varepsilon\tilde{\kappa}_\phi \tilde{\zeta}\right)\!+\!\varepsilon \tilde{\varsigma}\cos{\tilde{\psi}}\right]^2\mathrm{d}\phi^2+\varepsilon^2\!\left[\left(\tilde{\zeta}_{,\zeta}\right)^2\!+\!\left(\tilde{\varsigma}_{,\zeta}\right)^2\right]\mathrm{d}\zeta^2\nonumber\\
&\;+2\varepsilon\Bigl\{\tilde{\varsigma}_{,\zeta}\left[\tilde{f}_s\left(1-\varepsilon\tilde{\kappa}_s \tilde{\zeta}\right)+\varepsilon \tilde{\varsigma}_{,s}\right]+\varepsilon\tilde{\zeta}_{,\zeta}\left(\tilde{\zeta}_{,s}+\tilde{f}_s\tilde{\kappa}_s \tilde{\varsigma}\right)\Bigr\}\mathrm{d}s\,\mathrm{d}\zeta. \label{eq:defmetric}
\end{align}
From $\tilde{\zeta}=\tilde{\varsigma}=0$ on $\zeta=0$, it follows that $\tilde{\zeta}_{,s}=\tilde{\varsigma}_{,s}=0$ on $\zeta=0$. Hence the metric of $\tilde{\mathcal{S}}$ is simply
\begin{align}
\tilde{f}_s^{\,2}\,\mathrm{d}s^2+\tilde{f}_\phi^{\,2}\,r^2\mathrm{d}\phi^2.\label{eq:defmidmetric}
\end{align}
\end{subequations}
At the surfaces $\tilde{\zeta}=\pm\tilde{h}^\pm(s)$ of the deformed shell, the unit tangent vectors are $\vec{\tilde{e}_s^\pm}$ and $\vec{\tilde{e}_{\smash\phi}^\pm}=\vec{\tilde{e}_\phi}$. They define the normals $\vec{\tilde{n}^\pm}\parallel\smash{\vec{\tilde{e}_{\smash{s}}^\pm}}\times\vec{\tilde{e}_{\smash{\phi}}^\pm}$~\figref{fig2}{e}.

\subsubsection{Intrinsic configuration of the shell: Incompatibility}
To specify the intrinsic configuration $\mathcal{V}^0$ of the shell, we introduce the intrinsic stretches $\smash{f_s^0,f_\phi^0}$ and the intrinsic curvatures $\smash{\kappa_s^0,\kappa_\phi^0}$ and the intrinsic normal displacement $\zeta^0$. We assume that $\smash{f_s^0,f_\phi^0}$ and $\smash{\kappa_s^0,\kappa_\phi^0}$ are functions of $s$ only, while $\zeta^0(s,\zeta)$ is strictly increasing in $\zeta$, with $\zeta^0=0$ on $\zeta=0$. Further, we assume that the analogue of the displacement parallel to the midsurface vanishes, $\varsigma^0=0$.

Although we have named these functions with reference to similar quantities defined for the deformed configuration, they lack a geometric meaning at this stage. In fact, the Riemannian metric that we can write down by analogy with Eq.~\eqref{eq:defmetric},
\begin{subequations}
\begin{align}
&\!\left\{\left[f_s^0\bigl(1-\varepsilon\kappa_s^0\zeta^0\bigr)\right]^2\!+\varepsilon^2\bigl(\zeta^0_{,s}\bigr)^2\right\}\mathrm{d}s^2+\bigl[f_\phi^0\bigl(1-\varepsilon\kappa_\phi^0\zeta^0\bigr)\bigr]^2r^2\mathrm{d}\phi^2\nonumber\\
&\qquad+\varepsilon^2\bigl(\zeta^0_{,\zeta}\bigr)^2\mathrm{d}\zeta^2+2\varepsilon^2\zeta^0_{,\zeta}\zeta^0_{,s}\,\mathrm{d}s\,\mathrm{d}\zeta,\label{eq:intmetric}
\end{align}
is not in general compatible: its Riemann curvature tensor does not vanish in general, so it cannot in general be embedded into three-dimensional Euclidean space~\cite{goriely}. Mechanically, this means that relieving all stresses in the shell requires an infinite number of cuts~\cite{goriely}. This is not surprising because, in the biological system, each cell undergoes independent shape changes or division in general and, since cells are infinitesimal in this continuum description, isolating these infinitesimal building blocks requires infinitely many cuts.

We now define the intrinsic midsurface $\mathcal{S}^0$ of the shell by its Riemannian metric, which is, by analogy with Eq.~\eqref{eq:defmidmetric} and consistently with Eq.~\eqref{eq:intmetric},
\begin{align}
\bigl(f_s^0\bigr)^2\,\mathrm{d}s^2+\bigl(f_\phi^0\bigr)^2\,r^2\mathrm{d}\phi^2. \label{eq:intmidmetric}
\end{align}
\end{subequations}
It follows from a local embedding theorem for Riemannian metrics~\cite{janet26,cartan27} that this two-dimensional metric can be embedded, at least locally, into three-dimensional Euclidean space. In particular, this means that there exists a local (right-handed) orthonormal intrinsic basis $\mathcal{B}^0=\{\vec{E_s},\vec{E_\phi},\vec{N}\}$ of three-dimensional space such that $\vec{E_s},\vec{E_\phi}=\vec{u_\phi}$ are tangent to $\mathcal{S}^0$, and $\vec{N}$ is normal to it~\figref{fig2}{f}. Because the metric~\eqref{eq:intmetric} is incompatible, the curvatures $\smash{\varkappa_s^0,\varkappa_\phi^0}$ of $\mathcal{S}^0$ are in general different from the intrinsic curvatures $\smash{\kappa_s^0,\kappa_\phi^0}$. While Eq.~\eqref{eq:intmidmetric} assigns a geometric meaning to the intrinsic stretches $\smash{f_s^0,f_\phi^0}$, these intrinsic curvatures therefore remain without the direct geometric realisation that would result from an embedding into three-dimensional Euclidean space, as does the intrinsic normal displacement $\zeta^0$. 

We specify the latter by requiring the intrinsic deformations to conserve volume. This assumption is, for example, appropriate for \emph{Volvox} inversion~\figref{fig1}{b}: the cell measurements of Ref.~\cite{hohn11} suggest that the cell shape changes driving inversion preserve volume. For other developmental processes that include cell division, the assumption of intrinsic volume conservation would be replaced with a position-dependent constraint that takes account of this growth. Since $\zeta^0(s,\zeta)$ is increasing and can hence be inverted to yield $\smash{\zeta\bigl(s,\zeta^0\bigr)}$, Eq.~\eqref{eq:intmetric} becomes, on changing coordinates from $\{s,\phi,\zeta\}$ to $\{s,\phi,\zeta^0\}$,\begin{subequations}
\begin{align}
\bigl(\chi_s^0\bigr)^2\mathrm{d}s^2+\bigl(\chi_\phi^0\bigr)^2\mathrm{d}\phi^2+\bigl(\chi_{\zeta^0}^0\bigr)^2\bigl(\mathrm{d}\zeta^0\bigr)^2
\end{align}
with scale factors
\begin{align}
&\chi_s^0=f_s^0\bigl(1-\varepsilon\kappa_s^0\zeta^0\bigr),&&\chi_\phi^0=f_\phi^0r\bigl(1-\varepsilon\kappa_\phi^0\zeta^0\bigr),&&\chi_{\zeta^0}^0=\varepsilon.
\end{align}
Its volume element is therefore
\begin{align}
\mathrm{d}V^0&=\chi_s^0\chi_\phi^0\chi_{\zeta^0}^0\,\mathrm{d}s\,\mathrm{d}\phi\,\mathrm{d}\zeta^0\nonumber\\
&=\varepsilon f_s^0f_\phi^0\bigl(1-\varepsilon\kappa_s^0\zeta^0\bigr)\bigl(1-\varepsilon\kappa_\phi^0\zeta^0\bigr)\,r\, \mathrm{d}s\,\mathrm{d}\phi\,\mathrm{d}\zeta^0.\label{eq:dVint}
\end{align}
\end{subequations}
Intrinsic volume conservation requires $\mathrm{d}V=\mathrm{d}V^0$, so Eqs.~\eqref{eq:dVundef} and \eqref{eq:dVint} combine to yield a differential equation for $\zeta^0$ as a function of $\zeta$, which we will eventually integrate in subsection~\ref{sec:modelb} under the scaling assumptions of our shell theory.

At this stage, $\mathcal{S}$, $\tilde{\mathcal{S}}$, and $\mathcal{S}^0$ are defined to be corresponding surfaces within the shell. Indeed, it would it be possible to develop a shell theory for any choice of surfaces that correspond to each other in this way. At this point however, we make a particular choice of surfaces (that we shall refer to as midsurfaces) by imposing the following condition: the shell surfaces, at $\zeta=\pm h^\pm(s)$ and $\tilde\zeta=\pm\tilde{h}^\pm(s)$ in the deformed and undeformed configurations respectively, correspond to $\zeta^0=\pm h^0(s)/2$; the calculations in subsection~\ref{sec:modelb} will show that this choice can be made. We stress that, like $\zeta^0$, the intrinsic thickness $h^0(s)$ lacks a direct geometric realisation.

We close by noting that $\zeta^0(s,\zeta)$ and hence $h^0(s)$ can also be specified without reference to the incompatible metric~\eqref{eq:intmetric}, by imposing the condition $\det{\tens{F^0}}=1$. Indeed, with the intrinsic deformation gradient $\tens{F^0}$ defined as in Eq.~\eqref{eq:F0} below, this is easily seen to be equivalent with $\mathrm{d}V=\mathrm{d}V^0$. Conversely, the condition $\det{\tens{F^0}}=1$ can be used to define the intrinsic volume element $\mathrm{d}V^0$ without reference to Eq.~\eqref{eq:intmetric}.

\subsubsection{Calculation of the deformation gradient tensors}
The geometric deformation gradient is $\btilde{\tens{F}}=\operatorname{Grad}\vec{\tilde{r}}$~\cite{goriely}, in which the gradient with respect to the undeformed configuration is~\cite{goriely}
\begin{align}
\operatorname{Grad}{\cdot}=\dfrac{1}{\chi_s^2}\dfrac{\partial\cdot}{\partial s}\otimes\dfrac{\partial\vec{r}}{\partial s}+\dfrac{1}{\chi_\phi^2}\dfrac{\partial\cdot}{\partial\phi}\otimes\dfrac{\partial\vec{r}}{\partial\phi}+\dfrac{1}{\chi_\zeta^2}\dfrac{\partial\cdot}{\partial\zeta}\otimes\dfrac{\partial\vec{r}}{\partial\zeta}. \label{eq:gradundef}
\end{align}
Combining Eqs.~\eqref{eq:pdundef}, \eqref{eq:sfundef}, and \eqref{eq:pddef}, we thus obtain the geometric deformation gradient, 
\begin{widetext}
\begin{align}
\btilde{\tens{F}}&=\left(\begin{array}{ccc}
\dfrac{\tilde{f}_s\left(1-\varepsilon\tilde{\kappa}_s \tilde{\zeta}\right)+\varepsilon \tilde{\varsigma}_{,s}}{1-\varepsilon\varkappa_s\zeta}&0&\tilde{\varsigma}_{,\zeta}\\
0&\dfrac{\tilde{f}_\phi\left(1-\varepsilon\tilde{\kappa}_\phi \tilde{\zeta}\right)+\varepsilon \tilde{\varsigma}\cos{\tilde{\psi}}/r}{1-\varepsilon\varkappa_\phi\zeta}&0\\
\dfrac{\varepsilon\left(\tilde{\zeta}_{,s}+\tilde{f}_s\tilde{\kappa}_s \tilde{\varsigma}\right)}{1-\varepsilon\varkappa_s\zeta}&0&\tilde{\zeta}_{,\zeta}
\end{array}\right),\label{eq:Ftilde}
\end{align}
expressed here with respect to the mixed basis $\tilde{\mathcal{B}}\otimes\mathcal{B}$. We now complete specifying the intrinsic configuration $\mathcal{V}^0$ by writing down an analogous expression for the intrinsic deformation gradient with respect to the mixed basis $\mathcal{B}^0\otimes\mathcal{B}$, viz.
\begin{align}
&\tens{F^0}=\left(\begin{array}{ccc}
\dfrac{f_s^0\left(1-\varepsilon\kappa_s^0\zeta^0\right)}{1-\varepsilon\varkappa_s\zeta}&0&0\\
0&\dfrac{f_\phi^0\bigl(1-\varepsilon\kappa_\phi^0\zeta^0\bigr)}{1-\varepsilon\varkappa_\phi\zeta}&0\\
\dfrac{\varepsilon \zeta^0_{,s}}{1-\varepsilon\varkappa_s\zeta}&0&\zeta^0_{,\zeta}
\end{array}\right).  \label{eq:F0}
\end{align}
The elastic deformation gradient is therefore, with respect to the natural mixed basis $\tilde{\mathcal{B}}\otimes\mathcal{B}^0$,
\begin{align}
\tens{F}=\btilde{\tens{F}}\bigl(\tens{F^0}\bigr)^{-1}=\left(\begin{array}{ccc}
\dfrac{\tilde{f}_s\left(1-\varepsilon\tilde{\kappa}_s \tilde{\zeta}\right)+\varepsilon\bigl(\tilde{\varsigma}_{,s}+\tilde{\varsigma}_{,\zeta^0}\zeta^0_{,s}\bigr)}{f_s^0\left(1-\varepsilon\kappa_s^0\zeta^0\right)}&0&\tilde{\varsigma}_{,\zeta^0}\\
0&\dfrac{\tilde{f}_\phi\left(1-\varepsilon\tilde{\kappa}_\phi \tilde{\zeta}\right)+\varepsilon\tilde{\varsigma}\cos{\tilde{\psi}}/r}{f_\phi^0\bigl(1-\varepsilon\kappa_\phi^0\zeta^0\bigr)}&0\\
\dfrac{\varepsilon\bigl(\tilde{\zeta}_{,s}+\tilde{f}_s\tilde{\kappa}_s\tilde{\varsigma}-\zeta^0_{,s}\tilde{\zeta}_{,\zeta^0}\bigr)}{f_s^0\left(1-\varepsilon\kappa_s^0\zeta^0\right)}&0&\tilde{\zeta}_{,\zeta^0}
\end{array}\right).\label{eq:F}
\end{align}
\end{widetext}

\subsection{Thin shell theory for large bending deformations}\label{sec:modelb}
In this subsection, we derive the effective elastic energy for the shell by asymptotic expansion of three-dimensional elasticity. We assume the simplest constitutive law, that the shell is made of an incompressible neo-Hookean material~\cite{goriely}, so that its elastic energy is
\begin{align}
\mathcal{E}=\int_{\mathcal{V}^0}{e\,\mathrm{d}V^0},\quad\text{with }e=\dfrac{C}{2}(\mathcal{I}_1-3), \label{eq:E}
\end{align}
wherein $C>0$ is a material parameter, and $\mathcal{I}_1$ is the first invariant of the right Cauchy--Green tensor $\tens{C}=\tens{F}^\top\tens{F}$~\cite{goriely,* [] [{ Chap.~1, pp.~1--72, Chap.~2.2, pp.~83--121, Chap.~3.4, pp.~152--155, and Chap.~6.1, pp.~328--351.}] ogden}. The integration of the strain energy density $e$ is over the intrinsic configuration $\mathcal{V}^0$ of the shell, with volume element $\mathrm{d}V^0$. As we have noted above, this can be defined from the condition $\det{\tens{F^0}}=1$, independently of the incompatible metric~\eqref{eq:intmetric}.

The force on a area element $\mathrm{d}\tilde{S}$ with unit normal $\vec{\tilde{m}}$ of the deformed configuration is $\tens{T}\vec{\tilde{m}}\,\mathrm{d}\tilde{S}$~\cite{goriely,ogden}. In this expression, $\tens{T}$ is the Cauchy stress tensor, which, for this neo--Hookean material, is related to the deformation gradient by~\cite{dervaux09}
\begin{align}
\tens{T}=C\left(\tens{F}\tens{F}^\top-p\tens{I}\right), 
\end{align}
in which $\tens{I}$ is the identity, and the Lagrange multiplier $p$ is proportional to pressure and imposes the incompressibility condition $\det{\tens{F}}=1$. To this area element of the deformed configurations corresponds, in the undeformed configuration, an area element $\mathrm{d}S$ with unit normal $\vec{m}$. Nanson's relation~\cite{goriely,ogden} states that $\vec{\tilde{m}}\,\mathrm{d}\tilde{S}=\tilde{\mathcal{J}}\btilde{\tens{F}}^{-\top}\vec{m}\,\mathrm{d}S$, where ${\tilde{\mathcal{J}}=\det{\btilde{\tens{F}}}=\det{\tens{F}}\det{\tens{F^0}}=1}$. We introduce the tensor
\begin{align}
\tens{P}=\tens{T}\btilde{\tens{F}}^{-\top}=C\tens{Q}\quad\text{with }\tens{Q}=\tens{F}\bigl(\tens{F^0}\bigr)^{-\top}-p\btilde{\tens{F}}^{-\top}.\label{eq:P}
\end{align}
In particular, if $\tens{F^0}=\tens{I}$, then $\tens{P}=\tens{T}\tens{F}^{-\top}$ is the familiar (first) Piola--Kirchhoff tensor~\cite{goriely}. By definition, ${\tens{T}\vec{\tilde{m}}\,\mathrm{d}\tilde{S}=\tens{P}\vec{m}\,\mathrm{d}S}$, and hence, similarly to the derivation of the familiar Cauchy equation of classical elasticity~\cite{goriely,ogden}, the configuration of the shell minimising the energy~\eqref{eq:E} is determined by
\begin{subequations}
\begin{align}
\operatorname{Div}{\tens{Q}^\top}=\vec{0},
\end{align}
where the divergence (with respect to the undeformed configuration of the shell) is defined by contracting the first and last indices of the gradient in Eq.~\eqref{eq:gradundef}. Since $\mathcal{B}$ is independent of $\zeta$ by definition, and using the nabla operator to denote the gradient on $\mathcal{S}$, this becomes, on separating the components parallel and perpendicular to the midsurface,
\begin{align}
\dfrac{(\tens{Q}\vec{n})_{,\zeta}}{\varepsilon}+\vec{\nabla}\cdot\tens{Q}^\top=\vec{0}.\label{eq:Cauchy}
\end{align}
\end{subequations}

\subsubsection{Scaling assumptions}
At this point, we break the complete generality of our description by making scalings assumptions appropriate for a shell theory of large intrinsic bending deformations.

First, we introduce large intrinsic bending deformations explicitly by scaling the intrinsic curvatures so as to allow small radii of curvature in the meridional direction, viz.
\begin{align}
&\kappa_s^0=f_s^0f_\phi^0\dfrac{\lambda_s^0}{\varepsilon},&& \kappa_\phi^0=f_s^0f_\phi^0\,\lambda_\phi^0,\label{eq:curvscalings}
\end{align}
in which the scaled intrinsic curvatures $\lambda_s^0,\lambda_{\smash\phi}^0$ are assumed to be $O(1)$ quantities. This scaling regime in which the meridional intrinsic radius of curvature becomes comparable to the thickness of the cell sheet is, as shown in \textfigref{fig1}{b}, the one relevant for \emph{Volvox} invagination, which we shall analyse in Section~\ref{sec:volvox}. Appendix~\ref{appA} treats the general case in which all components of the curvature tensor are allowed to be large.

Second, we make the standard scaling assumptions of shell theory, that the elastic strains are small, i.e. that the stretches and curvatures in the deformed configuration do not differ ``too much'' from the intrinsic stretches and curvatures. In particular, while we have allowed the radius of curvature $1/\kappa_s^0$ to become comparable to the shell thickness in Eqs.~\eqref{eq:curvscalings}, we shall assume the deviations from this to remain small. More formally, we introduce the shell strains $E_s,E_\phi$ by writing
\begin{align}
\tilde{f}_s&=f_s^0\bigl(1+\varepsilon E_s\bigr), &\tilde{f}_\phi&=f_\phi^0\left(1+\varepsilon E_\phi\right),\label{eq:shellstr}
\end{align}
and the curvature strains $L_s,L_\phi$ by letting
\begin{align}
&\tilde{\kappa}_s=f_s^0f_\phi^0\left(\dfrac{\lambda_s^0}{\varepsilon}+L_s\right),&& \tilde{\kappa}_\phi=f_s^0f_\phi^0\left(\lambda_\phi^0+L_\phi\right).\label{eq:curvstr}
\end{align}
Finally, we introduce the scaled variables
\begin{align}
&Z^0=f_s^0f_\phi^0\,\zeta^0,&& Z=f_s^0f_\phi^0\,\tilde{\zeta},&& S=f_s^0f_\phi^0\,\tilde{\varsigma},\label{eq:scaleds}
\end{align}
While we will come back to discussing the factors $\smash{f_s^0f_{\smash\phi}^0}$ that arise in Eqs.~\eqref{eq:curvscalings}, \eqref{eq:curvstr}, and \eqref{eq:scaleds}, we note, for now and from Eq.~\eqref{eq:intmidmetric}, the following: the intrinsic midsurface $\mathcal{S}^0$ has surface element $\smash{\mathrm{d}S^0=f_s^0f_\phi^0\,r\,\mathrm{d}r\,\mathrm{d}\phi=f_s^0f_\phi^0\,\mathrm{d}S}$, with $\mathrm{d}S$ the surface element of the undeformed midsurface $\mathcal{S}$. Hence these rescalings by $\smash{f_s^0f_{\smash\phi}^0}$ absorb the intrinsic stretching of the midsurface. This will turn out to simplify expressions that arise in subsequent calculations.

\subsubsection{Boundary and incompressibility conditions}
We solve the Cauchy equation~\eqref{eq:Cauchy} subject to the incompressibility condition $\det{\tens{F}}=1$ and force-free boundary conditions. These boundary conditions, that there be no external forces on the surfaces of the shell, are relevant for many problems in developmental biology, where deformations are, as discussed in the introduction, driven by changes of the intrinsic geometry only; including external forces does not pose any additional difficulty, though. 

These force-free boundary conditions read $\smash{\tens{T}^\bpm\vec{\tilde{n}^\pm}=\vec{0}}$~\cite{goriely}, where $\tens{T}^\bpm$ are the Cauchy tensors evaluated on the surfaces ${\tilde{\zeta}=\pm\tilde{h}^\pm}$ of $\tilde{\mathcal{V}}$. By the above, these are equivalent with ${\tens{P}^\bpm\vec{n^\pm}=\vec{0}}$, where, from Eq.~\eqref{eq:P}, $\tens{P}^\bpm=C\tens{Q}^\bpm$ are evaluated on the surfaces $\zeta=\pm h^\pm$ of $\mathcal{V}$, the normal vectors $\vec{n^\pm}$ of which are given by Eq.~\eqref{eq:nundef}. The latter yields the expansion\begin{align}
\vec{n^\pm}&=\vec{n}\mp \varepsilon h^\pm_{,s}\vec{e_s}+O\bigl(\varepsilon^2\bigr). \label{eq:Npm}
\end{align}

The incompressibility condition is ${\det{\tens{F}}=1}$. Since the bases $\tilde{\mathcal{B}}$ and $\mathcal{B}^0$ are orthonormal, there exist rotations, represented by proper orthogonal matrices $\tilde{\mat{R}}$ and $\mat{R^0}$, that map the standard Cartesian basis $\mathcal{X}$ onto $\tilde{\mathcal{B}}$ and $\mathcal{B}^0$, respectively. Hence, if $\mat{F}$ denotes the matrix in Eq.~\eqref{eq:F} that represents $\tens{F}$ with respect to the mixed basis ${\tilde{\mathcal{B}}\otimes\mathcal{B}^0}$, then $\tens{F}$ is represented by $\smash{\tilde{\mat{R}}^\top\mat{F}\mat{R^0}}$ with respect to $\mathcal{X}\otimes\mathcal{X}$. Since $\det{\tilde{\mat{R}}}=\det{\mat{R^0}}=1$, $\det{\tens{F}}=\det{\bigl(\smash{\tilde{\mat{R}}^\top\mat{F}\mat{R^0}}\bigr)}=\det{\mat{F}}$. The incompressibility condition can therefore be evaluated using the matrix in Eq.~\eqref{eq:F}, but it is important to recognise that incompressibility is a tensorial condition. For the general, non-axisymmetric deformations discussed in Appendix~\ref{appA}, we shall indeed have to distinguish more carefully between tensors and the matrices representing them with respect to mixed non-orthogonal bases, which is why we have already introduced different notations, based on Ogden's~\cite{ogden}, for matrices (sans serif font) and tensors (bold sans serif font) that could be used interchangeably here.

\subsubsection{Intrinsic volume conservation}
Before expanding the boundary and incompressibility conditions asymptotically, we determine the dependence of $\zeta^0$ and hence $Z^0$ on $\zeta$ that results from the condition $\mathrm{d}V=\mathrm{d}V^0$ of intrinsic volume conservation. On recalling that $\kappa_s^0=O\bigl(\varepsilon^{-1}\bigr)$, the expressions for $\mathrm{d}V$ in Eq.~\eqref{eq:dVundef} and $\smash{\mathrm{d}V^0}$ in Eq.~\eqref{eq:dVint} yield, at leading order, a differential equation for $Z^0(\zeta)$,
\begin{align}
\left(1-\lambda_s^0Z^0\right)Z^0_{,\zeta}=1\;\Longrightarrow\;Z^0=\dfrac{1}{\lambda_s^0}\left(1-\sqrt{1-2\lambda_s^0\zeta}\right),\label{eq:Z0zeta}
\end{align}
where we have imposed $Z^0=0$ at $\zeta=0$. Let \smash{$H^0=h^0f_s^0f_\phi^0$}. Since $\zeta^0=\pm h^0/2\Longleftrightarrow Z^0=\pm H^0/2$ at $\zeta=\pm h^\pm$ by definition, Eq.~\eqref{eq:Z0zeta} implies
\begin{align}
h^\pm&=\dfrac{H^0}{2}\left(1\mp\dfrac{\lambda_s^0}{4}H^0\right)\quad\Longrightarrow\quad h=h^++h^-=H^0,\label{eq:h}
\end{align}
wherein $h$ is again the undeformed thickness of the cell sheet~\figref{fig2}{c}. We note that Eq.~\eqref{eq:h} is a leading-order result only, since we have ignored $O(\varepsilon)$ corrections in Eq.~\eqref{eq:Z0zeta}.

\subsubsection{Expansion of the boundary and incompressibility conditions}
To expand the incompressibility and boundary conditions in the small parameter $\varepsilon$, we posit regular expansions
\begin{align}
Z&=Z_{(0)}+\varepsilon Z_{(1)}+O\bigl(\varepsilon^2\bigr),&S&=S_{(0)}+O(\varepsilon),\label{eq:expansions}
\end{align}
for the scaled transverse and parallel displacements. Throughout this paper, we shall use subscripts in parentheses in this way to denote the different terms in asymptotic expansions in~$\varepsilon$. We further expand
\begin{align}
&\tens{Q}=\tens{Q_{(0)}}+\varepsilon\tens{Q_{(1)}}+O\bigl(\varepsilon^2\bigr), &p&=p_{(0)}+O(\varepsilon).\label{eq:Qpexpansions}
\end{align}
\paragraph{Expansion at order $O(1)$.} At leading order, Eq.~\eqref{eq:Cauchy} yields $(\tens{Q_{(0)}}\vec{n})_{,\zeta}=\vec{0}$, so $\tens{Q_{(0)}}\vec{n}=\vec{Q}(s)$ is independent of $\zeta$. It follows that $\vec{0}=\tens{Q}^\bpm\vec{n^\pm}=\tens{Q}^\bpm_{\smash{\tens{(0)}}}\vec{n}+O(\varepsilon)=\pm\vec{Q}+O(\varepsilon)$ using Eq.~\eqref{eq:Npm}. Thus $\vec{0}\equiv\vec{Q}=\tens{Q_{(0)}}\vec{n}=\smash{\bigl(q^s_{\smash{(0)}},0,q^n_{\smash{(0)}}\bigr)}$. Expanding definition~\eqref{eq:P} using Eqs.~\eqref{eq:Ftilde}, \eqref{eq:F0}, and \eqref{eq:F}, this yields~\footnote{Expansions were carried out using \textsc{Mathematica} (Wolfram, Inc.) to assist with manipulating the complicated algebraic expressions that arise in these calculations.}

\begin{widetext}
\begin{subequations}\label{eq:G0}
\begin{align}
0=q^s_{(0)}&=f_s^0f_\phi^0\bigl(1-\lambda_s^0Z^0\bigr)\dfrac{\lambda_s^0S_{(0)}\left[p_{(0)}-\bigl(S_{(0),Z^0}\bigr)^2\right]+\left(1-\lambda_s^0Z_{(0)}\right)S_{(0),Z^0}Z_{(0),Z^0}}{\left(1-\lambda_s^0Z_{(0)}\right)Z_{(0),Z^0}-\lambda_s^0S_{(0)}S_{(0),Z^0}},\\
0=q^n_{(0)}&=f_s^0f_\phi^0\bigl(1-\lambda_s^0Z^0\bigr)\dfrac{\left(1-\lambda_s^0Z_{(0)}\right)\left[\bigl(Z_{(0),Z^0}\bigr)^2-p_{(0)}\right]-\lambda_s^0S_{(0)}S_{(0),Z^0}Z_{(0),Z^0}}{\left(1-\lambda_s^0Z_{(0)}\right)Z_{(0),Z^0}-\lambda_s^0S_{(0)}S_{(0),Z^0}},
\end{align}
\end{subequations}
where we have used $\bigl({\zeta^0}_{,\zeta}\bigr)^{-1}=f_s^0f_\phi^0\left(1-\lambda_s^0Z^0\right)+O(\varepsilon)$, which follows from Eq.~\eqref{eq:Z0zeta} on recalling the rescalings~\eqref{eq:scaleds}. Moreover, on expanding the incompressibility condition using Eq.~\eqref{eq:F}, we find
\begin{align}
1=\det{\mat{F}}=1-\dfrac{1-\lambda_s^0Z^0-\left(1-\lambda_s^0Z_{(0)}\right)Z_{(0),Z^0}+\lambda_s^0S_{(0)}S_{(0),Z^0}}{1-\lambda_s^0Z^0}+O(\varepsilon). \label{eq:det0}
\end{align}
\end{widetext}
Eqs.~\eqref{eq:G0} and~\eqref{eq:det0} define a system of three simultaneous linear algebraic equations for $p_{(0)}$, $Z_{(0),Z^0}$, and $S_{(0),Z^0}$, with solution\begin{subequations}
\begin{align}
p_{(0)}&=\dfrac{\left(1-\lambda_s^0Z^0\right)^2}{\left(1-\lambda_s^0Z_{(0)}\right)^2+\left(\lambda_s^0 S_{(0)}\right)^2},\label{eq:P0}\\
Z_{(0),Z^0}&=\dfrac{\left(1-\lambda_s^0Z^0\right)\left(1-\lambda_s^0Z_{(0)}\right)}{\left(1-\lambda_s^0Z_{(0)}\right)^2+\left(\lambda_s^0 S_{(0)}\right)^2},\label{eq:Z0}\\
S_{(0),Z^0}&=-\dfrac{\lambda_s^0 S_{(0)}\left(1-\lambda_s^0Z^0\right)}{\left(1-\lambda_s^0Z_{(0)}\right)^2+\left(\lambda_s^0 S_{(0)}\right)^2}.\label{eq:S0}
\end{align}
\end{subequations}
Eq.~\eqref{eq:det0} or Eqs.~\eqref{eq:Z0} and \eqref{eq:S0} imply
\begin{subequations}
\begin{align}
&-2Z_{(0),Z^0}\bigl(1-\lambda_s^0Z_{(0)}\bigr)+2\lambda_s^0S_{(0)}S_{(0),Z^0}=-2\bigl(1-\lambda_s^0Z^0\bigr).
\end{align}
Integrating and using the fact that $Z_{(0)}=S_{(0)}=0$ at $Z^0=0$ by definition of the midsurfaces, we obtain
\begin{align}
\bigl(1-\lambda_s^0Z_{(0)}\bigr)^2+\bigl(\lambda_s^0S_{(0)}\bigr)^2=\bigl(1-\lambda_s^0Z^0\bigr)^2.\label{eq:int}
\end{align}
\end{subequations}
Eq.~\eqref{eq:P0} now becomes $p_{(0)}=1$. Moreover, on substituting Eq.~\eqref{eq:int} into Eq.~\eqref{eq:Z0},
\begin{align}
\dfrac{\partial Z_{(0)}}{\partial Z^0}=\dfrac{1-\lambda_s^0 Z_{(0)}}{1-\lambda_s^0Z^0}\quad\Longrightarrow\quad\dfrac{1-\lambda_s^0Z_{(0)}}{1-\lambda_s^0 Z^0}=\text{const.}, 
\end{align}
which, using $Z_{(0)}=0$ at $Z^0=0$ again, yields $Z_{(0)}\equiv Z^0$. Hence $S_{(0)}\equiv 0$ from Eq.~\eqref{eq:int}. The last equality is the Kirchhoff ``hypothesis''~\cite{audoly}: normals to the intrinsic midsurface remain, at lowest order, normal to the deformed midsurface.
\paragraph{Expansion at order $O(\varepsilon)$.} We now expand the incompressibility condition further, finding
\begin{align}
0&=\det{\mat{F}}-1\nonumber\\
&=\varepsilon\left(E_s+E_\phi-L_\phi Z^0+\dfrac{\partial Z_{(1)}}{\partial Z^0}-\dfrac{L_s Z^0+\lambda_s^0Z_{(1)}}{1-\lambda_s^0Z^0}\right)+O\bigl(\varepsilon^2\bigr). \label{eq:det1}
\end{align}
On solving the resulting differential equation for $Z_{(1)}$ by imposing $Z_{(1)}=0$ at $Z^0=0$, we obtain
\begin{widetext}

\vspace{-5mm}
\begin{align}
Z_{(1)}=-\dfrac{Z^0\left\{6(E_s+E_\phi)-3Z^0\left[L_s+L_\phi+\lambda_s^0(E_s+E_\phi)\right]+2\lambda_s^0L_\phi\left(Z^0\right)^2\right\}}{6\left(1-\lambda_s^0Z^0\right)}.\label{eq:Z1}
\end{align}
\paragraph{Expansion at order $O\bigl(\varepsilon^2\bigr)$.} It will turn out not to be necessary to expand the deformation gradient explicitly beyond order~$O(\varepsilon)$. Indeed, it will suffice to consider a formal expansion,
\begin{align}
\tens{F}=\left(\begin{array}{ccc}
1+\varepsilon a_{(1)}+\varepsilon^2 a_{(2)}+O\bigl(\varepsilon^3\bigr)&0&\varepsilon v_{(1)}+O\bigl(\varepsilon^2\bigr)\\
0&1+\varepsilon b_{(1)}+\varepsilon^2 b_{(2)}+O\bigl(\varepsilon^3\bigr)&0\\
\varepsilon w_{(1)}+O\bigl(\varepsilon^2\bigr)&0&1+\varepsilon c_{(1)}+\varepsilon^2 c_{(2)}+O\bigl(\varepsilon^3\bigr)
\end{array}\right),\label{eq:F2}
\end{align}
with the leading-order terms found from Eq.~\eqref{eq:F}. This also yields, using Eq.~\eqref{eq:Z1},
\begin{align}
a_{(1)}&=\dfrac{6 E_s-6\left[L_s+\lambda_s^0\left(E_s-E_\phi\right)\right]Z^0+3\lambda_s^0\left[L_s-L_\phi+\lambda_s^0\left(E_s-E_\phi\right)\right]\left(Z^0\right)^2+2\left(\lambda_s^0\right)^2L_\phi\left(Z^0\right)^3}{6\left(1-\lambda_s^0Z^0\right)^2},&b_{(1)}&=E_\phi-Z^0 L_\phi.\label{eq:a1b1}
\end{align}
Expressions for $a_{(2)},b_{(2)},c_{(1)},c_{(2)},v_{(1)},w_{(1)}$ could similarly be obtained in terms of the expansions~\eqref{eq:expansions}, but, as announced, will turn out to be of no consequence. Using Eq.~\eqref{eq:F2}, the incompressibility condition becomes
\begin{align}
1=\det{\mat{F}}=1+\varepsilon\left(a_{(1)}+b_{(1)}+c_{(1)}\right)+\varepsilon^2\left(a_{(2)}+b_{(2)}+c_{(2)}+a_{(1)}b_{(1)}+b_{(1)}c_{(1)}+c_{(1)}a_{(1)}-v_{(1)}w_{(1)}\right)+O\bigl(\varepsilon^3\bigr).\label{eq:detFexp}
\end{align}
Next, using Eq.~\eqref{eq:F0}, we introduce an analogous formal expansion for the intrinsic deformation gradient, viz. 
\begin{align}
\tens{F^0}=\left(\begin{array}{ccc}
a^0_{(0)}+O(\varepsilon)&0&0\\
0&b^0_{(0)}+O(\varepsilon)&0\\
\varepsilon w^0_{(1)}+O\bigl(\varepsilon^2\bigr)&0&c^0_{(0)}+O(\varepsilon)
\end{array}\right),
\end{align}
where $c^0_{(0)}=\smash{\bigl[f_s^0f_\phi^0\bigl(1-\lambda_s^0Z^0\bigr)\bigr]^{-1}}$ using Eq.~\eqref{eq:Z0zeta}, and the values of $a^0_{(0)},b^0_{(0)},w^0_{(1)}$ are of no consequence. Hence, using Eq.~\eqref{eq:F2},
\begin{align}
\btilde{\tens{F}}=\tens{FF^0}=\left(\begin{array}{ccc}
a^0_{(0)}+O(\varepsilon)&0&\varepsilon c^0_{(0)}v_1+O\bigl(\varepsilon^2\bigr)\\
0&b^0_{(0)}+O(\varepsilon)&0\\
\varepsilon\bigl(w^0_{\smash{(1)}}+a^0_{\smash{(0)}}w_{(1)}\bigr)+O\bigl(\varepsilon^2\bigr)&0&c^0_{(0)}+O(\varepsilon).
\end{array}\right),
\end{align}
and thus, since $p=1+O(\varepsilon)$,
\begin{align}
\tens{Q}=\left(\begin{array}{ccc}
O(\varepsilon)&0&\varepsilon(v_{(1)}+w_{(1)})/c^0_{(0)}+O\bigl(\varepsilon^2\bigr)\\
0&O(\varepsilon)&0\\
O(\varepsilon)&0&O(\varepsilon). 
\end{array}\right)\quad\Longrightarrow\quad\tens{Q_{(0)}}=\tens{O},\;\tens{Q_{(1)}}\vec{n}=\left(\begin{array}{c}(v_{(1)}+w_{(1)})/c^0_{(0)}\\0\\O(1)\end{array}\right).
\end{align}
In particular, Eq.~\eqref{eq:Cauchy} at order $O(1)$ is just \smash{$(\tens{Q_{(1)}}\vec{n})_{,\zeta}=\vec{0}$}. Moreover $\vec{0}=\tens{Q^\bpm}\vec{n^\pm}=\varepsilon\tens{Q_{\smash{\tens{(1)}}}^\bpm}\vec{n}+O\bigl(\varepsilon^2\bigr)$, since $\tens{Q_{(0)}}=\tens{O}$ and using Eq.~\eqref{eq:Npm}. Similarly to above, this implies $\tens{Q_{(1)}}\vec{n}\equiv\vec{0}$. From this and from Eq.~\eqref{eq:detFexp}, we infer
\begin{align}
w_{(1)}&=-v_{(1)},&c_{(1)}&=-\left(a_{(1)}+b_{(1)}\right),&c_{(2)}&=a_{(1)}^2+a_{(1)}b_{(1)}+b_{(1)}^2-a_{(2)}-b_{(2)}+v_{(1)}w_{(1)}.\label{eq:BCres}
\end{align}
\subsubsection{Asymptotic expansion of the constitutive relations}
On computing the expansion of $\tens{C}=\tens{F}^\top\tens{F}$ from Eq.~\eqref{eq:F2} and hence that of $\mathcal{I}_1=\tr{\tens{C}}$, and simplifying using Eqs.~\eqref{eq:BCres}, we obtain\begin{subequations}
\begin{align}
\mathcal{I}_1&=3+\varepsilon\bigl[2\bigl(a_{(1)}+b_{(1)}+c_{(1)}\bigr)\bigr]+\varepsilon^2\left[a_{\smash{(1)}}^2+b_{\smash{(1)}}^2+c_{\smash{(1)}}^2+v_{\smash{(1)}}^2+w_{\smash{(1)}}^2+2\left(a_{(2)}+b_{(2)}+c_{(2)}\right)\right]+O\bigl(\varepsilon^3\bigr)\nonumber\\
&=3+\varepsilon^2\bigl[4\bigl(a_{\smash{(1)}}^2+a_{(1)}b_{(1)}+b_{\smash{(1)}}^2\bigr)\bigr]+O\bigl(\varepsilon^3\bigr).\label{eq:I1}
\end{align}
Hence, from Eqs.~\eqref{eq:a1b1} and on introducing $x=\lambda_s^0Z^0$,
\begin{align}
\mathcal{I}_1&=3+\dfrac{\varepsilon^2}{(1-x)^4}\Biggl\{\left[1+\left(1-x\right)^2\right]^2E_s^2+2\left[1+(1-x)^2\right]E_sE_\phi+\left(4-12x+18x^2-12x^3+3x^4\right)E_\phi^2\nonumber\\
&\hspace{24mm}-\dfrac{1}{\lambda_s^0}\left[2x\left(4-6x+4x^2-x^3\right)E_sL_s-2x(2-x)E_\phi L_s-\dfrac{2x}{3}\left(6-12x+11x^2-5x^3+x^4\right) E_sL_\phi\right.\nonumber\\
&\hspace{34mm}\left.-\dfrac{2x}{3}\left(12-39x+55x^2-36x^3+9x^4\right)E_\phi L_\phi\right]\nonumber\\
&\hspace{24mm}\left.+\dfrac{1}{\left(\lambda_s^0\right)^2}\left[x^2(2-x)^2L_s^2+\dfrac{2x^2}{3}\left(6-9x+5x^2-x^3\right)L_sL_\phi+\dfrac{x^2}{9}\left(36-126x+177x^2-114x^3+28x^4\right)L_\phi^2\right]\right\}\nonumber\\
&\hspace{12mm}+O\bigl(\varepsilon^3\bigr).\label{eq:I2}
\end{align}
\end{subequations}
\end{widetext}
This determines the leading-order term in the asymptotic expansion of the energy density in Eq.~\eqref{eq:E}. On defining, from Eq.~\eqref{eq:F2}, the (symmetric) effective two-dimensional deformation gradient and associated two-dimensional strain,
\begin{align}
\bhat{\tens{F}}&=\left(\begin{array}{cc}
1+\varepsilon a_{(1)}&0\\
0&1+\varepsilon b_{(1)}
\end{array}\right)+O\bigl(\varepsilon^2\bigr),&&\bhat{\tens{E}}=\dfrac{\bhat{\tens{F}}^\top \bhat{\tens{F}}-\tens{I}}{2\varepsilon},
\end{align}
wherein $\tens{I}$ is the identity, we rewrite Eq.~\eqref{eq:I1} as
\begin{align}
\mathcal{I}_1-3=2\varepsilon^2\left[\bigl(\tr{\bhat{\tens{E}}}\bigr)^2+\tr{\bhat{\tens{E}}{}^2}\right]+O\bigl(\varepsilon^3\bigr).\label{eq:effstrain}
\end{align}
This shows how, at leading order, the energy density depends only on the two invariants of the effective two-dimensional strain. In the asymptotic limit of a thin shell, the constitutive relations have thus become effectively two-dimensional.

\subsubsection{Derivation of the thin shell theory}
We are now set up to average out the transverse coordinate and thus obtain the thin shell theory. We obtain, from Eq.~\eqref{eq:dVint}, the leading-order expansion for the volume element in the intrinsic configuration,
\begin{align}
\mathrm{d}V^0&=\varepsilon\bigl(1-\lambda_s^0Z^0\bigr)\,r\,\mathrm{d}s\,\mathrm{d}\phi\,\mathrm{d}Z^0+O\bigl(\varepsilon^2\bigr)\nonumber\\
&=\dfrac{1-x}{\lambda_s^0}\,\varepsilon\,r\,\mathrm{d}s\,\mathrm{d}\phi\,\mathrm{d}x+O\bigl(\varepsilon^2\bigr). \label{eq:dV0}
\end{align}
Moreover, we introduce $\eta=\lambda_s^0h/2$, so that the shell surfaces ${\zeta^0=\pm h^0/2}$ correspond to $x=\pm \eta$. 

On substituting Eqs.~\eqref{eq:I2} and \eqref{eq:dV0} into Eq.~\eqref{eq:E}, integrating with respect to $x$, and using axisymmetry, we then obtain
\begin{subequations}
\begin{align}
\mathcal{E}=\int_{\mathcal{S}}{\hat{e}\,r\,\mathrm{d}s\,\mathrm{d}\phi}=2\pi\int_{\mathcal{C}}{\hat{e}\,r\,\mathrm{d}s}, \label{eq:E2}
\end{align}
\end{subequations}
with the first integration over the undeformed axisymmetric midsurface $\mathcal{S}$ and the second over the curve $\mathcal{C}$ generating $\mathcal{S}$. The effective two-dimensional energy density $\hat{e}$ in Eq.~\eqref{eq:E2} is
\begin{widetext}

\vspace{-3mm}
\setcounter{equation}{56}\begin{subequations}\setcounter{equation}{1}
\begin{align}
\hat{e}&=\dfrac{\varepsilon}{\lambda_s^0}\int_{-\eta}^{\eta}{e(x)(1-x)\,\mathrm{d}x}=\dfrac{C}{2}\varepsilon^3\Bigl\{h\bigl[\alpha_{ss}E_s^2+(\alpha_{s\phi}+\alpha_{\phi s})E_sE_\phi+\alpha_{\phi\phi}E_\phi^2\bigr]+2h^2\left[\beta_{ss}E_sL_s+\beta_{s\phi}E_sL_\phi+\beta_{\phi s}E_\phi L_s\right.\nonumber\\
&\hspace{55mm}+\left.\beta_{\phi\phi}E_\phi L_\phi\right]+h^3\bigl[\gamma_{ss}L_s^2+(\gamma_{s\phi}+\gamma_{\phi s})L_sL_\phi+\gamma_{\phi\phi}L_\phi^2\bigr]\Bigr\}+O\bigl(\varepsilon^4\bigr),\label{eq:edens}
\end{align}\end{subequations}
\end{widetext}
wherein
\begin{subequations}\label{eq:coeffs}
\begin{align}
\alpha_{ss}&=\dfrac{\eta^4-2\eta^2+2}{\left(1-\eta^2\right)^2}+\dfrac{2\tanh^{-1}{\eta}}{\eta},\\
\alpha_{s\phi}&=\alpha_{\phi s}=\dfrac{1}{\left(1-\eta^2\right)^2}+\dfrac{\tanh^{-1}{\eta}}{\eta},\\
\alpha_{\phi\phi}&=\dfrac{3\eta^4-6\eta^2+4}{\left(1-\eta^2\right)^2},\\
\beta_{ss}&=-\dfrac{\eta\left(2-\eta^2\right)}{2\left(1-\eta^2\right)^2},\\
\beta_{s\phi}&=\dfrac{\eta^6+4\eta^4-11\eta^2+3}{18\eta\left(1-\eta^2\right)^2}-\dfrac{\tanh^{-1}{\eta}}{6\eta^2},\\
\beta_{\phi s}&=-\dfrac{1}{2\eta\left(1-\eta^2\right)^2}+\dfrac{\tanh^{-1}{\eta}}{2\eta^2},\\
\beta_{\phi\phi}&=\dfrac{3\eta^5-5\eta^3+\eta}{6\left(1-\eta^2\right)^2},\\
\gamma_{ss}&=\dfrac{\eta^4-2\eta^2+2}{4\eta^2\left(1-\eta^2\right)^2}-\dfrac{\tanh^{-1}{\eta}}{2\eta^3},\\
\gamma_{s\phi}&=\gamma_{\phi s}=\dfrac{\eta^6-2\eta^4+\eta^2+3}{36\eta^2\left(1-\eta^2\right)^2}-\dfrac{\tanh^{-1}{\eta}}{12\eta^3},\\
\gamma_{\phi\phi}&=\dfrac{10\eta^4-21\eta^2+12}{36\left(1-\eta^2\right)^2}
\end{align}
\end{subequations}
are functions of the large bending parameter
\begin{align}
\eta=\dfrac{\lambda_s^0}{2}h=\dfrac{\kappa_s^0}{2f_s^0f_\phi^0}(\varepsilon h)=\frac{\kappa_s^0}{2}\bigl(\varepsilon h^0\bigr)\label{eq:eta}
\end{align}
only. Moreover, from Eqs.~\eqref{eq:shellstr} and \eqref{eq:curvstr}, the shell strains in Eq.~\eqref{eq:edens} are
\begin{align}
\varepsilon E_s&=\dfrac{\tilde{f}_s-f_s^0}{f_s^0},&\varepsilon E_\phi&=\dfrac{\tilde{f}_\phi-f_{\smash\phi}^0}{f_\phi^0},\label{eq:EsEphi}
\end{align}
while the curvature strains are
\begin{subequations}\label{eq:LsLphi}
\begin{align}
L_s&=\dfrac{\tilde{\kappa}_s-\kappa_s^0}{f_s^0f_\phi^0}=K_s-\dfrac{2\eta}{h}E_s+O(\varepsilon),\label{eq:Ls}\\
L_\phi&=\dfrac{\tilde{\kappa}_\phi-\kappa_{\smash\phi}^0}{f_s^0f_\phi^0}=K_\phi+O(\varepsilon),
\end{align}
\end{subequations}
where we have defined
\begin{align}
&K_s=\dfrac{\tilde{f}_s\tilde{\kappa}_s-f_s^0\kappa_s^0}{\left(f_s^0\right)^2f_\phi^0},&&K_\phi=\dfrac{\tilde{f}_\phi\tilde{\kappa}_\phi-f_\phi^0\kappa_\phi^0}{f_s^0\bigl(f_\phi^0\bigr)^2}.\label{eq:KsKphi}
\end{align}
Shell theories are expressed more naturally in terms of the alternative curvature strains $K_s,K_\phi$. Indeed, $K_s,K_\phi$ vanish for pure stretching deformations, whereas $L_s,L_\phi$ do not: consider a shell, the undeformed (and intrinsic) configuration of which is a sphere of radius $R$, and which deforms into a sphere of radius $R'=fR$, for example because of a pressure difference between the inside and outside. For this deformation, $\smash{f_s^0=f_{\smash{\phi}}^0=1}$, $\smash{\kappa_s^0=\kappa_\phi^0=1/R}$, while $\tilde{f}_s=\tilde{f}_\phi=f$, $\tilde{\kappa}_s=\tilde{\kappa}_\phi=1/fR$, and so $\smash{L_s=L_\phi=(1-f)\big/f^3R\not=0}$ for $f\not=1$, but $K_s=K_\phi=0$. Reference~\cite{audoly} has also discussed this point, noting that $L_s,L_\phi$ and $K_s,K_\phi$ can be used interchangeably in classical shell theories. However, Eq.~\eqref{eq:Ls} shows that, in the large bending limit considered here, $L_s-K_s=O(1)$. Even at leading order, the stretching deformations associated with changes in curvature cannot therefore be neglected in this limit. In terms of the alternative curvature strains $K_s,K_\phi$, Eq.~\eqref{eq:edens} becomes
\begin{widetext}

\vspace{-2mm}
\begin{align}
\hat{e}&= \dfrac{C}{2}\varepsilon^3\Bigl\{h\bigl[\bar{\alpha}_{ss}E_s^2+(\bar{\alpha}_{s\phi}+\bar{\alpha}_{\phi s})E_sE_\phi+\alpha_{\phi\phi}E_\phi^2\bigr]+2h^2\left[\bar{\beta}_{ss}E_sK_s+\bar{\beta}_{s\phi}E_sK_\phi+\beta_{\phi s}E_\phi K_s+\beta_{\phi\phi}E_\phi K_\phi\right]\nonumber\\
&\hspace{20mm}+h^3\bigl[\gamma_{ss}K_s^2+(\gamma_{s\phi}+\gamma_{\phi s})K_sK_\phi+\gamma_{\phi\phi}K_\phi^2\bigr]\Bigr\}+O\bigl(\varepsilon^4\bigr),\label{eq:edensk}
\end{align}
\end{widetext}
where $\alpha_{\phi\phi},\beta_{\phi s},\beta_{\phi\phi},\gamma_{ss},\gamma_{s\phi}=\gamma_{\phi s},\gamma_{\phi\phi}$ are still given by Eqs.~\eqref{eq:coeffs}, while
\begin{subequations}\label{eq:coeffs2}
\begin{align}
\bar{\alpha}_{ss}&=\alpha_{ss}-4\eta\beta_{ss}+4\eta^2\gamma_{ss}=\dfrac{4}{\left(1-\eta^2\right)^2},\\
\bar{\alpha}_{s\phi}&=\bar{\alpha}_{\phi s}=\alpha_{s\phi}-2\eta\beta_{\phi s}=\dfrac{2}{\left(1-\eta^2\right)^2},\\
\bar{\beta}_{ss}&=\beta_{ss}-2\eta\gamma_{ss}=-\dfrac{1}{\eta\left(1-\eta^2\right)^2}+\dfrac{\tanh^{-1}{\eta}}{\eta^2},\\
\bar{\beta}_{s\phi}&=\beta_{s\phi}-2\eta\gamma_{s\phi}=-\dfrac{\eta\left(2-\eta^2\right)}{3\left(1-\eta^2\right)^2}.
\end{align}
\end{subequations}
\subsubsection{Stretching, coupling, and bending energies}
The terms that appear in the elastic energy~\eqref{eq:edensk} separate into stretching, coupling, and bending terms, viz.
\begin{align}
\hat{e}=\hat{e}_{\text{stretch}}+\hat{e}_{\text{couple}}+\hat{e}_{\text{bend}}+O\bigl(\varepsilon^4\bigr), 
\end{align}
with
\begin{subequations}
\begin{align}
\hat{e}_{\text{stretch}}&=\dfrac{Ch}{2}\varepsilon^3\bigl[\bar{\alpha}_{ss}E_s^2+(\bar{\alpha}_{s\phi}+\bar{\alpha}_{\phi s})E_sE_\phi+\alpha_{\phi\phi}E_\phi^2\bigr],\\
\hat{e}_{\text{couple}}&=Ch^2\varepsilon^3\left[\bar{\beta}_{ss}E_sK_s+\bar{\beta}_{s\phi}E_sK_\phi+\beta_{\phi s}E_\phi K_s\right.\nonumber\\
&\hspace{25mm}+\left.\beta_{\phi\phi}E_\phi K_\phi\right],\label{eq:ecouple}\\
\hat{e}_{\text{bend}}&=\dfrac{Ch^3}{2}\varepsilon^3\bigl[\gamma_{ss}K_s^2+(\gamma_{s\phi}+\gamma_{\phi s})K_sK_\phi+\gamma_{\phi\phi}K_\phi^2\bigr].
\end{align}
\end{subequations}
As $(\bar\alpha_{s\phi}+\bar\alpha_{\phi s})^2-4\bar\alpha_{ss}\bar\alpha_{\phi\phi}=-48\left(1-\eta^2\right)^{-2}<0$ for $|\eta|<1$, the stretching energy $\hat{e}_{\text{stretch}}$ is positive semidefinite. Numerically, we also find that $(\gamma_{s\phi}+\gamma_{\phi s})^2-4\gamma_{ss}\gamma_{\phi\phi}<0$ for $|\eta|<1$, and hence the bending energy $\hat{e}_{\text{bend}}$ is positive semidefinite, too. However, the coupling energy $\hat{e}_{\text{couple}}$ can clearly be of either sign, though $\hat{e}$ is of course positive semidefinite.

All of the coefficient functions defined in Eqs.~\eqref{eq:coeffs} and \eqref{eq:coeffs2} diverge as $\eta\rightarrow\pm 1$. More precisely, the coefficients diverge like $(1-|\eta|)^{-2}$, and so Eq.~\eqref{eq:edensk} loses asymptoticity when $1-|\eta|=O\left(\sqrt{\varepsilon}\right)$, and hence the shell theory is not formally valid in this limit. This is mirrored by a similar breakdown of asymptoticity at other places in the analysis: for example, Eqs.~\eqref{eq:a1b1} show that the expansion of the deformation gradient in Eq.~\eqref{eq:F2} also breaks down when $1-|\eta|=O\left(\sqrt{\varepsilon}\right)$. However, this divergence, absent from theories not valid for large bending deformations, is not surprising in the first place. Indeed, the limit $\eta\rightarrow\pm 1$ corresponds to constricted cells, i.e. wedge-shaped, triangular cells~\figrefi{fig1}{b} for which the intrinsic meridional radius of curvature is half the intrinsic cell sheet thickness: one of the surfaces of the shell has contracted to a point in the intrinsic configuration, so is geometrically singular. As the intrinsic configuration approaches this constricted limit somewhere, deviations from the intrinsic configuration become more and more expensive energetically there compared to other positions in the shell, unless the divergence of $\hat{e}$ as $\eta\rightarrow\pm 1$ is suppressed. This happens if $\hat{e}_{\text{couple}}\approx-(\hat{e}_{\text{stretch}}+\hat{e}_{\text{bend}})<0$ or the divergence of each of $\hat{e}_{\text{stretch}},\hat{e}_{\text{couple}},\hat{e}_{\text{bend}}$ is suppressed, which is possible for special values of $E_s,E_\phi,K_s,K_\phi$, as discussed in more detail below.

Plots of the coefficient functions in Eqs.~\eqref{eq:coeffs} and \eqref{eq:coeffs2}, arbitrarily scaled with $\bar{\alpha}_{ss}$ to absorb their divergence as $\eta\rightarrow\pm 1$, are shown in~\textwholefigref{fig3}. These illustrate how the relative importance of different deformation modes depends on the amount of intrinsic bending. In other words, large bending deformations break the material isotropy, so that different directions of stretching have different effective stretching moduli; similarly, different effective bending moduli are associated with different directions of bending. This anisotropy is therefore geometric; as discussed below, this effect is absent from the classical theories not valid for large bending deformations.

\begin{figure}
\includegraphics{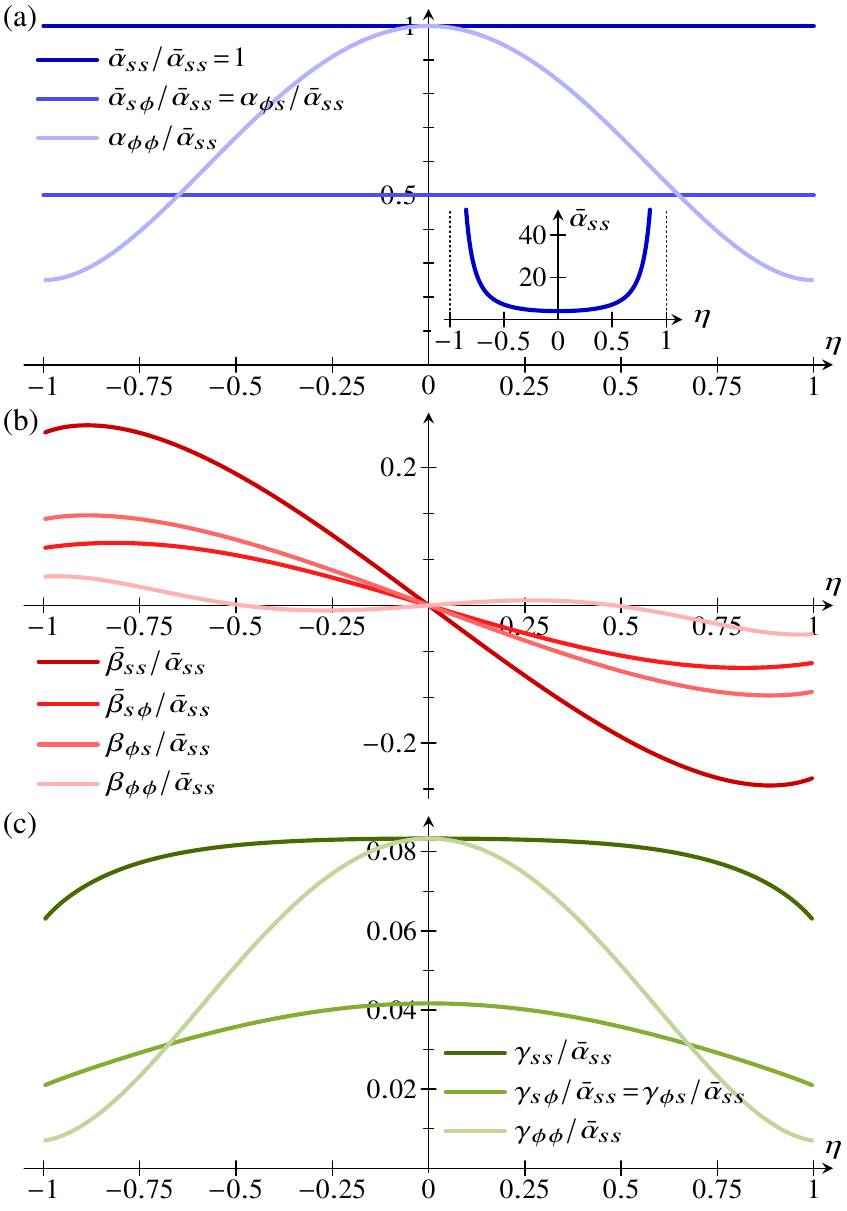}
\caption{Effective two-dimensional energy density. Plots of the coefficients in Eq.~\eqref{eq:edensk}, defined in Eqs.~\eqref{eq:coeffs} and \eqref{eq:coeffs2}, against $\eta$. All coefficients are arbitrarily scaled with $\bar{\alpha}_{ss}$ to absorb their divergence in the constriction limit $\eta\rightarrow\pm 1$. (a)~Plot of the stretching coefficients $\bar{\alpha}_{ss},\bar{\alpha}_{s\phi},\bar{\alpha}_{\phi s},\alpha_{\phi\phi}$. Inset: unscaled plot of $\bar{\alpha}_{ss}$ against $\eta$, diverging as $\eta\rightarrow\pm 1$. (b)~Plot of the mixed coefficients $\bar{\beta}_{ss},\bar{\beta}´_{s\phi},\beta_{\phi s},\beta_{\phi\phi}$. (c)~Plot of the bending coefficients $\gamma_{ss},\gamma_{s\phi},\gamma_{\phi s},\gamma_{\phi\phi}$.}
\label{fig3}
\end{figure}

This completes the derivation of the elastic energy~\eqref{eq:E2} of a thin shell undergoing large axisymmetric bending deformations. In Appendix~\ref{appB}, we derive the   associated governing equations, using the expression~\eqref{eq:edensk} of the energy density in terms of the alternative curvature strains defined in Eqs.~\eqref{eq:KsKphi}. 

\subsection{Limit of small bending deformations}\label{sec:modelc}
We conclude our calculations by taking the limit $\eta\rightarrow 0$, in which the bending deformations become small compared to the thickness of the shell. The energy density in Eq.~\eqref{eq:edensk} then limits to the form familiar from classical shell theories \smash{\cite{audoly}},
\begin{align}
&\hat{e}_0=2C\varepsilon^3\left[h\left(E_s^2+E_sE_\phi+E_\phi^2\right)+\dfrac{h^3}{12}\left(K_s^2+K_sK_\phi+K_\phi^2\right)\right],\label{eq:edens2}
\end{align}
up to corrections of order $O\bigl(\varepsilon^4\bigr)$. This is the energy density of a thin Hookean shell~\cite{ventsel,libai,audoly} with Poisson's ratio $\nu=1/2$, implying incompressibility, and elastic modulus $E=3C$. In particular, our analysis also provides a formal derivation of the morphoelastic version of this classical shell theory. Again, the energy density separates into stretching and bending terms,
\begin{align}
\hat{e}_0=\hat{e}_{0,\text{stretch}}+\hat{e}_{0,\text{bend}},
\end{align}
with
\begin{subequations}
\begin{align}
\hat{e}_{0,\text{stretch}}&=\dfrac{1}{2}(4Ch)\varepsilon^3\bigl[E_s^2+E_sE_\phi+E_\phi^2\bigr],\\
\hat{e}_{0,\text{stretch}}&=\dfrac{1}{2}\left(\dfrac{Ch^3}{3}\right)\varepsilon^3\bigl[K_s^2+K_sK_\phi+K_\phi^2\bigr],
\end{align}
\end{subequations}
but there is no term that couples the strains and curvature strains. In this theory, the same stretching modulus $E(\varepsilon h)/\bigl(1-\nu^2\bigr)=4C(\varepsilon h)$ and the same bending modulus \smash{$E(\varepsilon h)^3/\bigl[12\bigl(1-\nu^2\bigr)\bigr]=C(\varepsilon h)^3/3$} are associated with all directions of stretching or bending; to pick up on a point made earlier, it is this isotropy resulting from the constitutively assumed isotropy of the material that is broken by the geometry of large bending deformations.

Of course, Eq.~\eqref{eq:edens2} could be derived directly by imposing different scalings, of small intrinsic bending, replacing those for large bending deformations in Eq.~\eqref{eq:curvscalings}; these scalings would considerably simplify the solutions of Eqs.~\eqref{eq:G0}, \eqref{eq:det0}, and \eqref{eq:det1}. Indeed, the structure of these calculations would be broadly similar to the earlier asymptotic derivation of the classical shell theories in Ref.~\cite{steigmann13}. We emphasise that, in either derivation, the terms at order $O\bigl(\varepsilon^2\bigr)$ in the expansion~\eqref{eq:F2} of the deformation gradient need not be computed explicitly.

\subsubsection{Stretching and bending energies for small and large bending}
We compare the stretching and bending energies in the small and large bending limits by observing that
\begin{subequations}
\begin{align}
\hat{e}_{\text{stretch}}&=\hat{e}_{0,\text{stretch}}+\dfrac{\eta^2\left(2-\eta^2\right)}{\left(1-\eta^2\right)^2}\left(2E_s+E_\phi\right)^2,\label{eq:estretch}\\
\hat{e}_{\text{bend}}&=\hat{e}_{0,\text{bend}}+\dfrac{\eta^2\left(3-2\eta^2\right)}{36\left(1-\eta^2\right)^2}\left(3K_s+K_\phi\right)\left(k(\eta)K_s+K_\phi\right), \label{eq:ebend}
\end{align}
\end{subequations}
where we have used Eqs.~\eqref{eq:coeffs} and \eqref{eq:coeffs2} and defined
\begin{align}
k(\eta)=-\dfrac{\eta\left(4\eta^6-11\eta^4+10\eta^2-6\right)+6\left(1-\eta^2\right)^2\tanh^{-1}{\eta}}{\eta^5\left(3-2\eta^2\right)}. 
\end{align}
This shows that the classical theory underestimates the stretching energy of large bending deformations: ${\hat{e}_{\text{stretch}}\geq \hat{e}_{0,\text{stretch}}}$ from Eq.~\eqref{eq:estretch}. Moreover, $\hat{e}_{\text{stretch}}$ diverges as $|\eta|\rightarrow 1$ unless the deformations are such that $E_\phi=-2E_s$.

The classical theory may however overrestimate the bending energy of large bending deformations. Indeed, numerically, we find $13/5=k(0)<k(\eta)<k(\pm1)=3$ for $|\eta|<1$, and hence, from Eq.~\eqref{eq:ebend}, $\hat{e}_{\text{bend}}<\hat{e}_{0,\text{bend}}$ if and only if $K_sK_\phi<0$ and $k(\eta)|K_s|<|K_\phi|<3|K_s|$. Also from Eq.~\eqref{eq:ebend}, $\hat{e}_{\text{bend}}$ diverges as $|\eta|\rightarrow 1$ unless $K_\phi=-3K_s$.

In particular, $\hat{e}_{\text{stretch}}$ and $\hat{e}_{\text{bend}}$ are both bounded as $|\eta|\rightarrow 1$ if and only if $E_\phi=-2E_s$ and $K_\phi=-3K_s$. In this case, Eq.~\eqref{eq:ecouple} shows that $\hat{e}_{\text{couple}}$ is also bounded as $|\eta|\rightarrow 1$. The conditions $E_\phi=-2E_s$, $K_\phi=-3K_s$ thus define the special deformations that allow the stretching, bending, and coupling energies to remain bounded as $|\eta|\rightarrow 1$ that we mentioned earlier.

\vspace{-1mm}
\subsubsection{Other elastic shell theories}
The energy density in Eq.~\eqref{eq:edens2} has the same structure as the elastic energy densities used in the models referenced in the introduction, but the morphoelastic definitions of the shell and curvature strains in Eqs.~\eqref{eq:EsEphi} and \eqref{eq:KsKphi} differ from those in these previous models: in models not based on morphoelasticity and its multiplicative decomposition of the deformation gradient~\cite{hohn15,haas15,heer17,yevick19,miller18}, the shell and curvature strains are simply differences of stretches or curvatures, missing the scaling factors of $\smash{f_s^0,f_\phi^0}$ that appear in Eqs.~\eqref{eq:EsEphi} and \eqref{eq:KsKphi}. We also note that the expressions for the curvature strains in Eqs.~\eqref{eq:KsKphi} differ, by a factor of $g^0=\smash{f_s^0f_\phi^0}$, from those in Refs.~\cite{haas18a,haas18b}, which, as discussed in the Introduction, used a geometric approach to derive a morphoelastic shell theory. Earlier, we noted that this factor corresponds to the stretching of the intrinsic midsurface. Moreover, the $O(1)$ solution implies that $\tilde{\zeta}=\zeta^0+O(\varepsilon)$. Hence, by the definition of the midsurfaces, $\tilde{h}^\pm=\pm h^0/2+O(\varepsilon)$, and so the deformed cell sheet has thickness $\tilde{h}=\tilde{h}^++\tilde{h}^-=h^0+O(\varepsilon)$. Eq.~\eqref{eq:h} therefore yields $h/\tilde{h}=h/h^0+O(\varepsilon)=g^0+O(\varepsilon)$. The fact that the curvature strains in Eq.~\eqref{eq:KsKphi} decrease as $g^0$ increases therefore expresses the fact that the shell becomes easier to bend as it thins as a result of this stretching of the midsurface, with $\hat{e}_{\text{bend}},\hat{e}_{0,\text{bend}}\propto g_{\smash0}^{-2}$. This geometric role of the factor $g^0$ has been noticed previously in the context of uniform growth of an elastic shell~\cite{pezzulla17}. 

The geometric approach in Refs.~\cite{haas18a,haas18b} also leads to additional terms in the energy density. The present analysis proves that these terms are not leading-order terms in the thin shell limit. However, there is no reason to expect this geometric approach to yield all terms at next order in the asymptotics. A complete expansion could in principle be obtained by continuing the asymptotic analysis presented here. Taking the analysis to higher orders in this way would in particular answer the question: at what order does the Kirchhoff hypothesis break down, i.e. at what order do the normals to the deformed midsurface diverge from those to the undeformed midsurface? This would permit asymptotic justification of the so-called shear deformation theories~\cite{reddy} in which the normals to the undeformed midsurface need not remain normals in the deformed configuration, but we do not pursue this further here.

\section{Invagination in \emph{Volvox}}\label{sec:volvox}
\subsection{Biological background}
The green algal genus \emph{Volvox}~\cite{kirkbook} has become a model for the study of the evolution of multicellularity~\cite{kirkessay,herron16}, for biological fluid dynamics~\cite{goldstein15}, and for problems in developmental biology~\cite{kirkreview,matt16}. Adult \emph{Volvox} colonies~\figref{fig4}{a} are spheroidal, consisting of several thousand biflagellated somatic cells that enclose a small number of germ cells~\cite{kirkbook}. Each germ cell undergoes several rounds of cell division to form a spherical embryonic cell sheet~[Figs.~\figrefp{fig4}{b} and \figrefp{fig4}{e}], at which stage those cell poles whence will emanate the flagella point into the sphere~\cite{kirkbook}. To acquire motility, the embryo turns itself inside out in a process called inversion~\cite{hallmann06,desnitskiy18}.

\begin{figure}[b]
\includegraphics{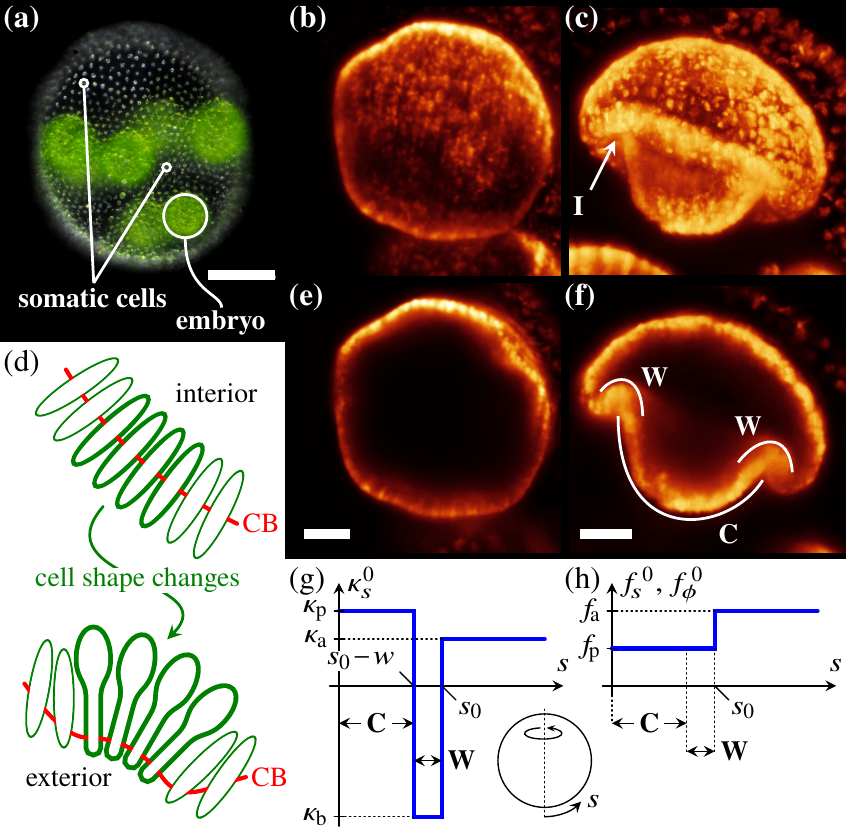}
\caption{Invagination in \emph{Volvox}. (a) \emph{Volvox} colony, with somatic cells and one embryo labelled. (b)~Light-sheet microscopy image of a spherical \emph{Volvox} embryo before inversion. (c) Corresponding image at an early stage of inversion, when a circular invagination (I) has formed. (d)~Splaying of cells and bending of the cell sheet result from the formation of wedge-shaped cells and the rearrangement of the cytoplasmic bridges (CBs); red lines indicate position of CBs. (e)~Midsagittal cross-section of a \emph{Volvox} embryo before inversion. (f)~Corresponding cross-section during invagination, with the regions where wedge-shaped cells (W) and contracted spindle-shaped cells (C) have formed labelled. (g)~Plot of the intrinsic curvature $\kappa_s^0$ against arclength $s$, defined in the inset. The plot defines the model parameters $\kappa_{\mathrm{p}},\kappa_{\mathrm{b}},\kappa_{\mathrm{a}}$, $s_0$, and $w$. Regions of cell shape changes (W,~C) as in (f) are also indicated. (h)~Corresponding plot of the intrinsic stretches \smash{$f_s^0,f_{\smash{\phi}}^0$}, defining additional model parameters $f_{\mathrm{p}},f_{\mathrm{a}}$. Panels \mbox{(a)--(f)} include microscopy images by Stephanie H\"ohn and have been redrawn from Ref.~\cite{haas15}. Scale bars: (a) $50\,\text{\textmu m}$; \mbox{(e), (f) $20\,\text{\textmu m}$}.}\label{fig4} 
\end{figure}

In some species of \emph{Volvox}~\cite{hohn11,hallmann06}, inversion starts with the formation of a circular invagination~[Figs.~\figrefp{fig4}{c} and \figrefp{fig4}{f}], reminiscent of the cell sheet folds associated with processes such as gastrulation or neurulation in higher organisms. At the cell level, this invagination results from two types of cell shape changes~\cite{hohn11,hohn15}: (1) cells near the equator become wedge-shaped~\figref{fig4}{d}, while the cytoplasmic bridges (cell-cell connections resulting from incomplete division) rearrange to connect the cells at their thin wedge ends, and (2) cells in the posterior hemisphere narrow in the meridional direction. These cell shape changes arise simultaneously, with (1) splaying the cells and thereby bending the cell sheet~\figref{fig4}{d} and (2) contracting the posterior hemisphere to facilitate the subsequent inversion of the posterior hemisphere inside the as yet uninverted anterior hemisphere. 

At later stages of inversion, other cell shape changes arise in different parts of the cell sheet~\cite{haas18a,hohn11} to ease the peeling of the anterior hemisphere over the inverted posterior and thus complete inversion. In particular, the anterior hemisphere of the cell sheet thins as cells there stretch anisotropically~\cite{haas18a,hohn11}.

\subsection{Results}
Following our earlier work~\cite{hohn15,haas15,haas18a,haas18b}, we model \emph{Volvox} inversion by considering the deformations of an incompressible elastic spherical shell under quasi-static axisymmetric variations of its intrinsic stretches and curvatures representing the cell shape changes driving inversion. The slow speed of inversion---it takes about an hour for a \emph{Volvox} embryo to turn itself inside out~\cite{hallmann06,hohn11}---justifies this quasi-static approximation. In more detail, Figs.~\figrefp{fig4}{g} and \figrefp{fig4}{h} show functional forms of the intrinsic stretches and curvatures encoding the cell shape changes driving invagination and define the model parameters $\kappa_{\mathrm{p}},\kappa_{\mathrm{b}},\kappa_{\mathrm{a}}$, $f_{\mathrm{p}},f_{\mathrm{a}}$, $s_0$, and $w$ that encode the intrinsic curvatures and intrinsic stretches of different regions of the cell sheet and the extent of these regions. In numerical calculations, we regularise the step discontinuities in the definitions of the intrinsic stretches and curvatures in Figs.~\figrefp{fig4}{g} and \figrefp{fig4}{h}, we non-dimensionalise all lengths with the pre-inversion radius $R$ of the embryo, and we take $\varepsilon h=0.15$, appropriate for \emph{Volvox globator}~\cite{hohn15,haas18a}.

We solve the governing equations derived in Appendix~\ref{appB} numerically using the boundary value problem solver \texttt{bvp4c} of \textsc{Matlab} (The MathWorks, Inc.) and the continuation software \textsc{auto}~\cite{auto}.

\begin{figure*}
\includegraphics{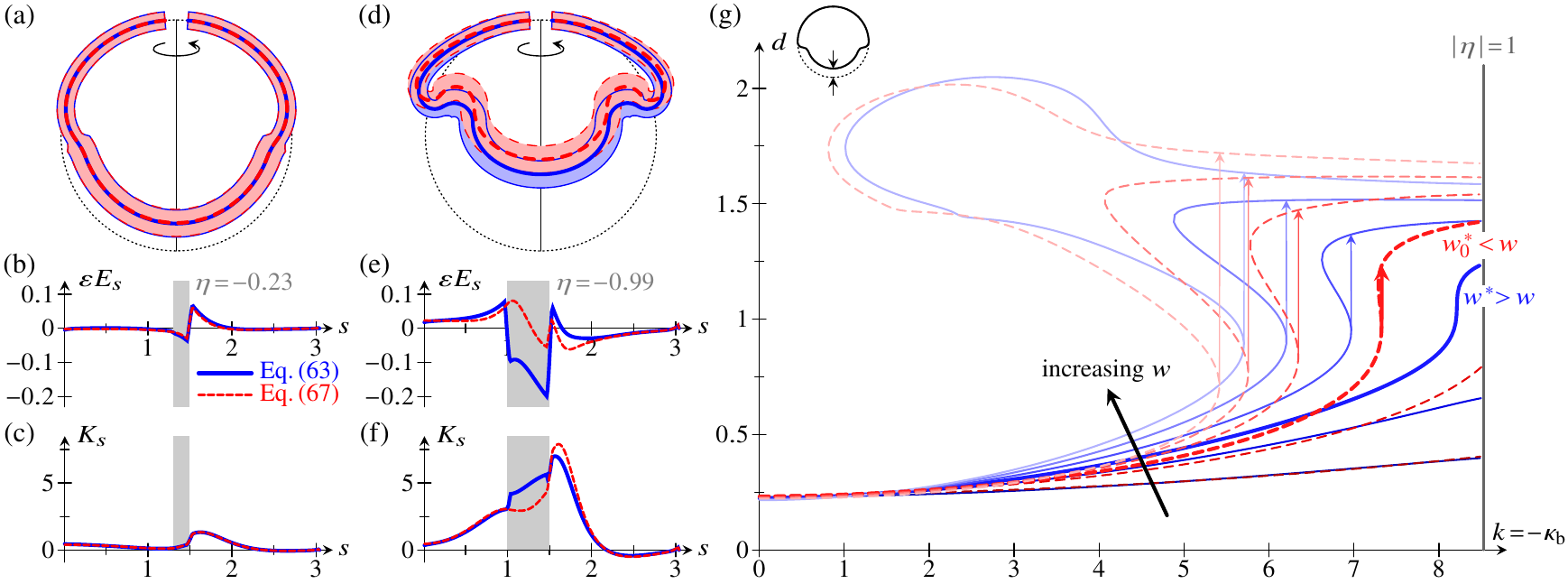}
\caption{Comparison of the elastic model for large bending deformations and the classical model. Solid lines: large bending model with energy density given by Eq.~\eqref{eq:edensk}; dashed lines: classical model with energy density given by Eq.~\eqref{eq:edens2}. (a)~Early invagination stage: the two models yield very similar shapes. Thick lines: midline of the cell sheet; thin lines and shaded area: transverse extent of the shell, illustrating the thickness variations resulting from the cell shape changes. Dotted line: midline of the undeformed spherical shell. Parameter values: $\kappa_{\mathrm{p}}=\kappa_{\mathrm{a}}=1$, $\kappa_{\mathrm{b}}=-2$, $f_{\mathrm{p}}=0.8$, $f_{\mathrm{a}}=1$, $s_0=1.5$, $w=0.2$. (b)~Corresponding plot of the meridional shell strain $E_s$. The grey shaded area marks the bend region $s_0-w<s<s_0$. (c)~Corresponding plot of the meridional curvature strain $K_s$. (d)~Later invagination stage: as the cells in the bend region approach the constriction limit, the shapes resulting from the two models start differ increasingly. Parameter values are as in (a), except $\kappa_{\mathrm{b}}=-8.5$, $w=0.5$. (e)~Corresponding plot of the meridional shell strain $E_s$. (f)~Corresponding plot of the meridional curvature strain $K_s$. (g)~Bifurcation diagram, for different values of $w$, in $(k,d)$ space, where $k=-\kappa_{\mathrm{b}}$ and $d$ is the posterior displacement defined in the axis inset. Different lines correspond to parameter values $w=0.3,0.5,0.6,0.7,0.8,0.9$. Other parameter values are as in (a). The vertical line $|\eta|=1$ corresponding to the constriction limit is also shown. For $w>w^\ast$ (in the large bending model) or $w>w^\ast_{\smash0}$ (in the classical model), discontinuous jumps in $d$, denoted by vertical arrows, arise as $k$ is increased. The thick lines correspond to $w=0.6$ and show that $w^\ast>w^\ast_{\smash0}$.}
\label{fig5}
\end{figure*}

During the invagination stage, the radius of curvature in the bend region of wedge-shaped cells~\figref{fig4}{f} becomes comparable to the thickness of the cell sheet: this is the scaling limit of large bending deformations studied in Section~\ref{sec:model}. We therefore compare the resulting elastic model, with energy density~\eqref{eq:edens}, to the classical theory, in which the energy density is given by Eq.~\eqref{eq:edens2}. For weakly invaginated stages of \emph{Volvox} inversion (corresponding to small values of $\eta$ in the large bending theory), the two models yield, unsurprisingly, very similar shapes~\figref{fig5}{a}, mirrored by very similar profiles of meridional shell strain~\figref{fig5}{b} and meridional curvature strain~\figref{fig5}{c}. The contraction of the posterior hemisphere leads to thickening of the cell sheet there~\figref{fig5}{a}. However, the more the intrinsic configuration of the cell sheet approaches the limit of cell constriction, the more the shapes resulting from the two models differ~\figref{fig5}{d}. Correspondingly, the meridional shell strain~\figref{fig5}{e} and meridional curvature strain~\figref{fig5}{f} in the two models differ increasingly. It may seem counterintuitive that these strains are larger in the bend region of nearly constricted cells for the large-bending model than for the classical model~[Figs.~\figrefp{fig5}{e} and \figrefp{fig5}{f}], since the stretching and bending cost of these larger strains is much higher in the large-bending model than in the classical model. Indeed, on computing the stretching and bending energies (not shown) of the shapes in \textfigref{fig5}{d}, we find them to be much larger in the large-bending model than in the classical model. However, these large energies are balanced by a correspondingly large and negative coupling energy: for example, $E_s<0$ and $K_s>0$ in the bend region~[Figs.~\figrefp{fig5}{e} and \figrefp{fig5}{f}], while $\eta<0\Longrightarrow\bar{\beta}_{ss}>0$~\figref{fig3}{b}, and so $\bar{\beta}_{ss}E_sK_s<0$. This negative coupling energy therefore explains the large strains in the bend region that arise in the large-bending model.

The largest curvature strains~\figref{fig5}{f} arise, however, in the anterior fold, i.e. in the second bend region that arises as a passive mechanical consequence of the wedge-shaped cells in the bend region just next to it~\cite{hohn15,haas18a}. As a result of the contraction of the posterior hemisphere, the cell sheet is thinner in the anterior~\figref{fig5}{d}, and hence is easier to bend there, as discussed earlier. In fact, around the invagination stage in \textfigref{fig5}{d}, cells in the anterior fold begin to stretch in the meridional direction~\cite{hohn11,haas18a}, leading to further thinning and increased bendability of the cell sheet there.

The examples in Figs.~\figrefp{fig5}{a} and \figrefp{fig5}{d} indicate that the results of the two models differ at a quantitative, if not at a qualitative level. We extend this observation by plotting, for both models, $k=-\kappa_{\mathrm{b}}$ against the displacement $d$ of the posterior pole~\figrefi{fig5}{g} for different values of the width $w$ of the bend region in \textfigref{fig5}{g}. Again, the solution curves show similar behaviour in the two models, but differ at a quantitative level. They confirm what one observes in \textfigref{fig5}{d}, that the cell sheet is more invaginated, at the same parameter values and for sufficiently large $k$, in the classical model than in the large-bending model. Nonetheless, the cell sheet invaginates completely even in the large-bending model as $w$ increases~\figref{fig5}{g}, i.e. as more cells become wedge-shaped and the bend region widens, as observed during \emph{Volvox} inversion~\cite{hohn11}. Moreover, one can argue that invagination is actually more stable in the large-bending model: there is a critical bend region width, $w_\ast$ in the large-bending model and $w^\ast_{\smash0}$ in the classical model, such that the solution curves in the $(k,d)$ diagram are single-valued for $w<w_\ast$ or $w<w^\ast_{\smash0}$, but become multivalued for $w>w_\ast$ or $w>w^\ast_{\smash0}$, respectively, leading to discontinuous jumps in $d$ as $k$ is varied. Where multiple solutions exist for a given value of $k$, the one with the lowest value of $d$ has the lowest energy (not shown). For the classical theory, we have discussed this bifurcation behaviour in Ref.~\cite{haas15}, and rationalised it by constructing an effective energy that estimates different elastic contributions. It is therefore not surprising that, here, we find qualitatively identical bifurcation behaviour in the two models, but that again, there are quantitative differences in the bifurcation behaviour. However, \textfigref{fig5}{g} shows that $w_\ast>w^\ast_{\smash0}$. In other words, continuous invagination is possible in a larger region of parameter space in the large bending theory than in the classical theory: in this sense, invagination is stabilised in the large-bending theory.

This discussion shows how the geometry of large bending deformations modifies the mechanical picture of invagination suggested by the classical theory. When we introduced the problem of large bending deformations, we argued that classical shell theories cannot describe these deformations because of the assumption of large radii of curvature inherent in them. At this stage, we must therefore ask: can the large-bending theory derived here provide a complete description of the mechanics of invagination? This is first a question of self-consistency: is the intrinsic configuration not ``too incompatible''? In other words, are the deformations resulting from the imposed intrinsic stretches and curvatures consistent with the scalings~\eqref{eq:shellstr} and \eqref{eq:curvstr} assumed in the derivation of the shell theory? Even for the late invagination stage in \textfigref{fig5}{d}, the meridional shell strain remains small~\figref{fig5}{e}, although the meridional curvature strain reaches values of order $O(1/\varepsilon)$~\figref{fig5}{f}. Of course, the invagination stage in~\textfigref{fig5}{d} does not satisfy the restriction $1-|\eta|\gg\sqrt{\varepsilon}$ of our shell theory discussed earlier. This kind of condition is particularly restrictive for biological tissues in which $\varepsilon$ is not ``that small''~\wholefigref{fig1}. While results remain qualitatively unchanged for somewhat smaller values of $|\eta|$ within that range of validity, this hints that understanding the elasticity of the constriction limit $|\eta|\rightarrow 1$ remains a key open problem for future work.

\section{Conclusion}\label{sec:concl}
In this paper, we have derived a morphoelastic shell theory valid for the large bending deformations that are commonly observed in developmental biology~\wholefigref{fig1}, and have shown how this new scaling limit of large bending deformations induces a purely geometric effective material anisotropy absent from classical shell theories. Taking the invagination of the green alga \emph{Volvox} as an example, we have compared this large-bending theory to a simpler, classical theory not formally valid for large bending deformations. Since the classical theory does not account for the geometric material anisotropy or the singularity of cell constriction, it differs, for strongly invaginated shapes as in Figs.~\figrefp{fig1}{b}, \figrefp{fig4}{c}, and \figrefp{fig4}{f}, from the theory for large bending deformation at a quantitative, if not at a qualitative level. In particular, we have argued that these geometric effects stabilise \emph{Volvox} invagination.

This and the growing interest in quantitative rather than merely qualitative analyses of morphogenesis~\cite{cooper08,oates09} emphasise the importance of this new scaling limit of large bending deformations for studies of the mechanics of developmental biology. The theory we have derived here is not however the most general theory of these large bending deformations. Indeed, when writing down the expression for the intrinsic deformation gradient in Eq.~\eqref{eq:F0}, we assumed that there is no intrinsic displacement parallel to the midsurface, $\varsigma^0=0$. The nonlinear differential equations extending Eqs.~\eqref{eq:G0} and \eqref{eq:det0} that arise in the expansions of the boundary and incompressibility conditions for $\varsigma^0\not=0$ still admit a trivial solution $p_{(0)}=1$, $Z_{(0)}\equiv Z^0$, $S_{(0)}\equiv S^0$, where $S^0=\smash{f_s^0f_{\smash\phi}^0}\varsigma^0$. We were however unable to extend our calculations in Section~\ref{sec:model} to prove that this solution is unique; a similar issues arises when extending the calculations of this paper to more general constitutive relations, as discussed below and in Appendix~\ref{appC}. It therefore remains unclear what form the extension of the Kirchhoff ``hypothesis''~\cite{audoly} to this case takes.

In this paper, we assumed the simplest, incompressible neo-Hookean constitutive relations when deriving our shell theory for large bending deformations. The restriction to incompressible elastic materials is justified by the biological context of our analysis, in which the models derived here describe sheets of fluid-filled cells that are therefore indeed incompressible to a first approximation. However, the bulk elastic response of biological materials such as brain tissue is not linear~\cite{mihai15,mihai17,budday17}. The restriction to linear neo-Hookean relations may therefore appear to be a limitation of the analysis, but that turns out not to be the case: in the thin shell limit, general hyperelastic constitutive relations reduce to neo-Hookean relations. This result has been established previously for thin plates~\cite{erbay97,dervaux08}, and, in Appendix~\ref{appC}, we (partially) extend it to the large bending deformations of thin shells considered here. In the context of shell theories, the problem of specifying the nonlinear constitutive relations of biological tissues does not therefore arise. However, we have recently shown that the continuum limit of a class of discrete models of cell sheets involves not only nonlinear elastic, but also nonlocal, nonelastic terms~\cite{haas19}. Moreover, adding the geometric singularity of apical constriction (corresponding to triangular cells in the underlying discrete model) as a constraint to the variational problem that arises in this continuum limit remains an important open problem~\cite{haas19}. Solving this may provide a regularisation of the singularity that breaks asymptoticity as $|\eta|\rightarrow1$ in the theory derived here, and hence a yet more complete mechanical picture of the bend region of wedge-shaped cells in \emph{Volvox} invagination~\figref{fig4}{d}. Meanwhile, all of this suggests that the journey towards understanding the continuum mechanics of biological materials, on which we have taken another step with the present analysis of large bending deformations of thin elastic shells, will continue to abound with new problems in nonlinear mechanics. 

\begin{acknowledgments}
We thank two anonymous referees for helpful reports and, in particular, incisive questions bearing on the definition of the intrinsic configuration. We also thank S. S. M. H. H\"ohn for discussions about \emph{Volvox} inversion, M. Gomez for comments on the manuscript, and A. Goriely and C. P. Turner for a discussion of tensor algebraic matters. We gratefully acknowledge support from the Engineering and Physical Sciences Research Council (Established Career Fellowship EP/M017982/1 to R.E.G.), the Wellcome Trust (Investigator Award 207510/Z/17/Z to R.E.G.), and Magdalene College, Cambridge (Nevile Research Fellowship to P.A.H.).
\end{acknowledgments}

\appendix\section{\uppercase{Thin-Shell Theory for Large Bending Deformations of an Elastic Shell}}\label{appA}
In this Appendix, we extend the calculations for axisymmetric deformations of an elastic shell in Section~\ref{sec:model} to general deformations.
\subsection{Deformations of an elastic shell}
As in Section~\ref{sec:model}, we begin by deriving expressions for the deformation gradient tensors of an elastic shell of thickness $\varepsilon h$, where $\varepsilon$ is, again, a small asymptotic parameter that expresses the thinness of the shell.
\subsubsection{Undeformed configuration of the shell}
We parameterise the undeformed midsurface $\mathcal{S}$ of the shell in terms of generalised, not necessarily orthogonal coordinates; we shall use Greek letters to denote these coordinates. Thus, if $\vec{\rho}$ is the position of a point on $\mathcal{S}$, the tangent vectors there are $\vec{e_\alpha}=\partial\vec{\rho}/\partial\alpha$. The metric $\mat{g}$ of the midsurface thus has components $g_{\alpha\beta}=\vec{e_\alpha}\cdot\vec{e_\beta}$, and we set $g=\det{\mat{g}}$.

Next, we define a basis $\mathcal{B}$ for the shell by adjoining the unit normal vector $\vec{n}$ to this tangent basis. This obeys the Weingarten equation~\cite{* [] [{ Chap.~4, pp.~78--99 and Chap.~6, pp.~125--136.}] kreyszig}
\begin{align}
&\vec{n}_{,\alpha}=-{\varkappa_\alpha}^\beta\vec{e_\beta},\label{eq:W}
\end{align}
in which commata denote partial differentiation and the (symmetric) curvature tensor is $\varkappa_{\alpha\beta}=-\vec{e_\alpha}\cdot\vec{n}_{,\beta}$.

The position of a point in the undeformed configuration $\mathcal{V}$ of the shell is $\vec{r}=\vec{\rho}+\varepsilon\zeta\vec{n}$, where $\zeta$ denotes the transverse coordinate, as defined for axisymmetric deformations in \textfigref{fig2}{c}. Hence
\begin{align}
&\vec{r}_{,\alpha}=\left({\delta_\alpha}^\beta-\varepsilon\zeta {\varkappa_\alpha}^\beta\right)\vec{e_\beta},&& \vec{r}_{,\zeta}=\varepsilon\vec{n},\label{eq:dr}
\end{align}
wherein we have used the Weingarten equation~\eqref{eq:W}, and where $\delta$ is the Kronecker delta. The metric $\mat{G}$ of the undeformed configuration therefore has components
\begin{subequations}\label{eq:Gmetric}
\begin{align}
&G_{\zeta\zeta}=\varepsilon^2,&&G_{\alpha\zeta}=G_{\zeta\alpha}=0,\label{eq:Gmetrica}
\end{align}
and
\begin{align}
G_{\alpha\beta}=g_{\alpha\gamma}\Bigl({\delta^\gamma}_\delta-\varepsilon\zeta{\varkappa^\gamma}_\delta\Bigr)\left({\delta^\delta}_\beta-\varepsilon\zeta{\varkappa^\delta}_\beta\right),\label{eq:Gmetricb}
\end{align}
\end{subequations}
where we have used the symmetry of the curvature tensor. In particular, its inverse has components
\begin{align}
&G^{\zeta\zeta}=\varepsilon^{-2},&&G^{\alpha\zeta}=G^{\zeta\alpha}=0,&&G^{\alpha\beta}.\label{eq:Gmetricup} 
\end{align}

The position vectors of the surfaces $\zeta=\pm h^\pm$ of the undeformed shell are $\vec{r^\pm}=\vec{\rho}\pm\varepsilon h^\pm\vec{n}$, and hence the tangent vectors to these surfaces are
\begin{align}
\vec{e^\pm_\alpha}=\vec{r^\pm}_{,\alpha}=\left({\delta_\alpha}^\beta\mp \varepsilon h^\pm{\varkappa_\alpha}^\beta\right)\vec{e_\beta}\pm\varepsilon {h^\pm}_{,\alpha}\vec{n}.
\end{align}
We now order $\mathcal{B}=\{\vec{e_1},\vec{e_2},\vec{n}\}$ as a right-handed basis by exchanging $\vec{n}\leftrightarrow-\vec{n}$ if required. Expanding in components, this implies that $\vec{e_1}\times\vec{e_2}=\sqrt{g}\vec{n}$, and hence $\vec{e_1}\times\vec{n}=-\smash{\vec{e^2}/\sqrt{g}}$, $\vec{e_2}\times\vec{n}=\smash{\vec{e^1}/\sqrt{g}}$. Continuing to expand in components and after some calculations, we infer
\begin{align}
\vec{e_{\smash{1}}^\pm}\times\vec{e_{\smash{2}}^\pm}&=\left[1\mp 2\varepsilon h^\pm H+\varepsilon^2(h^\pm)^2K\right]\sqrt{g}\vec{n}\nonumber\\
&\quad \mp\varepsilon {h^\pm}_{,\alpha}\left[\left(1\pm \varepsilon h^\pm H\right){\delta^\alpha}_\beta\mp\varepsilon\smash{{\varkappa}_\beta}^\alpha h^\pm\right]\dfrac{\vec{e^\beta}}{\sqrt{g}},
\end{align}
wherein we have identified $H=\tfrac{1}{2}{\varkappa_\alpha}^\alpha$ and $K=\det{{\varkappa_\alpha}^\beta}$ as the mean and Gaussian curvatures~\cite{kreyszig} of $\mathcal{S}$. On normalising these vectors, we obtain the normals to the shell surfaces,
\begin{subequations}\label{eq:NA}
\begin{align}
\vec{n^\pm}=\dfrac{\vec{n}\mp{\nu^\pm}_\alpha\vec{e^\alpha}}{\sqrt{1+{\nu^\pm}_{\beta}\nu^{\pm\,\beta}}},
\end{align}
with
\begin{align}
{\nu^\pm}_{\alpha}=\dfrac{\varepsilon {h^\pm}_{,\beta}\left[\left(1\pm \varepsilon h^\pm H\right){\delta^\beta}_\alpha\mp\varepsilon\smash{{\varkappa}_\alpha}^\beta h^\pm\right]}{g\left[1\mp 2\varepsilon h^\pm H+\varepsilon^2\left(h^\pm\right)^2K\right]}.
\end{align}
\end{subequations}

\subsubsection{Deformed configuration of the shell}
We take the same generalised coordinates to parameterise the deformed midsurface $\tilde{\mathcal{S}}$ of the shell. The tangent vectors at a point $\vec{\tilde{\rho}}$ on $\tilde{\mathcal{S}}$ are thus $\vec{\tilde{e}_\alpha}=\partial\vec{\tilde{\rho}}/\partial\alpha$. The metric $\tilde{\mat{g}}$ of the midsurface has components \smash{$\tilde{g}_{\alpha\beta}=\vec{\tilde{e}_\alpha}\cdot\vec{\tilde{e}_\beta}$}, and we let ${\tilde{g}=\det{\tilde{\mat{g}}}}$. We extend the tangent basis of \smash{$\tilde{\mathcal{S}}$} to a basis $\tilde{\mathcal{B}}$ for the deformed shell by adding the unit normal $\vec{\tilde{n}}$, and introduce the (symmetric) curvature tensor $\tilde{\kappa}_{\alpha\beta}=-\vec{\tilde{e}_\alpha}\cdot\vec{\tilde{n}}_{,\beta}$. The Weingarten and Gau\ss{} equations~\cite{kreyszig}
\begin{align}
&\vec{\tilde{n}}_{,\alpha}=-{\tilde{\kappa}_\alpha}^\beta\vec{e_\beta},&&\vec{\tilde{e}}_{\vec{\alpha},\beta}=\tilde{\kappa}_{\alpha\beta}\vec{\tilde{n}}+{\tilde{\Gamma}_{\alpha\beta}}^\gamma\vec{\tilde{e}_\gamma}\label{eq:GW}
\end{align}
express the derivatives of the normal and tangent vectors in terms of the curvature tensor and Christoffel symbols associated with the deformed midsurface metric~\cite{kreyszig}. The position of a point in the deformed configuration $\smash{\tilde{\mathcal{V}}}$ of the shell is
\begin{align}
\vec{\tilde{r}}=\vec{\tilde{\rho}}+\varepsilon\left(\tilde{\zeta}\vec{\tilde{n}}+\tilde{\varsigma}^\alpha\vec{\tilde{e}_\alpha}\right),
\end{align}
where $\tilde{\zeta}$ and $\tilde{\varsigma}^\alpha$ are the transverse and parallel displacements of this point relative to the midsurface, defined for axisymmetric deformations in~\textfigref{fig2}{e}. In particular, the displacement parallel to the midsurface is now no longer a scalar. Using the Weingarten and Gau\ss{} equations~\eqref{eq:GW}, we find
\begin{subequations}\label{eq:drtilde}
\begin{align}
\vec{\tilde{r}}_{,\alpha}&=\left[{\delta_\alpha}^\beta+\varepsilon\left({\tilde{\varsigma}^\beta}_{;\alpha}-\tilde{\zeta}{\smash{\tilde{\kappa}_\alpha}}^{\beta}\right)\right]\vec{\tilde{e}_\beta}+\varepsilon\left(\tilde{\zeta}_{,\alpha}+\tilde{\varsigma}^\beta\tilde{\kappa}_{\alpha\beta}\right)\vec{\tilde{n}},\label{eq:rtildea}\\
\vec{\tilde{r}}_{,\zeta}&=\varepsilon\left(\tilde{\zeta}_{,\zeta}\vec{\tilde{n}}+{\tilde{\varsigma}^\alpha}_{,\zeta}\vec{\tilde{e}_\alpha}\right),
\end{align}
\end{subequations}
in which ${\tilde{\varsigma}^\beta}_{;\alpha}={\tilde{\varsigma}^\beta}_{,\alpha}+\smash{\tilde{\Gamma}_{\alpha\gamma}}^\beta\tilde{\varsigma}^\gamma$ is a covariant derivative. It follows that the metric $\smash{\tilde{\mat{G}}}$ of $\smash{\tilde{\mathcal{V}}}$ has components
\begin{subequations}\label{eq:Gdef}
\begin{align}
\tilde{G}_{\zeta\zeta}&=\varepsilon^2\left[\left(\tilde{\zeta}_{,\zeta}\right)^2+{\tilde{\varsigma}^\alpha}_{,\zeta}\tilde{\varsigma}_{\alpha,\zeta}\right],\\
\tilde{G}_{\alpha\zeta}&=\tilde{G}_{\zeta\alpha}=\varepsilon\tilde{\varsigma}_{\alpha,\zeta}+\varepsilon^2\left[\tilde{\zeta}_{,\zeta}\left(\tilde{\zeta}_{,\alpha}+\tilde{\varsigma}^\beta\tilde{\varkappa}_{\alpha\beta}\right)\right.\nonumber\\
&\hspace{30mm}+\left.\tilde{\varsigma}_{\beta,\zeta}\left({\tilde{\varsigma}^\beta}_{;\alpha}-\tilde{\zeta}{\tilde{\kappa}_\alpha}^\beta\right)\right],\\
\tilde{G}_{\alpha\beta}&=\tilde{g}_{\alpha\gamma}\left[{\delta^\gamma}_\delta\!+\!\varepsilon\Bigl({\tilde{\varsigma}_\delta}^{;\gamma}\!-\!\tilde{\zeta}{\tilde{\kappa}^\gamma}_\delta\Bigr)\right]\left[{\delta_\beta}^\delta\!+\!\varepsilon\left({\tilde{\varsigma}^\delta}_{;\beta}\!-\!\tilde{\zeta}{\tilde{\kappa}_\beta}^\delta\right)\right]\nonumber\\
&\quad+\varepsilon^2\Bigl(\tilde{\zeta}_{,\alpha}+\tilde{\varsigma}^\gamma\tilde{\kappa}_{\alpha\gamma}\Bigr)\left(\tilde{\zeta}_{,\beta}+\tilde{\varsigma}^\delta\tilde{\kappa}_{\beta\delta}\right).
\end{align}
\end{subequations}
\subsubsection{Intrinsic configuration of the shell: Incompatibility}
We define the intrinsic configuration of the shell by specifying the symmetric positive-definite intrinsic metric $\smash{{g^0}_{\alpha\beta}}$, the symmetric intrinsic curvature tensor $\smash{{\kappa^0}_{\alpha\beta}}$, and the intrinsic transverse displacement $\zeta^0$, which is an increasing function of $\zeta$. It follows from a local embedding theorem for Riemannian metrics~\cite{janet26,cartan27} that the surface $\mathcal{S}^0$ with metric $\smash{{g^0}_{\alpha\beta}}$ can be embedded into three-dimensional Euclidean space, and we denote by $\mathcal{B}^0$ the corresponding intrinsic basis containing the tangent vectors $\vec{E_\alpha}$ and the normal $\vec{N}$ such that ${g^0}_{\alpha\beta}=\vec{E_\alpha}\cdot\vec{E_\beta}$.

The components of the curvature tensor $\smash{{\varkappa^0}_{\alpha\beta}}=-\vec{E_\alpha}\cdot\smash{\vec{N}_{,\beta}}$ associated with $\mathcal{S}^0$ are in general different from the intrinsic curvatures $\smash{{\kappa^0}_{\alpha\beta}}$. This expresses the incompatibility of the intrinsic metric $\mat{G^0}$ of the intrinsic configuration $\mathcal{V}^0$ of the shell. This metric has components
\begin{subequations}\label{eq:GG0}
\begin{align}
&{G^0}_{\zeta\zeta}=\varepsilon^2\bigl({\zeta^0}_{,\zeta}\bigr)^2,&&{G^0}_{\alpha\zeta}={G^0}_{\zeta\alpha}=\varepsilon^2{\zeta^0}_{,\zeta}{\zeta^0}_{,\alpha},
\end{align}
and
\begin{align}
{G^0}_{\alpha\beta}&={g^0}_{\alpha\gamma}\left({\delta^\gamma}_\delta-\varepsilon\zeta^0{\kappa^{0\,\gamma}}_\delta\right)\left({\delta^\delta}_\beta-\varepsilon\zeta^0\smash{{\kappa^0}_\beta}^\delta\right)\nonumber\\
&\quad+\varepsilon^2{\zeta^0}_{,\alpha}{\zeta^0}_{,\beta},
\end{align}
\end{subequations}
that we write down by analogy with Eqs.~\eqref{eq:Gdef}, assuming, as we did in Section~\ref{sec:model}, that there is no intrinsic displacement parallel to the midsurface, $\varsigma^{0\,\alpha}=0$. We emphasise again that, in contrast with the intrinsic metric $\smash{{g^0}_{\alpha\beta}}$, the intrinsic curvatures $\smash{{\kappa^0}_{\alpha\beta}}$ and the intrinsic transverse displacement $\zeta^0$ remain without a direct geometric realisation.

As in Section~\ref{sec:model}, we specify $\zeta^0$ by imposing intrinsic volume conservation. The condition of intrinsic volume conservation reads $\sqrt{\det{\mat{G^0}}}=\sqrt{\det{\mat{G}}}$, or, as we argue in what follows and equivalently, $\det{\tens{F^0}}=1$, where the intrinsic deformation gradient $\tens{F^0}$ is given by Eq.~\eqref{eq:F0bm} below. We shall integrate the differential equation resulting from this condition under the scaling assumptions of shell theory later, and we shall again choose the midsurfaces $\mathcal{S}$, $\tilde{\mathcal{S}}$, and $\mathcal{S}^0$ in such a way that the shell surfaces $\zeta=\pm h^\pm$ and $\tilde{\zeta}=\pm\tilde{h}^\pm$ correspond to $\zeta^0=\pm h^0/2$. We recall that the intrinsic thickness $h^0$ also lacks a direct geometric realisation.

\subsubsection{Calculation of the deformation gradient tensors}
The geometric deformation gradient is $\btilde{\tens{F}}=\operatorname{Grad}\vec{\tilde{r}}$, where, by definition, $\operatorname{Grad}\vec{\tilde{r}}=\vec{\tilde{r}}_{,\alpha}\otimes\vec{r}^{,\alpha}+\vec{\tilde{r}}_{,\zeta}\otimes\vec{r}^{,\zeta}$. Now, from Eqs.~\eqref{eq:Gmetricup}, 
\begin{align}
&\vec{r}^{,\alpha}=G^{\alpha\gamma}\left(g_{\gamma\beta}-\varepsilon\zeta \varkappa_{\gamma\beta}\right)\vec{e^\beta},&&\vec{r}^{,\zeta}=\varepsilon^{-1}\vec{n}.
\end{align}
Using Eqs.~\eqref{eq:drtilde}, it follows that
\begin{subequations}\label{eq:Fg}
\begin{align}
\btilde{\tens{F}}&= \left({\delta^\alpha}_\gamma-\varepsilon\tilde{\zeta}{\smash{\tilde{\kappa}^\alpha}}_{\gamma}+\varepsilon{\tilde{\varsigma}^\alpha}_{;\gamma}\right)G^{\gamma\delta}\left(g_{\delta\beta}-\varepsilon\zeta \varkappa_{\delta\beta}\right)\vec{\tilde{e}_\alpha}\otimes\vec{e^\beta}\nonumber\\
&\qquad +\varepsilon\bigl(\tilde{\zeta}^{,\alpha}+{\tilde{\kappa}^\alpha}_\epsilon\tilde{\varsigma}^\epsilon\bigr)\tilde{g}_{\alpha\delta}G^{\delta\gamma}\left(g_{\gamma\beta}-\varepsilon\zeta \varkappa_{\gamma\beta}\right)\vec{\tilde{n}}\otimes\vec{e^\beta}\nonumber\\
&\qquad +{\tilde{\varsigma}^\alpha}_{,\zeta}\vec{\tilde{e}_\alpha}\otimes\vec{n}+\tilde{\zeta}_{,\zeta}\vec{\tilde{n}}\otimes\vec{n}.
\end{align}
or, in block matrix notation~\footnote{Block matrices will represent tensors with respect to tensor products of a left tangent basis and right dual basis, so can be multiplied without incurring metric factors and traces can be computed without raising or lowering indices. In fact, the indices of the entries of these block matrices (and also of the matrices that appear as their components) will never be raised or lowered. In particular, transposes are understood to be matrix transposes, so do not change the covariant or contravariant nature of indices.},
\begin{align}
\btilde{\tens{F}} &=\left(\begin{array}{c|c}
\tilde{\mat{A}}\mat{H}&\tilde{\svarsigma}_{,\zeta}\\[0.5mm]
\hline
\vphantom{A^{A^{A^A}}}\tilde{\mat{b}}^\top\tilde{\mat{g}}\mat{H}&\tilde{\zeta}_{,\zeta}
\end{array}\right)\quad\left[\tilde{\mathcal{B}}\otimes\mathcal{B}^\ast\right],\label{eq:Fg2}
\end{align}
\end{subequations}
in which the asterisk denotes a dual basis, and where we have introduced 
\begin{align}
{H^\alpha}_\beta=G^{\alpha\gamma}g_{\gamma\delta}{A^\delta}_\beta\quad\text{with }{A^\alpha}_\beta={\delta^\alpha}_\beta-\varepsilon\zeta{\varkappa^\alpha}_\beta,\label{eq:Hdef}
\end{align}
and where we have also let
\begin{align}
{\tilde{A}^\alpha}_\beta&={\delta^\alpha}_\beta-\varepsilon\tilde{\zeta}{\smash{\tilde{\kappa}}^\alpha}_\beta+\varepsilon{\tilde{\varsigma}^\alpha}_{;\beta},&
\tilde{b}^\alpha&=\varepsilon\left(\tilde{\zeta}^{,\alpha}+{\tilde{\kappa}^\alpha}_{\beta}\tilde{\varsigma}^\beta\right).\label{eq:Ab}
\end{align}
By analogy with Eqs.~\eqref{eq:Fg}, the intrinsic deformation gradient tensor is
\begin{subequations}\label{eq:F0bms}
\begin{align}
\tens{F^0}&=\bigl({\delta^\alpha}_\gamma-\varepsilon\zeta^0{\smash{\kappa^{0\,\alpha}}}_\gamma\bigr)G^{\gamma\delta}\left(g_{\delta\beta}-\varepsilon\zeta \varkappa_{\delta\beta}\right)\vec{E_\alpha}\otimes\vec{e^\beta}\nonumber\\
&\quad+\varepsilon\zeta^{0\,\alpha}{g^0}_{\alpha\delta}G^{\delta\gamma}\left(g_{\gamma\beta}-\varepsilon\zeta \varkappa_{\gamma\beta}\right)\vec{N}\otimes\vec{e^\beta}+{\zeta^0}_{,\zeta}\vec{N}\otimes\vec{n}, 
\end{align}
or, in block matrix notation,
\begin{align}
\tens{F^0}&=\left(\begin{array}{c|c}
\mat{A^0}\mat{H}&\mat{0}\\
\hline
\vphantom{A^{A^{A^A}}}{\mat{b^0}}^\top\mat{g^0}\mat{H}&{\zeta^0}_{,\zeta}
\end{array}\right)\quad\left[\mathcal{B}^0\otimes\mathcal{B}^\ast\right]. \label{eq:F0bm}
\end{align}
\end{subequations}
Here we have again assumed that there is no intrinsic displacement parallel to the midsurface, $\varsigma^{0\,\alpha}=0$, and we have introduced
\begin{align}
&{A^{0\,\alpha}}_\beta ={\delta^\alpha}_\beta-\varepsilon\zeta^0{\smash{\kappa^{0\,\alpha}}}_\beta,&&b^{0\,\alpha}=\varepsilon{\zeta^{0\,,\alpha}}.\label{eq:Ab0}
\end{align}

At this stage, we interrupt the computation of the deformation gradient tensors and we discuss the condition of intrinsic volume conservation. From Eq.~\eqref{eq:Gmetricb} and definition~\eqref{eq:Hdef}, $G_{\alpha\beta}=g_{\alpha\gamma}{A^{\gamma}}_{\delta}{A^\delta}_\beta$. Now $\det{\mat{MN}}=\det{\mat{M}}\det{\mat{N}}$ for matrices $\mat{M},\mat{N}$, so, from Eqs.~\eqref{eq:Gmetrica},\begin{subequations}\label{eq:detGG0}
\begin{align}
\det{\mat{G}}=\varepsilon^2 g(\det{\mat{A}})^2, \label{eq:detGG01}
\end{align}
where we recall the definition $g=\det{\mat{g}}$. Similarly, on introducing $g^0=\det{\mat{g^0}}$ and on evaluating the determinant of a block matrix~\cite{* [] [{ Chap. 2.8, pp. 115--119 and Chap. 4.4, pp. 261--267.}] matrix}, Eqs.~\eqref{eq:GG0} yield
\begin{align}
\det{\mat{G^0}}=\varepsilon^2\bigl({\zeta^0}_{,\zeta}\bigr)^2g^0\bigl(\det{\mat{A^0}}\bigr)^2. \label{eq:detGG02}
\end{align}
\end{subequations}
Above, we have claimed that the intrinsic volume conservation condition $\sqrt{\det{\mat{G^0}}}=\sqrt{\det{\mat{G}}}$ is equivalent with the tensorial condition $\det{\tens{F^0}}=1$. Since Eq.~\eqref{eq:F0bm} expresses the intrinsic deformation gradient with respect to a mixed non-orthogonal basis, we shall need the following observation to evaluate the determinant and hence prove our claim:
\begin{proposition}\label{prop1}
Let $\{\vec{e_\alpha}\}$ and $\{\vec{E_\beta}\}$ be right-handed bases with corresponding metrics ${g_{\alpha\beta}=\vec{e_\alpha}\cdot\vec{e_\beta}}$, and $G_{\alpha\beta}=\vec{E_\alpha}\cdot\vec{E_\beta}$, and let $\tens{M}={M^\alpha}_\beta\vec{e_\alpha}\otimes\vec{\smash{E^\beta}}$ be a tensor represented by the matrix $\mat{M}=({M^\alpha}_\beta)$ with respect to \smash{$\{\vec{e_\alpha}\}\otimes\{\vec{E^\beta}\}$}. Let $g=\det{g_{\alpha\beta}}$ and $G=\det{G_{\alpha\beta}}$. Then
\begin{align*}
\det{\tens{M}}=\sqrt{\dfrac{g}{G}}\det{\mat{M}}.
\end{align*}
\end{proposition}
\begin{proof}
Let $\{\vec{X_i}\}$ be the standard Cartesian basis, and write $\vec{e_\alpha}=e_{\alpha i}\vec{X_i}$, $\vec{E_\alpha}=E_{\alpha i}\vec{X_i}$. Let $e=\det{e_{\alpha i}}$, $E=\det{E_{\alpha i}}$. By assumption, $e,E>0$. By definition, $g_{\alpha\beta}=\vec{e_\alpha}\cdot\vec{e_\beta}=e_{\alpha i}e_{\beta i}$ as $\vec{X_i}\cdot\vec{X_j}=\delta_{ij}$. Since $\det{e_{\beta i}}=\det{e_{i \beta}}$, $e^2=g$. Similarly, $E^2=G$. Now 
\begin{align*}
\tens{M}=e_{\alpha i}{M^\alpha}_\beta G^{\beta\gamma}E_{\gamma j}\vec{X_i}\otimes\vec{X_j}, 
\end{align*}
which implies, since $\det{\mat{G}^{-1}}=G^{-1}$, $\det{\tens{M}}=e(\det{\mat{M}})G^{-1}E$. This completes the proof~\cite{* [{This result is doubtless known in the solid mechanical literature: e.g., it appears without proof as Eq.~(5.8) of }] [] yavari12}.
\end{proof}
Since the normal vectors $\vec{n}$ in $\mathcal{B}$ and $\vec{N}$ in $\mathcal{B}^0$ are, by definition, unit vectors perpendicular to the remaining basis vectors, Proposition~\ref{prop1} and Eq.~\eqref{eq:F0bm} yield
\begin{subequations}\label{eq:detF0}
\begin{align}
\det{\tens{F^0}}=\sqrt{\dfrac{g^0}{g}}\det{\mat{F^0}}=\sqrt{\dfrac{g^0}{g}}{\zeta^0}_{,\zeta}\det{\mat{A^0}}\det{\mat{H}}.
\end{align}
Now definition~\eqref{eq:Hdef} implies, since $G_{\alpha\beta}=g_{\alpha\gamma}{A^{\gamma}}_{\delta}{A^\delta}_\beta$, that 
\begin{align}
\det{\mat{H}}=\left[g(\det{\mat{A}})^2\right]^{-1}g\det{\mat{A}}=\dfrac{1}{\det{\mat{A}}}, 
\end{align}
\end{subequations}
Since we assume \smash{${\zeta^0}_{\smash{,\zeta}}>0$}, Eqs.~\eqref{eq:detGG0} and \eqref{eq:detF0} show that \smash{$\sqrt{\det{\mat{G^0}}}=\sqrt{\det{\mat{G}}}\Longleftrightarrow\det{\tens{F^0}}=1$}, as claimed. Because we have written down Eqs.~\eqref{eq:GG0} and \eqref{eq:F0bms} defining the incompatible metric of $\mathcal{V}^0$ and the intrinsic deformation gradient $\tens{F^0}$ by analogy with the corresponding results for the deformation configuration \smash{$\tilde{\mathcal{V}}$}, but have not derived them from an embedding of $\mathcal{V}^0$, it is not \emph{a priori} clear that these expressions are consistent. This is why we needed to show, as we did in Section~\ref{sec:model}, that the expression for $\mat{G^0}$ is consistent with that for $\tens{F^0}$ as far the only use of the former (i.e. intrinsic volume conservation or the definition of the intrinsic volume element) is concerned. Equivalently, intrinsic volume conservation can be imposed without reference to the incompatible metric $\mat{G^0}$; consequently, as also noted in Section~\ref{sec:model}, the volume element $\mathrm{d}V^0$ of $\mathcal{V}^0$ can be also be defined with reference to $\tens{F^0}$ only. 

We now return to the computation of the elastic deformation gradient \smash{$\tens{F}=\btilde{\tens{F}}\bigl(\tens{F^0}\bigr)^{-1}$}. On inverting the block-lower triangular matrix in Eq.~\eqref{eq:F0bm}, we find
\begin{align}
\bigl(\tens{F^0}\bigr)^{-1}=\left(\begin{array}{c|c}
\mat{H}^{-1}\left(\mat{A^0}\right)^{-1}&\mat{0}\\[0.5mm]
\hline
-\dfrac{\vphantom{A^{A^{A^A}}}{\mat{b^0}}^\top\mat{g^0}\left(\mat{A^0}\right)^{-1}}{{\zeta^0}_{,\zeta}}&\dfrac{1}{{\zeta^0}_{,\zeta}}
\end{array}\right)\quad\left[\mathcal{B}\otimes\bigl(\mathcal{B}^0\bigr)^\ast\right].\label{eq:F0inv}
\end{align}
From this and from Eq.~\eqref{eq:Fg2}, we obtain
\begin{align}
\tens{F}&=\left(\begin{array}{c|c}
\left(\tilde{\mat{A}}-{\tilde{\svarsigma}}_{,\zeta^0}{\mat{b^0}}^\top\mat{g^0}\right)\left(\mat{A^0}\right)^{-1}&{\tilde{\svarsigma}}_{,\zeta^0}\\[1.75mm]
\hline
\vphantom{\Bigl(^A}\left(\tilde{\mat{b}}^\top\tilde{\mat{g}}-\tilde{\zeta}_{,\zeta^0}{\mat{b^0}}^\top\mat{g^0}\right)\left(\mat{A^0}\right)^{-1}&\tilde{\zeta}_{,\zeta^0}
\end{array}\right)\quad\left[\tilde{\mathcal{B}}\otimes\bigl(\mathcal{B}^0\bigr)^\ast\right].\label{eq:FBB0}
\end{align}
\subsection{Thin shell theory for large bending deformations}
As in Section~\ref{sec:model}, we assume that the shell is made of an incompressible neo-Hookean material, with energy given by Eq.~\eqref{eq:E}. Eq.~\eqref{eq:P} still provides an expression for the stress tensor $\tens{Q}$, now with respect to $\tilde{\mathcal{B}}\otimes\bigl(\mathcal{B}^0\bigr)^\ast$, and with the deformation gradients $\smash{\btilde{\tens{F}},\tens{F^0},\tens{F}}$ now given by Eqs.~\eqref{eq:Fg2}, \eqref{eq:F0bm}, and \eqref{eq:FBB0}, respectively. Moreover, Eq.~\eqref{eq:Cauchy} still applies.

\subsubsection{Scaling assumptions}
Again as in Section~\ref{sec:model}, we rescale the intrinsic and deformed curvature tensors, $\tens{\skappab^0}={\kappa^{0\,\alpha}}_\beta\vec{E_\alpha}\otimes\vec{E^\beta}$ and $\btilde{\skappab}={\tilde{\kappa}^{\alpha}}_\beta\vec{\tilde{e}_\alpha}\otimes\vec{\tilde{e}^\beta}$, to introduce large bending deformations explicitly and absorb the intrinsic stretching of the midsurface by writing
\begin{align}
&\tens{\skappab^0}=\sqrt{\dfrac{g^0}{g}}\dfrac{\tens{\slambdab^0}}{\varepsilon},&&\btilde{\skappab}=\sqrt{\dfrac{g^0}{g}}\dfrac{\btilde{\slambdab}}{\varepsilon},\label{eq:kscale}
\end{align}
In what follows, we shall need explicit representations of these tensors, $\smash{\slambdab^{\tens0}={\lambda^{0\,\alpha}}_\beta\vec{E_\alpha}\otimes\vec{E^\beta}}$ and $\smash{\btilde{\slambdab}={\tilde{\lambda}^\alpha}_\beta\vec{\tilde{e}_\alpha}\otimes\vec{\tilde{e}^\beta}}$, and shall denote by $\slambda^{\mat0}$ and $\tilde{\slambda}$ the corresponding matrices of components.

Next, we make the standard scaling assumptions of shell theory, that the elastic strains remain small. To this end, we introduce the deformation gradient restricted to the midsurface,
\begin{align}
\tens{f}=\vec{\tilde{e}_\alpha}\otimes\vec{E^\alpha}
\end{align}
First, we require that the shell strains be small: accordingly, we define the shell strain tensor $\tens{E}$ by
\begin{align}
2\varepsilon\tens{E}=\tens{f}^\top\tens{f}-\tens{I}.\label{eq:Etens}
\end{align}
Now $\tens{f}^\top=\vec{E^\alpha}\otimes\vec{\tilde{e}_\alpha}$, so $\tens{f}^\top\tens{f}=\tilde{g}_{\alpha\beta}\vec{E^\alpha}\otimes\vec{E^\beta}=g^{0\,\alpha\gamma}\tilde{g}_{\gamma\beta}\vec{E_\alpha}\otimes\vec{E^\beta}$. Hence, if we set $\tens{E}=\smash{{E^\alpha}_\beta\vec{E_\alpha}\otimes\vec{E^\beta}}$, then
\begin{subequations}
\begin{align}
2\varepsilon{E^\alpha}_\beta=g^{0\,\alpha\gamma}g_{\gamma\beta}-{\delta^\alpha}_\beta\quad\text{or}\quad 2\varepsilon\mat{E}=\bigl(\mat{g^0}\bigr)^{-1}\tilde{\mat{g}}-\mat{I}, \label{eq:Emat}
\end{align}
in equivalent matrix notation. In the calculations that follow, we shall need a consequence of this definition,
\begin{align}
\tilde{\mat{g}}=\mat{g^0}\left(\mat{I}+2\varepsilon\mat{E}\right). \label{eq:uscale}
\end{align}
\end{subequations}
Second, we require that the curvature strains remain small: we therefore introduce two different (scaled) curvature strain tensors,
\begin{align}
\varepsilon\tens{L}=\tens{f}^{-1}\btilde{\slambdab}\tens{f}-\tens{\slambdab^0},&&\varepsilon\tens{K}=\tens{f}^\top\btilde{\slambdab}\tens{f}-\tens{\slambdab^0}.\label{eq:curvtens}
\end{align}
Since $\tens{f}^{-1}=\vec{E_\alpha}\otimes\vec{\tilde{e}^\alpha}$, $\tens{f}^{-1}\btilde{\slambdab}\tens{f}={\tilde{\lambda}^\alpha}_\beta\vec{E_\alpha}\otimes\vec{E^\beta}$, and hence, on writing $\smash{\tens{L}={L^\alpha}_\beta\vec{E_\alpha}\otimes\vec{E^\beta}}$, we find~\footnote{The indices in Eq.~\eqref{eq:Lmat} are raised or lowered with different metrics, $\tilde{\mat{g}}$ and $\mat{g^0}$, which are asymptotically close to each other by Eq.~\eqref{eq:uscale}. Hence taking tensor transposes explicitly by multiplying matrices by the appropriate metrics enables us to impose the asymptotic scaling~\eqref{eq:uscale} during the asymptotic expansion. This is the reason why transposes in our block matrix notation~\cite{Note2} are matrix transposes rather than tensor transposes.}
\begin{align}
\varepsilon{L^\alpha}_\beta={\tilde{\lambda}^\alpha}_\beta-{\lambda^{0\,\alpha}}_\beta\quad\text{or}\quad\varepsilon\mat{L}=\tilde{\slambda}-\slambda^{\mat{0}}. \label{eq:Lmat}
\end{align}
Similarly, $\tens{f}^\top\btilde{\slambdab}\tens{f}=g^{0\,\alpha\gamma}\tilde{g}_{\gamma\delta}{\tilde{\lambda}^\delta}_\beta\vec{E_\alpha}\otimes\vec{E^\beta}$, whence, on letting $\smash{\tens{K}={K^\alpha}_\beta\vec{E_\alpha}\otimes\vec{E^\beta}}$ and from Eqs.~\eqref{eq:Emat} and \eqref{eq:Lmat},
\begin{align}
{K^\alpha}_\beta={L^\alpha}_\beta+2{E^\alpha}_\gamma{\lambda^{0\,\gamma}}_\beta+O(\varepsilon)\;\;\text{or}\;\;\mat{K}=\mat{L}+2\mat{E\slambda^0}+O(\varepsilon).\label{eq:Kmat}
\end{align}

These scalings and definitions are consistent with the scalings~\eqref{eq:curvscalings} and the definitions \eqref{eq:shellstr} and \eqref{eq:curvstr} of the shell and curvature strains for the axisymmetric deformations analysed in Section~\ref{sec:model}. Indeed, for these axisymmetric deformations,
\begin{align}
&\mat{g}=\left(\begin{array}{cc}
1&0\\0&r^2               
\end{array}\right),\;\tilde{\mat{g}}=\left(\begin{array}{cc}
\tilde{f}_s^{\,2}&0\\0&r^2\tilde{f}_\phi^{\,2}               
\end{array}\right),\;\mat{g^0}=\left(\begin{array}{cc}
\bigl(f_s^0\bigr)^2&0\\0&r^2\bigl(f_\phi^0\bigr)^2               
\end{array}\right),\label{eq:gaxisyms}
\end{align}
from Eqs.~\eqref{eq:undefmetric}, \eqref{eq:defmidmetric}, and \eqref{eq:intmidmetric}. In particular, $\sqrt{g^0/g}=f_s^0f_{\smash\phi}^0$. Moreover, Eq.~\eqref{eq:uscale} yields
\begin{subequations}
\begin{align}
\tilde{f}_s&=f_s^0\sqrt{1+2\varepsilon {E^s}_s}=f_s^0\bigl(1+\varepsilon{E^s}_s\bigr)+O\bigl(\varepsilon^2\bigr),\\
\tilde{f}_\phi&=f_\phi^0\sqrt{1+2\varepsilon{E^{\smash\phi}}_\phi}=f_s^0\left(1+\varepsilon{E^\phi}_\phi\right)+O\bigl(\varepsilon^2\bigr),
\end{align}
while ${E^s}_\phi={E^\phi}_s=0$. Thus, identifying
\begin{align}
&E_s ={E^s}_s,&&E_\phi={E^\phi}_\phi,\label{eq:eaxisym}
\end{align}
\end{subequations}
we conclude that Eqs.~\eqref{eq:shellstr} are consistent with Eq.~\eqref{eq:uscale} at leading order, i.e. at the order to which the shell theory will be valid.

Direct computation relates the components of $\tilde{\slambda}$ to the principal curvatures of $\tilde{\mathcal{S}}$ defined in Eqs.~\eqref{eq:kappas}, viz.
\begin{subequations}
\begin{align}
&{\tilde{\lambda}^s}_s=\dfrac{\tilde{\kappa}_s}{\varepsilon f_s^0f_\phi^0}, &&&{\tilde{\lambda}^\phi}_\phi=\dfrac{\tilde{\kappa}_\phi}{\varepsilon f_s^0f_\phi^0},
\end{align}
while ${\tilde{\lambda}^s}_{\phi}={\tilde{\lambda}^\phi}_s=0$. Hence Eqs.~\eqref{eq:curvscalings} and \eqref{eq:curvstr} are consistent with Eqs.~\eqref{eq:kscale} and \eqref{eq:Lmat} if we identify
\begin{align}
&\lambda^0_s={\lambda^{0\,s}}_s,&&\varepsilon\lambda^0_\phi={\lambda^{0\,\phi}}_\phi,&&L_s={L^s}_s,&&L_\phi={L^\phi}_\phi, \label{eq:laxisym}
\end{align}
\end{subequations}
with the off-diagonal components vanishing. However, comparing Eqs.~\eqref{eq:LsLphi} and \eqref{eq:Kmat} shows that the alternative curvature strains defined here are different from those defined in Eqs.~\eqref{eq:KsKphi}: ${K^s}_s=L_s+2E_s\lambda_s^0+O(\varepsilon)\not=L_s+E_s\lambda_s^0+O(\varepsilon)=K_s$, using Eq.~\eqref{eq:eta}. We are not aware of a tensorial representation of the alternative curvature strains introduced in Eqs.~\eqref{eq:KsKphi} and that vanish for pure stretching deformations.

As in the axisymmetric calculations in Section~\ref{sec:model}, it will turn out to be convenient to scale the displacements parallel and perpendicular to the midsurfaces by absorbing the intrinsic stretching of the midsurface. We therefore introduce scaled variables
\begin{align}
&Z^0=\sqrt{\dfrac{g^0}{g}}\zeta^0, &&Z=\sqrt{\dfrac{g^0}{g}}\tilde{\zeta},&&\mat{S}=\sqrt{\dfrac{g^0}{g}}\tilde{\svarsigma}.\label{eq:scaledvars}
\end{align}

\subsubsection{Boundary and incompressibility conditions}
As in Section~\ref{sec:model}, we solve the Cauchy equation~\eqref{eq:Cauchy} subject to the incompressibility condition $\det{\tens{F}}=1$ and subject to force-free boundary conditions.

Again as in Section~\ref{sec:model}, these boundary conditions on the shell surfaces read $\tens{Q}^\bpm\vec{n^\pm}=\vec{0}$, where $\tens{Q}^\bpm$ are evaluated on the surfaces $\zeta=\pm h^\pm$ of $\mathcal{V}$. The normal vectors $\vec{n^\pm}$ to these undeformed shell surfaces are given by Eqs.~\eqref{eq:NA}, which yield the expansion
\begin{align}
\vec{n^\pm}=\vec{n}\mp\varepsilon\dfrac{{h^\pm}_{,\alpha}}{g}\vec{e^\alpha}+O\bigl(\varepsilon^2\bigr). \label{eq:NAexp}
\end{align}

The deformation gradient is given in Eq.~\eqref{eq:FBB0} with respect to the mixed basis $\tilde{\mathcal{B}}\otimes\bigl(\mathcal{B}^0\bigr)^\ast$. In what follows, we shall therefore use Proposition~\ref{prop1} to evaluate the tensorial incompressibility condition~$\det{\tens{F}}=1$.
\subsubsection{Intrinsic volume conservation}
We now impose volume conservation of the intrinsic configuration of the shell compared to the undeformed configuration. We need one preliminary result:
\begin{lemma}\label{lemma1}
Let $\mat{M}$ be a $2\times 2$ matrix, and $x$ be a scalar. Then
\begin{align*}
\det{\bigl(\mat{I}+x\mat{M}\bigr)}=1+x\tr{\mat{M}}+x^2\det{\mat{M}}. 
\end{align*}
\end{lemma}
\begin{proof}
By direct computation,
\begin{align*}
&\det{\left(\begin{array}{cc}
1+x M_{11}&xM_{12}\\
x M_{21}&1+x M_{22}
\end{array}\right)}\\
&\qquad=1+x(M_{11}+M_{22})+x^2(M_{11}M_{22}-M_{12}M_{21}),
\end{align*}
which proves the claim.
\end{proof}
Volume conservation between the undeformed and intrinsic configurations of the shell requires equality of the volume elements, \smash{$\sqrt{\det{\mat{G}}}=\sqrt{\det{\mat{G^0}}}$}. Now, from definition~\eqref{eq:Hdef}, ${A^\alpha}_\beta={\delta^\alpha}_\beta+O(\varepsilon)$, and so Eq.~\eqref{eq:detGG01} yields
\begin{subequations}\label{eq:IVC}
\begin{align}
\sqrt{\det{\mat{G}}}=\varepsilon\sqrt{g}+O\bigl(\varepsilon^2\bigr). \label{eq:IVC1}
\end{align}
Moreover, from Eqs.~\eqref{eq:Ab0} and \eqref{eq:detGG02} with the scalings introduced above and invoking Lemma~\ref{lemma1}, we find
\begin{align}
\sqrt{\det{\mat{G^0}}}&=\varepsilon\left(\sqrt{\dfrac{g}{g^0}}{Z^0}_{,\zeta}\right)\left\{\sqrt{g^0}\left[1-2\mathcal{H}^0Z^0+\mathcal{K}^0\bigl(Z^0\bigr)^2\right]\right\}\nonumber\\
&\quad+O\bigl(\varepsilon^2\bigr),\label{eq:IVC2}
\end{align}
\end{subequations}
wherein $\mathcal{H}^0=\tfrac{1}{2}{\smash{\lambda^0}_\alpha}^{\alpha}$ and $\mathcal{K}^0=\det{{\smash{\lambda^0}_\alpha}^{\beta}}$, which we think of as (scaled) intrinsic mean and Gaussian curvatures~\cite{kreyszig}. Since these are not associated with an embedding of $\mathcal{S}^0$ into three-dimensional Euclidean space, we must establish their properties from first principles, based on the assumed symmetry of the intrinsic metric and intrinsic curvature tensor. The following results are undoubtedly folklore:
\begin{proposition}
If $\mat{M}$ is a symmetric matrix and~$\mat{N}$ is a positive-definite symmetric matrix, then $\mat{MN}$ has real eigenvalues. \label{prop2}
\end{proposition}
\begin{proof}
Since $\mat{N}$ is positive-definite and symmetric, it has a symmetric square root $\mat{N}^{1/2}$~\cite{* [] [{ Chap.~1.3, pp. 44--57, Chap.~4.1, pp.~169--176, and Chap.~7.2, pp.~402--411.}] linalg}. Now 
\begin{align*}
\mat{MN}=\bigl(\mat{N}^{1/2}\bigr)^{-1}\bigl(\mat{N}^{1/2}\mat{M}\mat{N}^{1/2}\bigr)\mat{N}^{1/2},
\end{align*}
so $\mat{MN}$ is similar to and hence has the same eigenvalues~\cite{linalg} as $\mat{N}^{1/2}\mat{M}\mat{N}^{1/2}$. Since $\mat{M}$ and $\mat{N}^{1/2}$ are symmetric, so is $\mat{N}^{1/2}\mat{M}\mat{N}^{1/2}$, which therefore has real eigenvalues~\cite{linalg}. Hence $\mat{MN}$ has real eigenvalues, too, as claimed.
\end{proof}
\begin{corollary}
If $\mat{M}$ is a symmetric $2\times 2$ matrix and~$\mat{N}$ is a positive-definite symmetric $2\times 2$ matrix, then\label{cor0}
\begin{align*}
[\tr{(\mat{MN})}]^2\geq4\det{(\mat{MN})}. 
\end{align*}
\end{corollary}
\begin{proof}
By Proposition~\ref{prop2}, the $2\times 2$ matrix $\mat{MN}$ has real eigenvalues $\mu_1,\mu_2$. Hence
\begin{align*}
[\tr{(\mat{MN})}]^2\!-\!4\det{(\mat{MN})}\!=\!(\mu_1\!+\!\mu_2)^2\!-\!4\mu_1\mu_2\!=\!(\mu_1\!-\!\mu_2)^2\geq 0, 
\end{align*}
which completes the proof.
\end{proof}
Now ${\smash{\lambda^0}_\alpha}^{\beta}={\lambda^0}_{\alpha\gamma}g^{0\,\gamma\beta}$. Since ${\kappa^0}_{\alpha\beta}$ is symmetric, so is its rescaling $\smash{{\lambda^0}_{\alpha\beta}}$. As $\smash{{g^0}_{\alpha\beta}}$ is symmetric and positive definite, so is its inverse $\smash{g^{0\,\alpha\beta}}$. Hence the conditions of Corollary~\ref{cor0} are satisfied; it implies the inequality \smash{$\bigl(\mathcal{H}^0\bigr)^2\geq\mathcal{K}^0$}.

Next, integrating the differential equation for $Z^0(\zeta)$ resulting from Eqs.~\eqref{eq:IVC} and imposing $Z^0=0$ at $\zeta=0$, we find
\begin{align}
Z^0-\mathcal{H}^0\bigl(Z^0\bigr)^2+\dfrac{\mathcal{K}^0}{3}\bigl(Z^0\bigr)^3=\zeta.\label{eq:Z0zeta2}
\end{align}
Since Eqs.~\eqref{eq:IVC} neglect $O\bigl(\varepsilon^2\bigr)$ corrections, this result holds at leading order only.

\begin{figure}[b]
\includegraphics{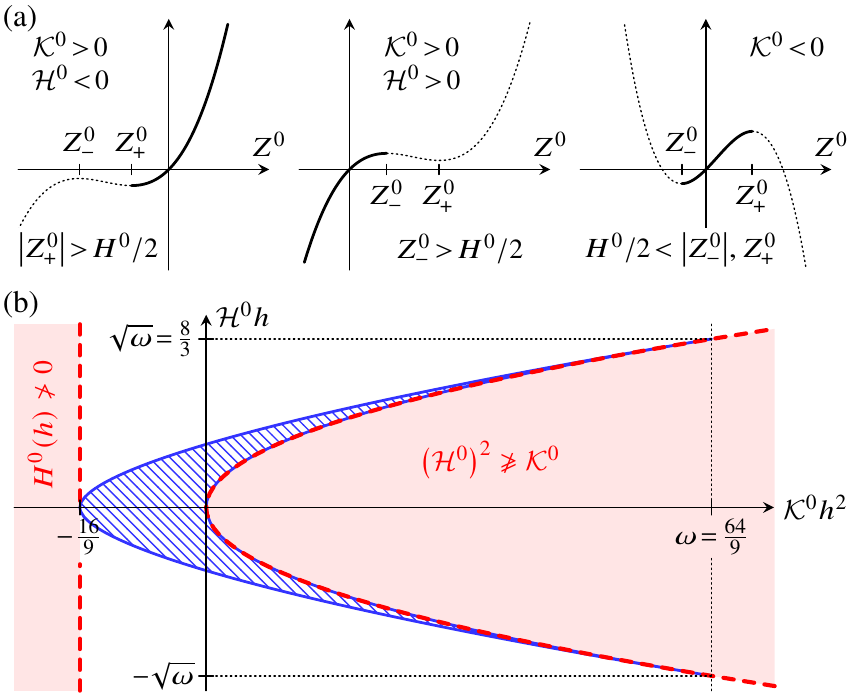}
\caption{Intrinsic volume conservation. (a)~Plot of $\smash{\zeta\bigl(Z^0\bigr)}$ defined in Eq.~\eqref{eq:Z0zeta2} for the cases $\mathcal{K}^0>0$, $\mathcal{H}^0<0$; $\mathcal{K}^0>0$, $\mathcal{H}^0>0$; $\mathcal{K}^0<0$. The positions of the turning points at \smash{$Z^0=Z^0_\pm$} are indicated, and $\zeta\bigl(Z^0\bigr)$ must increase monotonically for $\bigl|Z^0\bigr|<H^0/2$. This condition excludes the dotted parts of the graphs. (b)~Intrinsic volume conservation in $\bigl(\mathcal{K}^0h^2,\mathcal{H}^0h\bigr)$ space: conservation of intrinsic volume is only possible within the region of parameter space enclosed by the solid curve, in which $-16/9<h^2\mathcal{K}^0<\omega=64/9$ and ${h\bigl|\mathcal{H}^0\bigr|<\sqrt{\omega}=8/3}$. The dashed lines delimit the regions of parameter space excluded by the inequality \smash{$\bigl(\mathcal{H}^0\bigr)^2\geq\mathcal{K}^0$} and the condition that Eq.~\eqref{eq:H0} have a positive real solution.}\label{fig6} 
\end{figure}

We recall that, by definition, the shell surfaces are at $\zeta^0=\pm h^0/2$ in the intrinsic configuration, and at $\zeta=\pm h^\pm$ in the undeformed configuration, so that $h^++h^-=h$ is the undeformed thickness of the cell sheet. On defining \smash{\mbox{$H^0=h^0\sqrt{g^0/g}$}}, so that the shell surfaces are at $Z^0=\pm H^0/2$ in the intrinsic configuration, Eq.~\eqref{eq:Z0zeta2} yields
\begin{subequations}
\begin{align}
h^\pm=\dfrac{H^0}{2}\left[1\mp\dfrac{\mathcal{H}^0}{2}H^0+\dfrac{\mathcal{K}^0}{12}\bigl(H^0\bigr)^2\right], 
\end{align}
whence
\begin{align}
h=h^++h^-=H^0+\dfrac{\mathcal{K}^0}{12}\bigl(H^0\bigr)^3,\label{eq:H0}
\end{align}
\end{subequations}
which is a depressed cubic equation for $H^0(h)$ that can be solved in closed form. In particular, Eq.~\eqref{eq:H0} has a unique positive real solution if $\mathcal{K}^0>0$, but has no positive real solution if $h^2\mathcal{K}^0<-16/9$. If $0>h^2\mathcal{K}^0>-16/9$, two positive real solutions exist; by continuity, the smaller must be chosen.

More generally, we require that $\zeta$ increase with $Z^0$, for \smash{$\bigl|Z^0\bigr|\leq H^0/2$}. As \smash{$\bigl(\mathcal{H}^0\bigr)^2\geq\mathcal{K}^0$}, the cubic in Eq.~\eqref{eq:Z0zeta2} has two turning points~\figref{fig6}{a}, at $Z^0=Z^0_\pm$, where explicit expressions for $Z^0_-\leq Z^0_+$ in terms of $\mathcal{K}^0,\mathcal{H}^0$ can be found by solving a quadratic equation. The requirement that $\zeta$ increase with $Z^0$ translates to inequalities $Z^0_\pm\gtrless H^0(h)/2$ depending on the signs of $\mathcal{K}^0,\mathcal{H}^0$~\figref{fig6}{a}. These inequalities involving $h,\mathcal{H}^0,\mathcal{K}^0$ only depend on $\mathcal{H}^0h$ and $\mathcal{K}^0h^2$, since the curvatures can be nondimensionalised with $h$. The inequalities can then be solved numerically to determine the region in $\bigl(\mathcal{K}^0h^2,\mathcal{H}^0h\bigr)$ parameter space for which intrinsic volume conservation is possible~\figref{fig6}{b}. In particular, \textfigref{fig6}{b} shows that intrinsic volume conservation requires $-16/9\leq\mathcal{K}^0h^2\leq \omega$ and ${\bigl|\mathcal{H}^0h\bigr|\leq\smash{\sqrt{\omega}}}$ , where $\omega$ is a numerical constant. An expression for the boundary of this region can also be determined in closed form using \textsc{Mathematica} (Wolfram, Inc.); this can be used to show that $\omega=64/9$. 

For the axisymmetric deformations considered in Section~\ref{sec:model}, $\mathcal{K}^0={\lambda^{0\,s}}_s\smash{{\lambda^{0\,\phi}}_\phi}=O(\varepsilon)$ from Eqs.~\eqref{eq:laxisym}. For $\mathcal{K}^0=0$, the condition derived here is $\bigl|h\mathcal{H}^0\bigr|\leq 1$. But, using Eqs.~\eqref{eq:laxisym} again, \smash{$h\mathcal{H}^0=h\lambda_s^0/2+O(\varepsilon)=\eta+O(\varepsilon)$} on recalling definition~\eqref{eq:eta}, and so this condition is equivalent, as expected, to the condition $|\eta|\leq 1$ found in Section~\ref{sec:model}.

\subsubsection{Expansion of the boundary and incompressibility conditions}
To avoid drowning in a bath of indices, we shall use the block matrix notation for tensors~\cite{Note2} introduced above in the expansions that follow below. This means, however, that some care needs to be taken over distinguishing between tensor and matrix transposes and, in particular, over the bases with respect to which transposes of block matrices represent tensor transposes~\cite{Note2}. We shall use the following results repeatedly:
\begin{proposition}\label{prop3}
Let $\mathcal{B}$ and $\mathcal{B}'$ be bases of three-dimensional space with corresponding metrics $\mat{g}$, $\mat{G}$. A tensor $\tens{M}$ is represented by the matrix $\mat{M}$ with respect to $\mathcal{B}\otimes(\mathcal{B}')^\ast$. Then $\tens{M}^\top$ is represented by $\mat{G}^{-1}\mat{M}^\top\mat{g}$ with respect to $\mathcal{B}'\otimes\mathcal{B}^\ast$.
\end{proposition}
\begin{proof}
Let $\mathcal{B}=\{\vec{e_\alpha}\}$, $\mathcal{B}'=\{\vec{E_\alpha}\}$, so that $\tens{M}={M^{\alpha}}_\beta\vec{e_\alpha}\otimes\vec{E^\beta}$. By definition, $\smash{\tens{M}^\top={M^\beta}_{\alpha}\vec{E^\alpha}\otimes\vec{e_\beta}=G^{\alpha\gamma}{M^{\delta}}_\gamma g_{\delta\beta}\vec{E_\alpha}\otimes\vec{e^\beta}}$, as claimed.
\end{proof}
\begin{corollary}\label{cor2}
Let $\mathcal{B}=\{\vec{e_\alpha}\}\cup\{\vec{n}\}$ and $\mathcal{B}'=\{\vec{E_\alpha}\}\cup\{\vec{N}\}$ be bases of three-dimensional space, where $\vec{n},\vec{N}$ are the respective unit normals to the planes spanned by $\{\vec{e_\alpha}\}$, $\{\vec{E_\alpha}\}$. Let the metrics $\mat{g}$, $\mat{G}$ have components $g_{\alpha\beta}=\vec{e_\alpha}\cdot\vec{e_\beta}$, $G_{\alpha\beta}=\vec{E_\alpha}\cdot\vec{E_\beta}$. If $\tens{M}$ is a tensor such that
\begin{align*}
\tens{M}=\left(\begin{array}{c|c}
\mat{A}&\mat{b}\\\hline\mat{c}^\top&d                
\end{array}\right)\quad\left[\mathcal{B}\otimes\left(\mathcal{B}'\right)^\ast\right], 
\end{align*}
then
\begin{align*}
\tens{M}^\top=\left(\begin{array}{c|c}
\mat{G}^{-1}\mat{A}^\top\mat{g}&\mat{G}^{-1}\mat{c}\\\hline\mat{b}^\top\mat{g}&d                
\end{array}\right)\quad\left[\mathcal{B}'\otimes\mathcal{B}^\ast\right]. 
\end{align*}
\end{corollary}
\begin{proof}
Proposition~\ref{prop3} implies that $\tens{M}^\top$ is represented, with respect to $\mathcal{B}'\otimes\mathcal{B}^\ast$, by
\begin{align*}
\left(\begin{array}{c|c}
\mat{G}&\mat{0}\\\hline\mat{0}^\top&1                
\end{array}\right)^{-1}\left(\begin{array}{c|c}
\mat{A}&\mat{b}\\\hline\mat{c}^\top&d                
\end{array}\right)^\top \left(\begin{array}{c|c}
\mat{g}&\mat{0}\\\hline\mat{0}^\top&1                
\end{array}\right)=\left(\begin{array}{c|c}
\mat{G}^{-1}\mat{A}^\top\mat{g}&\mat{G}^{-1}\mat{c}\\\hline\mat{b}^\top\mat{g}&d                
\end{array}\right),
\end{align*}
which completes the proof.
\end{proof}

To expand the boundary and incompressibility conditions, we posit, analogously to Eqs.~\eqref{eq:expansions},
\begin{align}
&Z=Z_{(0)}+\varepsilon Z_{(1)}+O\bigl(\varepsilon^2\bigr) ,&&\mat{S}=\mat{S_{(0)}}+O(\varepsilon).\label{eq:expansionsA}
\end{align}
\paragraph*{Expansion at order $O(1)$.} On inserting the rescalings~\eqref{eq:scaledvars} into Eqs.~\eqref{eq:Ab} and \eqref{eq:Ab0}, we obtain
\begin{subequations}\label{eq:A0AA}
\begin{align}
&\mat{A^0}=\mat{I}-Z^0\mat{\slambda^0},&&\tilde{\mat{A}}=\tilde{\mat{A}}_{\mat{(0)}}+O(\varepsilon)&&\text{with }\tilde{\mat{A}}_{\mat{(0)}}=\mat{I}-Z_{(0)}\mat{\slambda^0},\label{eq:AA}
\end{align}
and
\begin{align}
\mat{b^0}&=O(\varepsilon),&\tilde{\mat{b}}&=\mat{\slambda^0}\mat{S_{(0)}}+O(\varepsilon),\label{eq:bb0}
\end{align}
\end{subequations}
and thence, from Eq.~\eqref{eq:FBB0},
\begin{align}
\tens{F}=\left(\begin{array}{c|c}
\mat{B}&\mat{v}\\
\hline
\mat{w}^\top&c
\end{array}\right)+O(\varepsilon),\label{eq:F0A}%\quad\left[\tilde{\mathcal{B}}\otimes\bigl(\mathcal{B}^0\bigr)^\ast\right]
\end{align}
where, with dashes now denoting differentiation with respect to $Z^0$,
\begin{align}
&\mat{B}=\tilde{\mat{A}}_{\mat{(0)}}\bigl(\mat{A^0}\bigr)^{-1},\;\mat{v}=\mat{S'_{\smash{(0)}}},\;\mat{w}=\bigl(\mat{A^0}\bigr)^{-\top}\mat{g^0}\mat{\slambda^0}\mat{S_{(0)}},\;c=Z_{(0)}',\label{eq:Fdefs}
\end{align}
since $\tilde{\mat{g}}=\mat{g^0}+O(\varepsilon)$ from Eq.~\eqref{eq:uscale}. Recalling the definitions $\tilde{g}=\det{\tilde{\mat{g}}}$, ${g^0=\det{\mat{g^0}}}$ introduced earlier, this also implies $\tilde{g}/g^0=1+O(\varepsilon)$. Using Proposition~\ref{prop1} and on computing the determinant of the block matrix~\cite{matrix} in Eq.~\eqref{eq:F0A}, the incompressibility condition thus becomes
\begin{align}
1=\det{\tens{F}}=\left(\det{\mat{B}}\right)\left(c-\mat{w}^\top\mat{B}^{-1}\mat{v}\right)+O(\varepsilon). \label{eq:detF0A}
\end{align}
Next, on substituting the first of Eqs.~\eqref{eq:bb0} into Eq.~\eqref{eq:F0inv} and using Corollary~\ref{cor2},
\begin{align}
\bigl(\tens{F^0}\bigr)^{-\top}=\left(\begin{array}{c|c}
O(1)&O(\varepsilon)\\[0.5mm]
\hline
\vphantom{A^{A^{A^A}}}O(1)&\bigl({\zeta^0}_{,\zeta}\bigr)^{-1}
\end{array}\right).\label{eq:F0invT}%\quad\left[\mathcal{B}^0\otimes\mathcal{B}^\ast\right]
\end{align}
Moreover, Eqs.~\eqref{eq:Hdef} yield $\mat{H}=\mat{I}+O(\varepsilon)$ using Eqs.~\eqref{eq:Gmetric}, so, on substituting Eqs.~\eqref{eq:A0AA} into Eq.~\eqref{eq:Fg2}, and using definitions~\eqref{eq:Fdefs},
\begin{subequations}
\begin{align}
\btilde{\tens{F}}=\left(\begin{array}{c|c}
\mat{BA^0}&{\zeta^0}_{,\zeta}\mat{v}\\[0.5mm]
\hline
\mat{w}^\top\mat{A^0}&{\zeta^0}_{,\zeta}c
\end{array}\right) +O(\varepsilon).%\quad\left[\tilde{\mathcal{B}}\otimes\mathcal{B}^\ast\right]
\end{align}
Hence, using further properties of block matrices~\cite{matrix} and, again, $\tilde{\mat{g}}=\mat{g^0}+O(\varepsilon)$ and Corollary~\ref{cor2},
\begin{align}
\btilde{\tens{F}}^{-\top}\!\!=\!\left(\begin{array}{c|c}
O(1)&-\bigl({\zeta^0}_{,\zeta}\bigr)^{-1}\bigl(\mat{g^0}\bigr)^{-1}\mat{B}^{-\top}\mat{w}\left(c\!-\!\mat{w}^\top\mat{B}^{-1}\mat{v}\right)^{-1}\!\!\\[0.5mm]
\hline
\vphantom{A^{A^{A^A}}}O(1)&\bigl({\zeta^0}_{,\zeta}\bigr)^{-1}\left(c-\mat{w}^\top\mat{B}^{-1}\mat{v}\right)^{-1}
\end{array}\right)\!+\!O(\varepsilon).\label{eq:FginvT}%\nonumber\\&\hspace{47mm}\left[\tilde{\mathcal{B}}\otimes\mathcal{B}^\ast\right]
\end{align}
\end{subequations}
We now write, as we have done previously in Eqs.~\eqref{eq:Qpexpansions}, 
\begin{align}
&\tens{Q}=\tens{Q_{(0)}}+\varepsilon\tens{Q_{(1)}}+O\bigl(\varepsilon^2\bigr),&&p=p_{(0)}+O(\varepsilon).
\end{align}
Inserting Eqs.~\eqref{eq:F0A}, \eqref{eq:F0invT}, and \eqref{eq:FginvT} into definition~\eqref{eq:P}, we obtain
\begin{align}
\hspace{-2mm}\tens{Q_{(0)}}\vec{n}=\bigl({\zeta^0}_{,\zeta}\bigr)^{-1}\left(\begin{array}{c}
\mat{v}+p_{(0)}\bigl(\mat{g^0}\bigr)^{-1}\mat{B}^{-\top}\mat{w}\left(c-\mat{w}^\top\mat{B}^{-1}\mat{v}\right)^{-1}\\[0.5mm]
\hline
\vphantom{A^{A^{A^A}}}c-p_{(0)}\left(c-\mat{w}^\top\mat{B}^{-1}\mat{v}\right)^{-1}\label{eq:Q0NA}
\end{array}\right). 
\end{align}
Now, as in Section~\ref{sec:model}, the governing equation \eqref{eq:Cauchy} of three-dimensional elasticity is, at leading order, \smash{$(\tens{Q_{(0)}}\vec{n})_{,\zeta}=\vec{0}$}, and hence $\tens{Q_{(0)}}\vec{n}$ is independent of $\zeta$. The boundary conditions therefore become $\vec{0}=\tens{Q^\bpm}\vec{n^\pm}=\tens{Q_{(0)}}\vec{n}+O(\varepsilon)$, where we have used Eq.~\eqref{eq:NAexp}. It follows that $\smash{\tens{Q_{(0)}}}\vec{n}\equiv\vec{0}$ as in Section~\ref{sec:model}.

From Eqs.~\eqref{eq:Fdefs}, $\smash{\mat{w}^\top\mat{B}^{-1}=\mat{S}^\top_{\smash{\mat{(0)}}}\mat{D}}$ with $\mat{D}=\smash{\bigl(\mat{\slambda^0}\bigr)^\top}\mat{g^0}\tilde{\mat{A}}^{-1}_{\smash{\mat{(0)}}}$, so that $\mat{B}^{-\top}\mat{w}=\mat{D}^{\top}\mat{S_{(0)}}$. Eqs.~\eqref{eq:detF0A} and \eqref{eq:Q0NA} then yield the leading-order incompressibility and boundary conditions,
\begin{subequations}\label{eq:loA}
\begin{align}
Z_{(0)}'-\mat{S}^\top_{\mat{(0)}}\mat{D}\mat{S'_{(0)}}&=\left(\det{\mat{B}}\right)^{-1},
\end{align}
and hence
\begin{align}
\mat{S'_{(0)}}\!+p_{(0)}(\det{\mat{B}})\bigl(\mat{g^0}\bigr)^{-1}\mat{D}^\top\mat{S_{(0)}}&=\mat{0},&Z_{(0)}'\!-p_{(0)}(\det{\mat{B}})&=0.\label{eq:dS0dZ0}
\end{align}
\end{subequations}
In particular, noting that $\mat{S'_{\smash{\mat{(0)}}}}\hspace{-2mm}^\top\mat{D}^\top\mat{S_{(0)}}=\mat{S}^{\top}_{\smash{\mat{(0)}}}\mat{D}\mat{S'_{\smash{\mat{(0)}}}}$ since this expression is a scalar,
\begin{align}
\mat{S'_{\smash{\mat{(0)}}}}\hspace{-2mm}^\top\mat{g^0}\mat{S'_{\smash{\mat{(0)}}}}\!&=-p_{(0)}(\det{\mat{B}})\mat{S'_{\smash{\mat{(0)}}}}\hspace{-2mm}^\top\mat{D}^\top\mat{S_{(0)}}\!\!=\!-p_{(0)}(\det{\mat{B}})\mat{S}^{\top}_{\smash{\mat{(0)}}}\mat{D}\mat{S'_{\smash{\mat{(0)}}}}\nonumber\\
&=-p_{(0)}(\det{\mat{B}})\left[Z_{\smash{(0)}}'\!-(\det{\mat{B}})^{-1}\right]=p_{(0)}\!-\bigl(Z_{\smash{(0)}}'\bigr)^2.\label{eq:invA}
\end{align}
Moreover, from Eqs.~\eqref{eq:A0AA} and definition~\eqref{eq:Fdefs} and using Lemma~\ref{lemma1}, we obtain
\begin{align}
\det{\mat{B}}=\dfrac{\det{\tilde{\mat{A}}_{\mat{(0)}}}}{\det{\mat{A^0}}}=\dfrac{1-2\mathcal{H}^0Z_{(0)}+\mathcal{K}^0\bigl(Z_{(0)}\bigr)^2}{1-2\mathcal{H}^0Z^0+\mathcal{K}^0\bigl(Z^0\bigr)^2}. 
\end{align}
Substituting in the second of Eqs.~\eqref{eq:dS0dZ0} and integrating,
\begin{align}
\tanh^{-1}{\dfrac{\mathcal{K}^0Z_{(0)}-\mathcal{H}^0}{\sqrt{\bigl(\mathcal{H}^0\bigr)^2-\mathcal{K}^0}}}=p_{(0)}\tanh^{-1}{\dfrac{\mathcal{K}^0Z^0-\mathcal{H}^0}{\sqrt{\bigl(\mathcal{H}^0\bigr)^2-\mathcal{K}^0}}}+t,\label{eq:Z0int}
\end{align}
in which $t$ is a constant of integration; the singular cases ${\mathcal{K}^0=0}$, $\mathcal{K}^0=\mathcal{H}^0=0$, or $\mathcal{K}^0=\smash{\bigl(\mathcal{H}^0\bigr)^2}$ can be dealt with similarly, but we will not discuss these in detail.

Next, by definition, on the midsurface $Z^0=0$, we have $Z_{(0)}=0$ and $\mat{S_{(0)}}=\mat{0}$. Thus $\det{\mat{B}}=1$ on $Z^0=0$, and hence, successively from Eqs.~\eqref{eq:loA}, $Z_{\smash{(0)}}'=0$, $\mat{S'_{\smash{\mat{(0)}}}}=\mat{0}$ on $Z^0=0$, and hence $p_{(0)}=1$ (which is constant). Then taking \smash{$Z^0=Z_{(0)}=0$} in Eq.~\eqref{eq:Z0int} yields $t=0$; the same equation then immediately yields \smash{$Z_{(0)}\equiv Z^0$}. Finally, Eq.~\eqref{eq:invA} yields $\mat{S'_{\smash{\mat{(0)}}}}\hspace{-2mm}^\top\mat{g^0}\mat{S'_{\smash{\mat{(0)}}}}=0$, so $\mat{S'_{\smash{(0)}}}\equiv\mat{0}$ since $\mat{g^0}$ is positive definite. Now $\mat{S_{(0)}}=\mat{0}$ on $Z^0=0$, so this implies that $\mat{S_{(0)}}\equiv\mat{0}$, which proves the Kirchhoff ``hypothesis''~\cite{audoly} for general large bending deformations.

For axisymmetric deformations, this argument provides an alternative to the direct integration of the leading-order equations in Section~\ref{sec:model}.

\paragraph*{Expansion at order $O(\varepsilon)$.} We now expand further. In particular, extending Eqs.~\eqref{eq:AA}, we find
\begin{align}
\tilde{\mat{A}}=\mat{A^0}-\varepsilon\bigl(Z^0\mat{L}+Z_{(1)}\mat{\slambda^0}\bigr)+O\bigl(\varepsilon^2\bigr). 
\end{align}
The leading-order solution also shows that $\mat{b^0}$, $\tilde{\mat{b}}$, $\tilde{\svarsigma}$ are all at the most of order $O(\varepsilon)$, whence
\begin{align}
\tens{F}=\mat{I}+\varepsilon\left(\begin{array}{c|c}
-\bigl(Z^0\mat{L}+Z_{(1)}\mat{\slambda^0}\bigr)\bigl(\mat{A^0}\bigr)^{-1}&O(1)\\[0.5mm]
\hline
O(1)& Z_{(1)}'
\end{array}\right)+O\bigl(\varepsilon^2\bigr),\label{eq:F1A}
\end{align}
from Eq.~\eqref{eq:FBB0}. Using Lemma~\ref{lemma1} and Eq.~\eqref{eq:uscale}, we also find
\begin{align}
\sqrt{\dfrac{\tilde{g}}{g^0}}&=\left(1+2\varepsilon\tr{\mat{E}}+4\varepsilon^2\det{\mat{E}}\right)^{1/2}\nonumber\\
&=1+\varepsilon\tr{\mat{E}}+\dfrac{\varepsilon^2}{2}\left[4\det{\mat{E}}-(\tr{\mat{E}})^2\right]+O\bigl(\varepsilon^3\bigr).\label{eq:gg0}
\end{align}
Accordingly, from Proposition~\ref{prop1} and using Lemma~\ref{lemma1} again,
\begin{align}
\det{\tens{F}}&=1+\varepsilon\left\{Z_{(1)}'+\tr{\mat{E}}-\tr{\left[\bigl(Z^0\mat{L}+Z_{(1)}\mat{\slambda^0}\bigr)\bigl(\mat{A^0}\bigr)^{-1}\right]}\right\} \nonumber\\
&\quad+O\bigl(\varepsilon^2\bigr). \label{eq:detF1A}
\end{align}
The incompressibility condition $\det{\tens{F}}=1$ thus yields, at order $O(\varepsilon)$, an ordinary differential equation for $Z_{(1)}$. To make further progress, we shall need the following result:
\begin{lemma}\label{lemma2}
Let $\mat{M}$ be a $2\times 2$ matrix, and $x$ be a scalar. Then
\begin{align*}
(\mat{I}+x\mat{M})^{-1}=\dfrac{\mat{I}+x\adj{\mat{M}}}{1+x\tr{\mat{M}}+x^2\det{\mat{M}}}. 
\end{align*}
\end{lemma}
\begin{proof}
By definition of the adjugate matrix,
\begin{align*}
(\mat{I}+x\mat{M})^{-1}=\dfrac{\adj{(\mat{I}+x\mat{M})}}{\det{(\mat{I}+x\mat{M})}}=\dfrac{\adj{(\mat{I}+x\mat{M})}}{1+x\tr{\mat{M}}+x^2\det{\mat{M}}}, 
\end{align*}
using Lemma~\ref{lemma1}. But, by direct computation,
\begin{align*}
\adj{(\mat{I}+x\mat{M})}&=\left(\begin{array}{cc}
1+xM_{22}&-M_{12}\\
-M_{21}&1+xM_{11}
\end{array}\right)\\
&=\left(\begin{array}{cc}
1&0\\
0&1
\end{array}\right)+x\left(\begin{array}{cc}
M_{22}&-M_{12}\\
-M_{21}&M_{11}
\end{array}\right)=\mat{I}+x\adj{\mat{M}}. 
\end{align*}
The result follows.
\end{proof}

\noindent On multiplying this result by a general $2\times 2$ matrix $\mat{N}$ and taking the trace on both sides, we obtain

\begin{corollary}\label{cor1}
Let $\mat{M},\mat{N}$ be $2\times 2$ matrices, and let $x$ be a scalar. The following equality holds:
\begin{align*}
\tr{\left[\mat{N}(\mat{I}+x\mat{M})^{-1}\right]}=\dfrac{\tr{\mat{N}}+x\tr{(\mat{N}\adj{\mat{M}})}}{1+x\tr{\mat{M}}+x^2\det{\mat{M}}}.
\end{align*}
\end{corollary}

\noindent We shall also need the following observation:

\begin{lemma}\label{lemma3}
Let $\mat{M},\mat{N}$ be $2\times 2$ matrices. Then
\begin{align*}
\tr{(\mat{N}\adj{\mat{M}})}=\tr{\mat{M}}\tr{\mat{N}}-\tr{(\mat{MN})}\;\text{and}\;\tr{(\mat{M}\adj{\mat{M}})}=2\det{\mat{M}}.\\
\end{align*}
\end{lemma}
\begin{proof}
Notice that $\mat{M}+\adj{\mat{M}}=(\tr{\mat{M}})\mat{I}$ since
\begin{align*}
\left(\begin{array}{cc}
M_{11}&M_{12}\\
M_{21}&M_{22}
\end{array}\right)+\left(\begin{array}{cc}
M_{22}&-M_{12}\\
-M_{21}&M_{11}
\end{array}\right)=(M_{11}+M_{22})\left(\begin{array}{cc}
1&0\\
0&1
\end{array}\right).
\end{align*}
Hence $\mat{NM}+\mat{N}\adj{\mat{M}}=\left(\tr{\mat{M}}\right)\mat{N}$ on multiplication by $\mat{N}$. Taking the trace gives the first result. The second result follows from the definition of the adjugate, $\mat{M}\adj{\mat{M}}=(\det{\mat{M}})\mat{I}$, by taking the trace and noting that $\tr{\mat{I}}=2$.
\end{proof}
Combining Corollary~\ref{cor1} and Lemma~\ref{lemma3}, and recalling the definitions $\tr{\mat{\slambda^0}}=2\mathcal{H}^0$, $\det{\mat{\slambda^0}}=\mathcal{K}^0$, we find the differential equation for $Z_{(1)}$ resulting from Eq.~\eqref{eq:detF1A} to be
\begin{widetext}

\vspace{-5mm}
\begin{align}
Z_{(1)}'+\left(\dfrac{-2\mathcal{H}^0+2\mathcal{K}^0Z^0}{1-2\mathcal{H}^0Z^0+\mathcal{K}^0\bigl(Z^0\bigr)^2}\right)Z_{(1)}+\tr{\mat{E}}-\dfrac{Z^0\tr{\mat{L}}-\bigl(Z^0\bigr)^2\bigl[2\mathcal{H}^0\tr{\mat{L}}-\tr{\bigl(\mat{L\slambda^0}\bigr)}\bigr]}{1-2\mathcal{H}^0Z^0+\mathcal{K}^0\bigl(Z^0\bigr)^2}=0. 
\end{align}
Integrating and imposing $Z_{(1)}=0$ at $Z^0=0$, we obtain
\begin{align}
Z_{(1)}=-\dfrac{\left[Z^0-\mathcal{H}^0\bigl(Z^0\bigr)^2+\tfrac{1}{3}\mathcal{K}^0\bigl(Z^0\bigr)^3\right]\tr{\mat{E}}-\tfrac{1}{2}\bigl(Z^0\bigr)^2\tr{\mat{L}}+\tfrac{1}{3}\bigl(Z^0\bigr)^3\bigl[2\mathcal{H}^0\tr{\mat{L}}-\tr{\bigl(\mat{L\slambda^0}\bigr)}\bigr]}{1-2\mathcal{H}^0Z^0+\mathcal{K}^0\bigl(Z^0\bigr)^2}.\label{eq:Z1A}
\end{align}

\vspace{-3mm}
\end{widetext}
\paragraph*{Expansion at order $O\bigl(\varepsilon^2\bigr)$.} From Eq.~\eqref{eq:F1A}, we may write
\begin{align}
\tens{F}=\left(\begin{array}{c|c}
\mat{I}\!+\!\varepsilon\mat{B_{(1)}}\!+\!\varepsilon^2\mat{B_{(2)}}+O\bigl(\varepsilon^3\bigr)&\varepsilon\mat{v_{(1)}}+O\bigl(\varepsilon^2\bigr)\\[0.5mm]
\hline
\varepsilon\mat{w}^\top_{\smash{\mat{(1)}}}+O\bigl(\varepsilon^2\bigr)&1\!+\!\varepsilon c_{(1)}\!+\!\varepsilon^2c_{(2)}+O\bigl(\varepsilon^3\bigr)
\end{array}\right),\label{eq:F2A}
\end{align}
where, in particular and using Lemma~\ref{lemma2},
\begin{align}
\mat{B_{(1)}}=-\dfrac{\bigl(Z^0\mat{L}+Z_{(1)}\mat{\slambda^0}\bigr)\bigl(\mat{I}-Z^0\adj{\mat{\slambda^0}}\bigr)}{1-2\mathcal{H}^0Z^0+\mathcal{K}^0\bigl(Z^0\bigr)^2},\label{eq:B1A} 
\end{align}
in which $Z_{(1)}$ is given by Eq.~\eqref{eq:Z1A}. Explicit expressions for the terms $\mat{B_{(2)}},\mat{v_{(1)}},\mat{w_{(1)}},c_{(1)},c_{(2)}$ of the formal expansion~\eqref{eq:F2A} could be obtained in terms of the expansions defined in Eqs.~\eqref{eq:expansionsA}, but will turn out not to be required. 

From the general expression for the determinant of block matrices~\cite{matrix} and Eq.~\eqref{eq:F2A},
\begin{subequations}
\begin{align}
\det{\mat{F}}&=\bigl(1+\varepsilon c_{(1)}+\varepsilon^2c_{(2)}\bigr)\det{\left[\mat{I}+\varepsilon\mat{B_{(1)}}+\varepsilon^2\mat{B_{(2)}}\right.}\nonumber\\
&\hspace{25mm}\left.-\left(\varepsilon\mat{w}^\top_{\smash{\mat{(1)}}}\right)\left(\varepsilon\mat{v_{(1)}}\right)\right]+O\bigl(\varepsilon^3\bigr).
\end{align}
Expanding this using Lemma~\ref{lemma1}, and using Proposition~\ref{prop1} and Eq.~\eqref{eq:gg0}, we deduce that
\begin{align}
\det{\tens{F}}&=1+\varepsilon\bigl(\tr{\mat{B_{(1)}}}+\tr{\mat{E}}+c_{(1)}\bigr)\nonumber\\
&\quad+\varepsilon^2\bigl[\tr{\mat{B_{(2)}}}+c_{(2)}+c_{(1)}\tr{\mat{B_{(1)}}}+\det{\mat{B_{(1)}}}-\mat{w}^\top_{\smash{\mat{(1)}}}\mat{v_{(1)}}\nonumber\\
&\hspace{11mm}+\bigl(\tr{\mat{B_{(1)}}}+c_{(1)}\bigr)\tr{\mat{E}}+2\det{\mat{E}}-\tfrac{1}{2}(\tr{\mat{E}})^2\bigr]\nonumber\\
&\quad+O\bigl(\varepsilon^3\bigr).\label{eq:detF2A}
\end{align}
\end{subequations}
Next we introduce a formal expansion of the intrinsic deformation gradient,
\begin{subequations}
\begin{align}
\tens{F^0}=\left(\begin{array}{c|c}
\mat{B^0_{\smash{\mathsf{(0)}}}}+O(\varepsilon)&\mat{0}\\[0.5mm]
\hline
\vphantom{A^{A^{A^A}}}\varepsilon\mat{w^0_{\smash{\mat{(1)}}}}^{\hspace{-2mm}\top}+O\bigl(\varepsilon^2\bigr)&c^0_{\smash{(0)}}+O(\varepsilon)
\end{array}\right), \label{eq:F02A}
\end{align}
from Eq.~\eqref{eq:F0bm} and using the first of Eqs.~\eqref{eq:bb0}. In this expansion, $c^0_{\smash{(0)}}=\smash{{\zeta^0}_{\smash{,\zeta}}}$, which is positive by assumption. The values of the expansion terms $\smash{\mat{B^0_{\smash{\mat{(0)}}}}}$ and $\smash{\mat{w^0_{\smash{\mat{(1)}}}}}$ will turn out to be of no consequence. In particular, using Corollary~\ref{cor2},
\begin{align}
\bigl(\tens{F^0}\bigr)^{-\top}\!\!=\left(\begin{array}{c|c}
\hspace{-2mm}\begin{array}{l}\bigl(\mat{g^0}\bigr)^{-1}\bigl(\mat{B^0_{\smash{\mat{(0)}}}}\bigr)^{-\top}\mat{g}\\\quad\phantom{.}+O(\varepsilon)\end{array}\!\!&-\varepsilon\dfrac{\bigl(\mat{g^0}\bigr)^{-1}\bigl(\mat{B^0_{\smash{\mathsf{(0)}}}}\bigr)^{-\top}\mat{w^0_{\smash{\mat{(1)}}}}}{c^0_{(0)}}\!+\!O\bigl(\varepsilon^2\bigr)\!\!\\[4mm]
\hline
\vphantom{A^{A^{A^A}}}\mat{0}^\top&\dfrac{1\vphantom{A^A}}{c^0_{(0)}}+O(\varepsilon)\label{eq:F02invT}
\end{array}\right).
\end{align}
\end{subequations}
Moreover, from Eqs.~\eqref{eq:F2A} and \eqref{eq:F02A},
\begin{subequations}
\begin{align}
\btilde{\tens{F}}=\left(\begin{array}{c|c}
\mat{B^0_{\smash{\mathsf{(0)}}}}+O(\varepsilon)&\varepsilon c_0^{\smash{(0)}}\mat{v_{(1)}}+O\bigl(\varepsilon^2\bigr)\\[0.5mm]
\hline
\vphantom{A^{A^{A^{A^A}}}}\varepsilon\Bigl(\mat{w}^\top_{\smash{\mat{(1)}}}\mat{B^0_{\smash{\mat{(0)}}}}+\mat{w^0_{\smash{\mat{(1)}}}}^{\hspace{-2mm}\top}\Bigr)+O\bigl(\varepsilon^2\bigr)&c_0^{\smash{(0)}}+O(\varepsilon)
\end{array}\right), 
\end{align} 
so that, using the general expression for the inverse of a block matrix~\cite{matrix} and, once again, Corollary~\ref{cor2} and $\tilde{\mat{g}}=\mat{g^0}+O(\varepsilon)$,
\begin{align}
\btilde{\tens{F}}^{-\top}\!=\left(\begin{array}{c|c}
\hspace{-2mm}\begin{array}{l}\bigl(\mat{g^0}\bigr)^{-1}\bigl(\mat{B^0_{\smash{\mat{(0)}}}}\bigr)^{-\top}\mat{g}\\\quad\phantom{.}+O(\varepsilon)\end{array}\!&\!\begin{array}{l}-\varepsilon\dfrac{\bigl(\mat{g^0}\bigr)^{-1}}{c^0_{(0)}}\left(\mat{w_{(1)}}+\bigl(\mat{B^0_{\smash{\mat{(0)}}}}\bigr)^{-\top}\mat{w^0_{\smash{\mat{(1)}}}}\right)\\\quad\phantom{.}+O\bigl(\varepsilon^2\bigr)\end{array}\!\!\!\\[0.5mm]
\hline
O(\varepsilon)&\dfrac{1\vphantom{A^A}}{c^0_{(0)}}+O(\varepsilon)\label{eq:Fg2invT}
\end{array}\right).
\end{align}
\end{subequations}
On substituting Eqs.~\eqref{eq:F2A}, \eqref{eq:F02invT}, and \eqref{eq:Fg2invT} into definition~\eqref{eq:P} and recalling that $p=1+O(\varepsilon)$, we obtain
\begin{subequations}
\begin{align}
\tens{Q}=\left(\begin{array}{c|c}
O(\varepsilon)&\varepsilon\dfrac{\mat{v_{(1)}}+\bigl(\mat{g^0}\bigr)^{-1}\mat{w_{(1)}}}{c^0_{(0)}}+O\bigl(\varepsilon^2\bigr)\\[4mm]
\hline
O(\varepsilon)&O(\varepsilon)
\end{array}\right), 
\end{align}
and hence
\begin{align}
&\tens{Q_{(0)}}=\tens{O},&& \tens{Q_{(1)}}\vec{n}=\left(\begin{array}{c}
\dfrac{\mat{v_{(1)}}+\bigl(\mat{g^0}\bigr)^{-1}\mat{w_{(1)}}}{c^0_{(0)}}\\[4mm]\hline O(1)                                              
\end{array}\right).\label{eq:Q0Q1A}
\end{align}
\end{subequations}
As in Section~\ref{sec:model}, the fact that $\tens{Q_{(0)}}=\tens{O}$ implies that, at leading order, Eq.~\eqref{eq:Cauchy} is \smash{$(\tens{Q_{(1)}}\vec{n})_{,\zeta}=\vec{0}$}, with boundary conditions \smash{$\tens{Q_{\smash{\tens{(1)}}}^\bpm}\vec{n^\pm}=\vec{0}$}, which, as above, leads to $\tens{Q_{(1)}}\vec{n}\equiv\vec{0}$. This and the incompressibility condition $\det{\tens{F}}=1$ yield, from Eqs.~\eqref{eq:detF2A} and \eqref{eq:Q0Q1A},
\begin{subequations}\label{eq:c1c2A}
\begin{align}
&c_{(1)}=-\tr{\mat{B_{(1)}}}-\tr{\mat{E}},&& \mat{w_{(1)}}=-\mat{g^0v_{(1)}},\label{eq:c1A}
\end{align}
and hence
\begin{align}
c_{(2)}&=-\tr{\mat{B_{(2)}}}+\left(\tr{\mat{B_{(1)}}}+\tr{\mat{E}}\right)\tr{\mat{B_{(1)}}}-\det{\mat{B_{(1)}}}\nonumber\\
&\quad-2\det{\mat{E}}+\tfrac{3}{2}(\tr{\mat{E}})^2-\mat{v}^\top_{\smash{\mat{(1)}}}\mat{g^0v_{(1)}}.\label{eq:c2A}
\end{align}
\end{subequations}
\subsubsection{Asymptotic expansion of the constitutive relations}
To expand the constitutive relations and hence obtain the asymptotic expansion of the three-dimensional energy density, we need one more result:
\begin{lemma}\label{lemma4}
Let $\mat{M,N}$ be $2\times 2$ matrices. Then 
\begin{enumerate}[leftmargin=*,label={(\roman{enumi})},widest=2]
\item $\tr{\bigl(\mat{M}^2\bigr)}=(\tr{\mat{M}})^2-2\det{\mat{M}}$,
\item $\tr{\bigl(\mat{M}^2\mat{N}\bigr)}=\tr{\mat{M}}\tr{(\mat{MN})}-\det{\mat{M}}\tr{\mat{N}}$.
\end{enumerate}
\end{lemma}
\begin{proof}
The Cayley--Hamilton theorem~\cite{matrix} for a $2\times 2$ matrix states that $\mat{M}^2=(\tr{\mat{M}})\mat{M}-(\det{\mat{M}})\mat{I}$. Taking the trace on both sides of this relation and noting that $\tr{\mat{I}}=2$, we obtain~(i). Multiplying the Cayley--Hamilton relation by $\mat{N}$ and taking the trace yields~(ii).
\end{proof}
We start by computing the expansion of the (left) Cauchy--Green tensor $\tens{C}=\tens{F}^\top\tens{F}$. From Eq.~\eqref{eq:F2A}, we obtain
\begin{widetext}

\vspace{-2.5mm}
\begin{align}
\tens{F}^\top=\left(\begin{array}{c|c}
\mat{I}+\varepsilon\left[2\mat{E}+\bigl(\mat{g^0}\bigr)^{-1}\mat{B^\top_{\smash{\mat{(1)}}}}\mat{g^0}\right]+\varepsilon^2\left[2\bigl(\mat{g^0}\bigr)^{-1}\mat{B^\top_{\smash{\mat{(1)}}}}\mat{g^0}\mat{E}+\bigl(\mat{g^0}\bigr)^{-1}\mat{B^\top_{\smash{\mat{(2)}}}}\mat{g^0}\right]+O\bigl(\varepsilon^3\bigr)&\varepsilon\bigl(\mat{g^0}\bigr)^{-1}\mat{w_{(1)}}+O\bigl(\varepsilon^2\bigr)\\[1.5mm]
\hline
\varepsilon\mat{v^\top_{(1)}}\mat{g^0}+O\bigl(\varepsilon^2\bigr)&1+\varepsilon c_{(1)}+\varepsilon^2c_{(2)}+O\bigl(\varepsilon^3\bigr)
\end{array}\right), \label{eq:FT2A}
\end{align}
using Corollary~\ref{cor2} and Eq.~\eqref{eq:uscale}, and hence
\begin{align}
\tens{C}=\left(\begin{array}{c|c}
\hspace{-1.5mm}\begin{array}{l}\mat{I}+\varepsilon\left[2\mat{E}+\mat{B_{(1)}}+\bigl(\mat{g^0}\bigr)^{-1}\mat{B^\top_{\smash{\mat{(1)}}}}\mat{g^0}\right]+\varepsilon^2\left\{2\left[\mat{EB_{(1)}}+\bigl(\mat{g^0}\bigr)^{-1}\mat{B^\top_{\smash{\mat{(1)}}}}\mat{g^0E}\right]\right.\\+\left.\mat{B_{(2)}}+\bigl(\mat{g^0}\bigr)^{-1}\mat{B^\top_{\smash{\mat{(2)}}}}\mat{g^0}+\bigl(\mat{g^0}\bigr)^{-1}\mat{B^\top_{\smash{\mat{(1)}}}}\mat{g^0B_{(1)}}+\bigl(\mat{g^0}\bigr)^{-1}\mat{w_{(1)}}\mat{w^\top_{\smash{\mat{(1)}}}}\right\}+O\bigl(\varepsilon^3\bigr)\end{array}\!\!&O(\varepsilon)\\[4.5mm]
\hline
O(\varepsilon)&1+2\varepsilon c_{(1)}\!+\varepsilon^2\!\left(2c_{(2)}\!+c_{\smash{(1)}}^2\!+\mat{v^\top_{\smash{\mat{(1)}}}}\mat{g^0v_{(1)}}\right)+O\bigl(\varepsilon^3\bigr)\!
\end{array}\right). \label{eq:CA} 
\end{align}
We recall general properties of the trace operator: for matrices $\mat{M},\mat{N}$, $\tr{\mat{M^\top}}=\tr{\mat{M}}$ and $\tr{\mat{MN}}=\tr{\mat{NM}}$. Since Eq.~\eqref{eq:CA} represents~$\tens{C}$ with respect to $\mathcal{B}^0\otimes\bigl(\mathcal{B}^0\bigr)^\ast$, it follows that
\begin{subequations}
\begin{align}
\mathcal{I}_1&=3+\varepsilon \bigl[2(\tr{\mat{B_{(1)}}}+\tr{\mat{E}}+c_{(1)})\bigr]+\varepsilon^2\left\{2\left(\tr{\mat{B_{(2)}}}+c_{(2)}\right)+\left(\mat{v^\top_{\smash{\mat{(1)}}}}\mat{g^0v_{(1)}}+\mat{w^\top_{\smash{\mat{(1)}}}}\bigl(\mat{g^0}\bigr)^{-1}\mat{w_{(1)}}\right)+c_{(1)}^2+\tr{\left(\bigl(\mat{g^0}\bigr)^{-1}\mat{B^\top_{\smash{\mat{(1)}}}}\mat{g^0}\mat{B_{(1)}}\right)}\right.\nonumber\\
&\hspace{60mm}+\left.2\left[\tr{\left(\mat{EB_{(1)}}\right)}+\tr{\left(\mat{E}\bigl(\mat{g^0}\bigr)^{-1}\mat{B^\top_{\smash{\mat{(1)}}}}\mat{g^0}\right)}\right]\right\}+O\bigl(\varepsilon^3\bigr)\nonumber\\
&=3+\varepsilon^2\left\{2\left(\tr{\mat{E}}+\tr{\mat{B_{(1)}}}\right)^2+2\tr{\mat{E}^2}+\tr{\mat{B}^2_{\mat{(1)}}}+\tr{\left(\bigl(\mat{g^0}\bigr)^{-1}\mat{B^\top_{\smash{\mat{(1)}}}}\mat{g^0}\mat{B_{(1)}}\right)}+2\left[\tr{\left(\mat{EB_{(1)}}\right)}+\tr{\left(\mat{E}\bigl(\mat{g^0}\bigr)^{-1}\mat{B^\top_{\smash{\mat{(1)}}}}\mat{g^0}\right)}\right]\right\}+O\bigl(\varepsilon^3\bigr),\label{eq:I1A1}
\end{align}
using Eqs.~\eqref{eq:c1c2A} and Lemma~\ref{lemma4}. Recasting this result into a more symmetric form, 
\begin{align}
\mathcal{I}_1-3&=2\varepsilon^2\left[\bigl(\tr{\hat{\mat{E}}}\bigr)^2+\tr{\hat{\mat{E}}^2}\right]+O\bigl(\varepsilon^3\bigr),\qquad\text{where}\quad\hat{\mat{E}}=\mat{E}+\dfrac{1}{2}\left[\mat{B_{(1)}}+\bigl(\mat{g^0}\bigr)^{-1}\mat{B^\top_{\smash{\mat{(1)}}}}\mat{g^0}\right],\label{eq:I1A2}
\end{align}
\end{subequations}
so that $\hat{\mat{E}}$ is the effective two-dimensional strain. Thus Eq.~\eqref{eq:I1A2} determines the leading-order term in the expansion of the three-dimensional energy density $e$ defined in Eq.~\eqref{eq:E}, analogously to Eq.~\eqref{eq:effstrain}. This completes its asymptotic expansion in the limit of a thin shell that undergoes general large bending deformations.

We are left to express the leading-order expansion of $e$ in terms of tensorial invariants of the midsurface, thereby emphasising the tensorial nature of the shell theory. We substitute Eq.~\eqref{eq:Kmat} into Eq.~\eqref{eq:B1A} to find
\begin{subequations}\label{eq:B1A2}
\begin{align}
\mat{B_{(1)}}=-\dfrac{Z^0\mat{K}-2Z^0\mat{E\slambda^0}+Z_{(1)}\mat{\slambda^0}-\bigl(Z^0\bigr)^2\mat{K}\adj{\mat{\slambda^0}}+2\mathcal{K}^0\bigl(Z^0\bigr)^2\mat{E}-Z_{(1)}Z^0\mathcal{K}^0\mat{I}}{1-2\mathcal{H}^0Z^0+\mathcal{K}^0\bigl(Z^0\bigr)^2}+O(\varepsilon). 
\end{align}
By assumption and definitions~\eqref{eq:Etens} and \eqref{eq:curvtens}, tensors $\tens{\slambdab^0},\tens{E},\tens{K}$ are symmetric. We note that the curvature strain $\tens{L}$ is, from its definition in Eq.~\eqref{eq:curvtens}, not necessarily symmetric. Our choice to switch to a different measure of curvature strain at this stage is therefore motivated by symmetry, and not geometric interpretation as in Section~\ref{sec:model}. Now, using Proposition~\ref{prop3}, it follows that
\begin{align}
\bigl(\mat{g^0}\bigr)^{-1}\mat{B^\top_{\smash{\mat{(1)}}}}\mat{g^0}=-\dfrac{Z^0\mat{K}-2Z^0\mat{\slambda^0E}+Z_{(1)}\mat{\slambda^0}-\bigl(Z^0\bigr)^2\bigl(\adj{\mat{\slambda^0}}\bigr)\mat{K}+2\mathcal{K}^0\bigl(Z^0\bigr)^2\mat{E}-Z_{(1)}Z^0\mathcal{K}^0\mat{I}}{1-2\mathcal{H}^0Z^0+\mathcal{K}^0\bigl(Z^0\bigr)^2}+O(\varepsilon). 
\end{align}
\end{subequations}
Moreover, on substituting Eq.~\eqref{eq:Kmat} into Eq.~\eqref{eq:Z1A}, and using Lemma~\ref{lemma4}, we find
\begin{align}
Z_{(1)}=-\dfrac{Z^0\left[1-\mathcal{H}^0Z^0-\tfrac{1}{3}\mathcal{K}^0\bigl(Z^0\bigr)^2\right]\tr{\mat{E}}-\tfrac{1}{2}\bigl(Z^0\bigr)^2\left(1-\tfrac{4}{3}\mathcal{H}^0Z^0\right)\tr{\mat{K}}+\bigl(Z^0\bigr)^2\tr{\mat{E\slambda^0}}-\tfrac{1}{3}\bigl(Z^0\bigr)^3\tr{\mat{K\slambda^0}}}{1-2\mathcal{H}^0Z^0+\mathcal{K}^0\bigl(Z^0\bigr)^2}+O(\varepsilon). \label{eq:Z1A2}
\end{align}

We introduce the anticommutator $\langle\mat{M},\mat{N}\rangle$ of two matrices $\mat{M},\mat{N}$ by setting $\langle\mat{M},\mat{N}\rangle=(\mat{MN}+\mat{NM})/2$. With this notation, substituting Eq.~\eqref{eq:Z1A2} into Eqs.~\eqref{eq:B1A2} and the result into the definition of $\hat{\mat{E}}$ in Eq.~\eqref{eq:I1A2} yields
\begin{subequations}
\begin{align}
\hat{\mat{E}}&=\dfrac{\left[1-2\mathcal{H}^0Z^0-\mathcal{K}^0\bigl(Z^0\bigr)^2\right]\mat{E}-Z^0\mat{K}+2Z^0\left\langle\mat{E},\mat{\slambda^0}\right\rangle+\bigl(Z^0\bigr)^2\left\langle\mat{K},\adj{\mat{\slambda^0}}\right\rangle}{1-2\mathcal{H}^0Z^0+\mathcal{K}^0\bigl(Z^0\bigr)^2}\nonumber\\
&\quad+\dfrac{Z^0\left[1-\mathcal{H}^0Z^0-\tfrac{1}{3}\mathcal{K}^0\bigl(Z^0\bigr)^2\right]\tr{\mat{E}}-\tfrac{1}{2}\bigl(Z^0\bigr)^2\left(1-\tfrac{4}{3}\mathcal{H}^0Z^0\right)\tr{\mat{K}}+\bigl(Z^0\bigr)^2\tr{\left\langle\mat{E},\mat{\slambda^0}\right\rangle}-\tfrac{1}{3}\bigl(Z^0\bigr)^3\tr{\left\langle\mat{K},\mat{\slambda^0}\right\rangle}}{\left[1-2\mathcal{H}^0Z^0+\mathcal{K}^0\bigl(Z^0\bigr)^2\right]^2}\left(\mat{\slambda^0}-\mathcal{K}^0Z^0\mat{I}\right)+O(\varepsilon).\raisetag{2mm}
\label{eq:eA}
\end{align} 
\end{subequations}
\vspace{-3mm}
\end{widetext}
For the axisymmetric deformations in Section~\ref{sec:model}, using the identifications~\eqref{eq:eaxisym} and \eqref{eq:laxisym} of the axisymmetric shell and curvature strains in terms of the components of the general shell and curvature strain tensors used here and Eq.~\eqref{eq:Kmat} to switch between curvature strains, we find that ${\hat{E}^s}_s=a_{(1)}$ and $\smash{{\hat{E}^\phi}_\phi}´=b_{(1)}$, where $a_{(1)},b_{(1)}$ are defined in Eqs.~\eqref{eq:a1b1}. Comparing Eqs.~\eqref{eq:I1A2} and \eqref{eq:effstrain} then shows that the general result derived here is consistent with the result for axisymmetric deformations obtained in Section~\ref{sec:model}.

The next step in the derivation is to substitute Eq.~\eqref{eq:eA}, finally, into Eq.~\eqref{eq:I1A2} and hence Eq.~\eqref{eq:E}. To express the resulting expansion of the energy density $e$ in terms of the first- and second-order invariants that can be constructed from $\tens{\slambdab^0},\tens{E},\tens{K}$ only, we need to make two more general observations:
\begin{lemma}\label{lemma5}
Let $\mat{A},\mat{B},\mat{C}$ be $2\times 2$ matrices. Then
\begin{align*}
2\tr{\left(\left\langle\mat{A},\mat{B}\right\rangle\mat{C}\right)}&=\tr{\left(\left\langle\mat{A},\mat{B}\right\rangle\right)}\tr{\mat{C}}+\tr{\left(\left\langle\mat{B},\mat{C}\right\rangle\right)}\tr{\mat{A}}\\
&\qquad+\tr{\left(\left\langle\mat{C},\mat{A}\right\rangle\right)}\tr{\mat{B}}-\tr{\mat{A}}\tr{\mat{B}}\tr{\mat{C}}. 
\end{align*}
\end{lemma}
\begin{proof}
The proof proceeds by direct calculation. We write
\begin{align*}
\mat{A}=\left(\begin{array}{cc}
A_{11}&A_{12}\\A_{21}&A_{22}               
\end{array}\right),\;\mat{B}=\left(\begin{array}{cc}
B_{11}&B_{12}\\B_{21}&B_{22}               
\end{array}\right),\;\mat{C}=\left(\begin{array}{cc}
C_{11}&C_{12}\\C_{21}&C_{22}               
\end{array}\right) 
\end{align*}
and compute
\begin{align*}
&2\tr{\left(\left\langle\mat{A},\mat{B}\right\rangle\mat{C}\right)}\\
&\;= A_{11} B_{11} C_{11}+A_{21} B_{12} C_{11}+A_{12} B_{21} C_{11}+A_{21} B_{11} C_{12}\\
&\qquad+A_{11} B_{21} C_{12}+A_{22} B_{21} C_{12}+A_{21} B_{22} C_{12}+A_{12} B_{11} C_{21}\\
&\qquad+A_{11} B_{12} C_{21}+A_{22} B_{12} C_{21}+A_{12} B_{22} C_{21}+A_{21} B_{12} C_{22}\\
&\qquad+A_{12} B_{21} C_{22}+2 A_{22} B_{22} C_{22}\\
&\;=(A_{11} B_{11}+A_{21} B_{12}+A_{12} B_{21}+A_{22} B_{22})(C_{11}+C_{22})\\
&\qquad+(B_{11} C_{11}+B_{21} C_{12}+B_{12} C_{21}+B_{22} C_{22})(A_{11}+A_{22})\\
&\qquad+(A_{11} C_{11}+A_{21} C_{12}+A_{12} C_{21}+A_{22} C_{22})(B_{11}+B_{22})\\
&\qquad-(A_{11}+A_{22})(B_{11}+B_{22})(C_{11}+C_{22})\\
&\;=\tr{(\mat{AB})}\tr{\mat{C}}+\tr{(\mat{BC})}\tr{\mat{A}}+\tr{(\mat{AC})}\tr{\mat{B}}-\tr{\mat{A}}\tr{\mat{B}}\tr{\mat{C}}.
\end{align*}
By the symmetry of trace, this completes the proof.
\end{proof}
\begin{corollary}\label{cor4}
Let $\mat{A},\mat{B},\mat{C}$ be $2\times 2$ matrices. Then
\begin{align*}
\tr{\left(\langle\mat{AB},\mat{CB}\rangle\right)}&=\tr{\left(\left\langle\mat{A},\mat{B}\right\rangle\right)}\tr{\left(\left\langle\mat{B},\mat{C}\right\rangle\right)}\\
&\qquad-\det{\mat{B}}\left[\tr{\left(\left\langle\mat{A},\mat{C}\right\rangle\right)}-\tr{\mat{A}}\tr{\mat{C}}\right]. 
\end{align*}
\end{corollary}
\begin{proof}
Using Lemmata~\ref{lemma4} and \ref{lemma5}, we find
\begin{align*}
\tr{\left(\langle\mat{AB},\mat{CB}\rangle\right)}&=2\tr{[\left(\langle\mat{A},\mat{B}\rangle\right)\mat{CB}]}-\tr{\bigl[\mat{B}^2(\mat{AC})\bigr]}\\
&=\{\tr{(\mat{AB})}\tr{(\mat{CB})}+\tr{[\mat{A}(\mat{CB})]}\tr{\mat{B}}\\
&\qquad+\tr{[\mat{B}(\mat{CB})]}\tr{\mat{A}}-\tr{\mat{A}}\tr{\mat{B}}\tr{(\mat{CB})}\}\\
&\quad-\{\tr{\mat{B}}\tr{[\mat{B(AC)}]}-\det{\mat{B}}\tr{(\mat{AC})}\}\\
&=\tr{(\mat{AB})}\tr{(\mat{CB})}+\left[\tr{\mat{B}}\tr{(\mat{BC})}\!-\!\det{\mat{B}}\tr{\mat{C}}\right]\tr{\mat{A}}\\
&\quad-\tr{\mat{A}}\tr{\mat{B}}\tr{(\mat{CB})}+\det{\mat{B}}\tr{(\mat{AC})}\\
&=\tr{(\mat{AB})}\tr{(\mat{CB})}+\det{\mat{B}}\left[\tr{(\mat{AC})}-\tr{\mat{A}}\tr{\mat{C}}\right],
\end{align*}
which, again by the symmetry of trace, finishes the proof.
\end{proof}
To simplify expressions in subsequent calculations, it will be convient to rewrite the expression for the effective strain $\hat{\mat{E}}$ in Eq.~\eqref{eq:eA} as
\begin{widetext}\setcounter{equation}{68}
\begin{subequations}\setcounter{equation}{1}
\begin{align}
\hat{\mat{E}}=e_1\mat{E}+e_2\mat{K}+e_3\left\langle\mat{E},\mat{\slambda^0}\right\rangle+e_4\left\langle\mat{K},\adj{\mat{\slambda^0}}\right\rangle+E\left(\mat{\slambda^0}-\mathcal{K}^0Z^0\mat{I}\right)+O(\varepsilon),
\end{align}
\end{subequations}
in which $e_1,e_2,e_3,e_4$ are functions of $Z^0$ and $\mathcal{H}^0,\mathcal{K}^0$ only, and $E$ additionally depends on $\tr{\mat{E}},\tr{\mat{K}},\tr{\left\langle\mat{E},\mat{\slambda^0}\right\rangle},\tr{\left\langle\mat{K},\mat{\slambda^0}\right\rangle}$. Explicit expressions for $e_1,e_2,e_3,e_4$ are easily extracted from Eq.~\eqref{eq:eA}. It follows that
\begin{subequations}\label{eq:trEA}
\begin{align}
\tr{\hat{\mat{E}}}&=e_1\tr{\mat{E}}+e_2\tr{\mat{K}}+e_3\tr{\left\langle\mat{E},\mat{\slambda^0}\right\rangle}+e_4\tr{\left\langle\mat{K},\adj{\mat{\slambda^0}}\right\rangle}+2E\left(\mathcal{H}^0-\mathcal{K}^0Z^0\right)+O(\varepsilon),\\
\tr{\hat{\mat{E}}^2}&=e_1^2\tr{\mat{E}^2}+e_2^2\tr{\mat{K}^2}+e_3^2\tr{\left\langle\mat{E},\mat{\slambda^0}\right\rangle^2}+e_4^2\tr{\left\langle\mat{K},\adj{\mat{\slambda^0}}\right\rangle^2}+E^2\left[\tr{\bigl(\mat{\slambda^0}\bigr)^2}-4\mathcal{H}^0\mathcal{K}^0Z^0+2\bigl(\mathcal{K}^0Z^0\bigr)^2\right]+2e_1e_2\tr{\left\langle\mat{E},\mat{K}\right\rangle}\nonumber\\
&\quad +2e_1e_3\tr{\left\langle\mat{E},\left\langle\mat{E},\mat{\slambda^0}\right\rangle\right\rangle}+2e_1e_4\tr{\left\langle\mat{E},\left\langle\mat{K},\adj{\mat{\slambda^0}}\right\rangle\right\rangle}+2e_1E\left(\tr{\left\langle\mat{E},\mat{\slambda^0}\right\rangle}-\mathcal{K}^0Z^0\tr{\mat{E}}\right)+2e_2e_3\tr{\left\langle\mat{K},\left\langle\mat{E},\mat{\slambda^0}\right\rangle\right\rangle}\nonumber\\
&\quad+2e_2e_4\tr{\left\langle\mat{K},\left\langle\mat{K},\adj{\mat{\slambda^0}}\right\rangle\right\rangle}+2e_2E\left(\tr{\left\langle\mat{K},\mat{\slambda^0}\right\rangle}-\mathcal{K}^0Z^0\tr{\mat{K}}\right)+2e_3e_4\tr{\left\langle\left\langle\mat{E},\mat{\slambda^0}\right\rangle,\left\langle\mat{K},\adj{\mat{\slambda^0}}\right\rangle\right\rangle}\nonumber\\
&\quad+2e_3E\left(\tr{\left\langle\mat{\slambda^0},\left\langle\mat{E},\mat{\slambda^0}\right\rangle\right\rangle}-\mathcal{K}^0Z^0\tr{\left\langle\mat{E},\mat{\slambda^0}\right\rangle}\right)+2e_4E\left(\tr{\left\langle\mat{\slambda^0},\left\langle\mat{K},\adj{\mat{\slambda^0}}\right\rangle\right\rangle}-\mathcal{K}^0Z^0\tr{\left\langle\mat{K},\adj{\mat{\slambda^0}}\right\rangle}\right)+O(\varepsilon).
\end{align}
\end{subequations}
Expressing Eqs.~\eqref{eq:I1A2} and hence \eqref{eq:E} in terms of first- and second-order invariants only requires simplifying the different traces of higher-order expressions appearing in Eqs.~\eqref{eq:trEA}. We do so by applying Lemmata~\ref{lemma3}, \ref{lemma4}, \ref{lemma5} and Corollary~\ref{cor4} repeatedly to find
\begin{subequations}\label{eq:idsA}
\begin{align}
\tr{\left\langle\mat{K},\adj{\mat{\slambda^0}}\right\rangle}&=2\mathcal{H}^0\tr{\mat{K}}-\tr{\left\langle\mat{K},\mat{\slambda^0}\right\rangle},&\tr{\bigl(\mat{\slambda^0}\bigr)^2}&=4\bigl(\mathcal{H}^0\bigr)^2-2\mathcal{K}^0,\\
\tr{\left\langle\mat{E},\left\langle\mat{E},\mat{\slambda^0}\right\rangle\right\rangle}&=\tr{\mat{E}}\tr{\left\langle\mat{E},\mat{\slambda^0}\right\rangle}+\mathcal{H}^0\left[\tr{\mat{E}^2}-(\tr{\mat{E}})^2\right],&\tr{\left\langle\mat{K},\left\langle\mat{K},\adj{\mat{\slambda^0}}\right\rangle\right\rangle}&=\mathcal{H}^0\left[\tr{\mat{K}^2}+(\tr{\mat{K}})^2\right]-\tr{\mat{K}}\tr{\left\langle\mat{K},\mat{\slambda^0}\right\rangle},\\
\tr{\left\langle\mat{\slambda^0},\left\langle\mat{E},\mat{\slambda^0}\right\rangle\right\rangle}&=2\mathcal{H}^0\tr{\left\langle\mat{E},\mat{\slambda^0}\right\rangle}-\mathcal{K}^0\tr{\mat{E}},&\tr{\left\langle\mat{\slambda^0},\left\langle\mat{K},\adj{\mat{\slambda^0}}\right\rangle\right\rangle}&=\mathcal{K}^0\tr{\mat{K}},
\end{align}
and
\begin{align}
\tr{\left\langle\mat{E},\mat{\slambda^0}\right\rangle^2}&=\bigl(\mathcal{H}^0\bigr)^2\tr{\mat{E}^2}+\mathcal{H}^0\tr{\mat{E}}\tr{\left\langle\mat{E},\mat{\slambda^0}\right\rangle}-\left[\bigl(\mathcal{H}^0\bigr)^2+\tfrac{1}{2}\mathcal{K}^0\right](\tr{\mat{E}})^2+\tfrac{1}{2}\left[\tr{\left\langle\mat{E},\mat{\slambda^0}\right\rangle}\right]^2,\\
\tr{\left\langle\mat{K},\adj{\mat{\slambda^0}}\right\rangle^2}&=\bigl(\mathcal{H}^0\bigr)^2\tr{\mat{K}^2}-3\mathcal{H}^0\tr{\mat{K}}\tr{\left\langle\mat{K},\mat{\slambda^0}\right\rangle}+\left[3\bigl(\mathcal{H}^0\bigr)^2-\tfrac{1}{2}\mathcal{K}^0\right](\tr{\mat{K}})^2+\tfrac{1}{2}\left[\tr{\left\langle\mat{K},\mat{\slambda^0}\right\rangle}\right]^2,\\
\tr{\left\langle\mat{E},\left\langle\mat{K},\adj{\mat{\slambda^0}}\right\rangle\right\rangle}&=\mathcal{H}^0\bigl(\tr{\left\langle\mat{E},\mat{K}\right\rangle}+\tr{\mat{E}}\tr{\mat{K}}\bigr)-\tfrac{1}{2}\left(\tr{\left\langle\mat{E},\mat{\slambda^0}\right\rangle}\tr{\mat{K}}+\tr{\left\langle\mat{K},\mat{\slambda^0}\right\rangle}\tr{\mat{E}}\right),\\
\tr{\left\langle\mat{K},\left\langle\mat{E},\mat{\slambda^0}\right\rangle\right\rangle}&=\mathcal{H}^0\bigl(\tr{\left\langle\mat{E},\mat{K}\right\rangle}-\tr{\mat{E}}\tr{\mat{K}}\bigr)+\tfrac{1}{2}\left(\tr{\left\langle\mat{E},\mat{\slambda^0}\right\rangle}\tr{\mat{K}}+\tr{\left\langle\mat{K},\mat{\slambda^0}\right\rangle}\tr{\mat{E}}\right),\\
\tr{\left\langle\left\langle\mat{E},\mat{\slambda^0}\right\rangle,\left\langle\mat{K},\adj{\mat{\slambda^0}}\right\rangle\right\rangle}&=\left[\bigl(\mathcal{H}^0\bigr)^2+\mathcal{K}^0\right]\tr{\left\langle\mat{E},\mat{K}\right\rangle}-\left[\bigl(\mathcal{H}^0\bigr)^2+\tfrac{1}{2}\mathcal{K}^0\right]\tr{\mat{E}}\tr{\mat{K}}+\tfrac{1}{2}\mathcal{H}^0\left(\tr{\left\langle\mat{E},\mat{\slambda^0}\right\rangle}\tr{\mat{K}}+\tr{\left\langle\mat{K},\mat{\slambda^0}\right\rangle}\tr{\mat{E}}\right)\nonumber\\
&\qquad-\tfrac{1}{2}\tr{\left\langle\mat{E},\mat{\slambda^0}\right\rangle}\tr{\left\langle\mat{K},\mat{\slambda^0}\right\rangle}.
\end{align}
\end{subequations}
Inserting Eqs.~\eqref{eq:idsA} into Eqs.~\eqref{eq:trEA}, and the result into Eqs.~\eqref{eq:I1A2} and \eqref{eq:E} as announced, we finally obtain
\begin{align}
e&=C\varepsilon^2\left\{\left(\alpha_1\tr{\tens{E}^2}+\alpha_2(\tr{\tens{E}})^2+\alpha_3\tr{\tens{E}}\tr{\left\langle\tens{E},\tens{\slambdab^0}\right\rangle}+\alpha_4\left[\tr{\left\langle\tens{E},\tens{\slambdab^0}\right\rangle}\right]^2\right)+\left(\beta_1\tr{\left\langle\tens{E},\tens{K}\right\rangle}+\beta_2\tr{\tens{E}}\tr{\tens{K}}+\beta_3\tr{\tens{E}}\tr{\left\langle\tens{K},\tens{\slambdab^0}\right\rangle}\right.\right.\nonumber\\
&\hspace{12mm}+\left.\left.\beta_4\tr{\tens{K}}\tr{\left\langle\tens{E},\tens{\slambdab^0}\right\rangle}+\beta_5\tr{\left\langle\tens{E},\tens{\slambdab^0}\right\rangle}\tr{\left\langle\tens{K},\tens{\slambdab^0}\right\rangle}\right)+\left(\gamma_1\tr{\tens{K}^2}+\gamma_2(\tr{\tens{K}})^2+\gamma_3\tr{\tens{K}}\tr{\left\langle\tens{K},\tens{\slambdab^0}\right\rangle}+\gamma_4\left[\tr{\left\langle\tens{K},\tens{\slambdab^0}\right\rangle}\right]^2\right)\right\}+O\bigl(\varepsilon^3\bigr),\label{eq:eA2}
\end{align}
\end{widetext}
in which the stretching coefficients $\alpha_1,\alpha_2,\alpha_3,\alpha_4$, the coupling coefficients $\beta_1,\beta_2,\beta_3,\beta_4,\beta_5$, and the bending coefficients $\gamma_1,\gamma_2,\gamma_3,\gamma_4$ are rational functions of $Z^0$ and $\mathcal{H}^0,\mathcal{K}^0$, so depend on the intrinsic configuration only. The explicit expressions for these coefficients are not edifying, and therefore not presented here. 

We have been able to use tensor traces rather than matrix traces in this expressions since $\mat{\slambda^0},\mat{E},\mat{K}$ represent $\tens{\slambdab^0},\tens{E},\tens{K}$ with respect to $\smash{\mathcal{B}^0\otimes\bigl(\mathcal{B}^0\bigr)^\ast}$. This stresses the tensorial invariance of the theory. The anticommutators in Eq.~\eqref{eq:eA2} could of course be simplified using the symmetry of trace, but we have not done so to emphasise their symmetry. 
\subsubsection{Averaging over the transverse coordinate}
The volume element in the intrinsic configuration $\mathcal{V}^0$ is, by definition and using intrinsic volume conservation and Eq.~\eqref{eq:IVC2},
\begin{align}
\mathrm{d}V^0&=\sqrt{\dfrac{\det{\smash{\mat{G^0}}}}{\det{\mat{G}}}}\mathrm{d}V=\sqrt{\det{\mat{G^0}}}\left(\dfrac{\mathrm{d}S}{\sqrt{g}}\right)\mathrm{d}\zeta\nonumber\\
&=\varepsilon\left[1+2\mathcal{H}^0Z^0+\mathcal{K}^0\bigl(Z^0\bigr)^2\right]\mathrm{d}S\,\mathrm{d}Z^0+O\bigl(\varepsilon^2\bigr),
\end{align}
where $\mathrm{d}V$ is the volume element of the undeformed configuration $\mathcal{V}$ and $\mathrm{d}S$ is the surface element of the undeformed midsurface~$\mathcal{S}$. From Eq.~\eqref{eq:E}, the elastic energy of the shell is therefore
\begin{subequations}
\begin{align}
\mathcal{E}=\int_{\mathcal{S}}{\hat{e}\,\mathrm{d}S}, 
\end{align}
in which, at leading order,
\begin{align}
\hat{e}=\varepsilon\int_{-H^0/2}^{H^0/2}{e\bigl(Z^0\bigr)\left[1+2\mathcal{H}^0Z^0+\mathcal{K}^0\bigl(Z^0\bigr)^2\right]\mathrm{d}Z^0} \label{eq:eZ0A}
\end{align}
\end{subequations}
is the effective two-dimensional energy density. In the integral limits, $H^0$ is determined in terms of the undeformed thickness $h$ of the shell by Eq.~\eqref{eq:H0}.

Since the coefficient functions $\alpha_1,\alpha_2,\alpha_3,\alpha_4$, $\beta_1,\beta_2,\beta_3,\beta_4$, $\beta_5$, and $\gamma_1,\gamma_2,\gamma_3,\gamma_4$ that appear in Eq.~\eqref{eq:eA2} are rational functions of $Z^0$, the integral with respect to $Z^0$ in Eq.~\eqref{eq:eZ0A} can be performed in closed form, but the resulting expressions are extremely cumbersome and therefore again not presented here. For this reason, the theory for large bending deformations is likely to be most useful for deformations with some additional symmetry, such as the axisymmetric deformations discussed in Section~\ref{sec:model}.
\subsection{Limit of small bending deformations}
We conclude our calculations by discussing the limit of small bending deformations. In this limit, $\tens{\slambdab^0}\rightarrow\tens{O}$, and hence $\mathcal{H}^0,\mathcal{K}^0\rightarrow 0$, and the effective strain in Eq.~\eqref{eq:eA} reduces to the rather simpler form
\begin{align}
\hat{\mat{E}}=\mat{E}-Z^0\mat{K}+O(\varepsilon),
\end{align}
and so Eqs.~\eqref{eq:I1A2} and \eqref{eq:E} yield
\begin{align}
e&=C\varepsilon^2\Bigl\{\left[\tr{\tens{E}^2}+(\tr{\tens{E}})^2\right]-2Z^0\bigl(\tr{\langle\tens{E},\tens{K}\rangle}+\tr{\tens{E}}\tr{\tens{K}}\bigr)\nonumber\\
&\hspace{12mm}+\left.\bigl(Z^0\bigr)^2\left[\tr{\tens{K}^2}+(\tr{\tens{K}})^2\right]\right\},
\end{align}
where we have again replaced matrix traces with the corresponding tensor traces. Moreover, Eq.~\eqref{eq:H0} shows that, in this limit, $H^0=h$, and so Eq.~\eqref{eq:eZ0A} becomes
\begin{align}
\hat{e}&=\varepsilon\int_{-h/2}^{h/2}{e\bigl(Z^0\bigr)\,\mathrm{d}Z^0}\nonumber\\
&=\varepsilon^3\!\left\{Ch\left[\tr{\tens{E}^2}+(\tr{\tens{E}})^2\right]\!+\dfrac{Ch^3}{12}\!\left[\tr{\tens{K}^2}+(\tr{\tens{K}})^2\right]\right\}+O\bigl(\varepsilon^4\bigr),
\end{align}
which recovers the tensorial form of the incompressible limit of Koiter's shell theory~\cite{gregory17}.

\section{\uppercase{Derivation of the Governing Equa- tions for Axisymmetric Deformations}}\label{appB}
In this Appendix, we derive the governing equations for axisymmetric deformations, by varying the elastic energy~\eqref{eq:E2}. Similar derivations are given in our previous work~\cite{haas18a,haas18b} for the elastic theories considered there, but here, we will keep the explicit asymptotic scalings in the derivation. From Eq.~\eqref{eq:edensk} and considering leading-order terms only,
\begin{align}
\delta\hat{e}=\varepsilon\left(n_s\,\delta E_s+n_\phi\,\delta E_\phi\right)+m_s\,\delta K_s+m_\phi\,\delta K_\phi, 
\end{align}
wherein the shell stresses and shell moments are
\begin{subequations}
\begin{align}
n_s&=C\varepsilon^2h\left[\bar{\alpha}_{ss}E_s+\bar{\alpha}_{s\phi}E_\phi+h\left(\bar{\beta}_{ss}K_s+\bar{\beta}_{s\phi}K_\phi\right)\right],\label{eq:NsB}\\
n_\phi&=C\varepsilon^2h\left[\bar{\alpha}_{\phi s}E_s+\alpha_{\phi\phi}E_\phi+h\left(\beta_{\phi s}K_s+\beta_{\phi\phi}K_\phi\right)\right],\label{eq:NphiB} \\
m_s&=C\varepsilon^3h^2\left[\bar{\beta}_{ss}E_s+\beta_{\phi s}E_\phi+h\left(\gamma_{ss}K_s+\gamma_{s\phi}K_\phi\right)\right],\label{eq:MsB}\\
m_\phi&=C\varepsilon^3h^2\left[\bar{\beta}_{s\phi}E_s+\beta_{\phi\phi}E_\phi+h\left(\gamma_{\phi s}K_s+\gamma_{\phi\phi}K_\phi\right)\right],\label{eq:MphiB}
\end{align}
\end{subequations}
since $\bar{\alpha}_{s\phi}=\bar{\alpha}_{\phi s}$, $\gamma_{s\phi}=\gamma_{\phi s}$. Now, from the definitions of the shell and curvature strains in Eqs.~\eqref{eq:EsEphi} and \eqref{eq:KsKphi}, 
\begin{subequations}
\begin{align}
\delta E_s&=\dfrac{\sec{\tilde{\psi}}\,\delta\tilde{r}'+\tilde{f}_s\tan{\tilde{\psi}}\,\delta\tilde{\psi}}{\varepsilon f_s^0},&\delta E_\phi&=\dfrac{1}{\varepsilon f_\phi^0}\left(\dfrac{\delta\tilde{r}}{r}\right),
\end{align}
and
\begin{align}
\delta K_s&=\dfrac{\delta\tilde{\psi}'}{\bigl(f_s^0\bigr)^2f_\phi^0},&\delta K_\phi&=\dfrac{1}{f_s^0\bigl(f_\phi^0\bigr)^2}\left(\dfrac{\cos{\psi}}{r}\delta\psi\right).
\end{align}
\end{subequations}
Hence, on letting
\begin{subequations}\label{eq:NMB}
\begin{align}
N_s&=\dfrac{n_s}{\tilde{f}_\phi f_s^0},& N_\phi&=\dfrac{n_\phi}{\tilde{f}_sf_\phi^0},\\
M_s&=\dfrac{m_s}{\tilde{f}_\phi\bigl(f_s^0\bigr)^2f_\phi^0},&M_\phi&=\dfrac{m_\phi}{\tilde{f}_sf_s^0\bigl(f_\phi^0\bigr)^2},
\end{align}
\end{subequations}
we obtain, from Eq.~\eqref{eq:E2} and using Eqs.~\eqref{eq:fs},
\begin{align}
\dfrac{\delta\mathcal{E}}{2\pi}&=\left\llbracket \tilde{r}N_s\sec{\tilde{\psi}}\,\delta\tilde{r}+\tilde{r}M_s\,\delta\tilde{\psi}\right\rrbracket\nonumber\\
&\quad-\int_{\mathcal{C}}{\left[\left(\dfrac{\mathrm{d}}{\mathrm{d}s}\left(\tilde{r}M_s\right)-\tilde{r}\tilde{f}_sN_s\tan{\tilde{\psi}}-\tilde{f}_sM_\phi\cos{\tilde{\psi}}\right)\delta\tilde{\psi}\right]\mathrm{d}s}\nonumber\\
&\quad-\int_{\mathcal{C}}{\left[\left(\dfrac{\mathrm{d}}{\mathrm{d}s}\left(\tilde{r}N_s\sec{\tilde{\psi}}\right)-\tilde{f}_sN_\phi\right)\delta\tilde{r}\right]\mathrm{d}s},
\end{align}
from which we read off the governing equations and boundary conditions.

As in standard shell theories~\cite{libai}, the apparent singularity in the resulting equations is removed by introducing the transverse shear tension, $T=-N_s\tan{\tilde{\psi}}$, and we obtain, using Eqs.~\eqref{eq:drz} and \eqref{eq:kappas},
\begin{subequations}\label{eq:dNTB}
\begin{align}
\dfrac{\mathrm{d}N_s}{\mathrm{d}s}&=\tilde{f}_s\left(\dfrac{N_\phi-N_s}{\tilde{r}}\cos{\tilde{\psi}}+\tilde{\kappa}_sT\right),\label{eq:dNsB}\\
\dfrac{\mathrm{d}M_s}{\mathrm{d}s}&=\tilde{f}_s\left(\dfrac{M_\phi-M_s}{\tilde{r}}\cos{\tilde{\psi}}-T\right).
\end{align}
Moreover, by differentiating the definition of $T$ and using Eq.~\eqref{eq:dNsB}, we find\begin{align}
\dfrac{\mathrm{d}T}{\mathrm{d}s}=-\tilde{f}_s\left(\tilde{\kappa}_sN_s+\tilde{\kappa}_\phi N_\phi+T\dfrac{\cos{\tilde{\psi}}}{\tilde{r}}\right).\label{eq:dTB}
\end{align}
\end{subequations}
Together with the relations 
\begin{align}
&\dfrac{\mathrm{d}\tilde{r}}{\mathrm{d}s}=\tilde{f}_s\cos{\tilde{\psi}},&&\dfrac{\mathrm{d}\tilde{\psi}}{\mathrm{d}s}=\tilde{f}_s\tilde{\kappa}_s\label{eq:drB}
\end{align}
from Eqs.~\eqref{eq:drz} and \eqref{eq:kappas}, Eqs.~\eqref{eq:dNTB} determine the deformed configuration of the shell. Having solved these equations, integrating the otherwise redundant shape equation $\tilde{z}'=\tilde{f}_s\sin{\tilde{\psi}}$ from Eqs.~\eqref{eq:drz} determines the shape of the shell completely.

\subsection*{Numerical solution of Eqs.~(\ref{eq:dNTB})}
We conclude the derivation of the governing equations for axisymmetric deformations with two remarks on the numerical solution of Eqs.~\eqref{eq:dNTB}.

First, we note that Eqs.~\eqref{eq:dNTB} are singular where $\tilde{r}=0$. At such a point, geometric continuity implies $\tilde{\psi}=0$. Hence $T=0$ there by definition, and $N_\phi=N_s$ for regularity in Eq.~\eqref{eq:dNsB}. Moreover, by applying l'H\^opital's rule to the definitions in Eqs.~\eqref{eq:fs} and \eqref{eq:kappas}, $\tilde{f}_s=\tilde{f}_\phi$, $\tilde{\kappa}_s=\tilde{\kappa}_\phi$. Hence Eqs.~\eqref{eq:dNTB} are replaced with
\begin{align}
\dfrac{\mathrm{d}N_s}{\mathrm{d}s}&=0,&&\dfrac{\mathrm{d}M_s}{\mathrm{d}s}=0,&&\dfrac{\mathrm{d}T}{\mathrm{d}s}=-\tilde{f}_s\tilde{\kappa}_s N_s,
\end{align}
of which the first two follow by reflection across the axis of symmetry, and the last follows by applying l'H\^opital's rule to Eq.~\eqref{eq:dTB} and using the previous observations and Eqs.~\eqref{eq:drB}.

Second, as discussed in Refs.~\cite{haas18a,haas18b}, too, at each stage of the numerical solution, $\tilde{f}_s,\tilde{f}_\phi,\tilde{\kappa}_s,\tilde{\kappa}_\phi$ must be determined from $\tilde{r},\smash{\tilde{\psi}},M_s,N_s$. To begin with, $\smash{\tilde{f}_\phi},\tilde{\kappa}_\phi$ and hence $E_\phi,K_\phi$ are computed directly from $\tilde{r},\tilde{\psi}$ using their definitions~\eqref{eq:EsEphi} and \eqref{eq:KsKphi}. We can then compute $\tilde{f}_s,\tilde{\kappa}_s$ by noting that, once $\tilde{f}_\phi,E_\phi,K_\phi$ are known, the definitions of $N_s,M_s$ in Eqs.~\eqref{eq:NsB}, \eqref{eq:MsB}, and $\eqref{eq:NMB}$ define a system of linear equations for $E_s,K_s$. Its solution and definitions~\eqref{eq:EsEphi} and \eqref{eq:KsKphi} yield $\tilde{f}_s$ and finally $\tilde{\kappa}_s$. We can then compute $N_\phi,M_\phi$ using Eqs.~\eqref{eq:NphiB}, \eqref{eq:MphiB}, and \eqref{eq:NMB}, and thus continue the numerical integration. Moreover, if $\tilde{r}=0$, we similarly obtain two linear equations for ${\tilde{f}=\tilde{f}_s=\tilde{f}_\phi}$ and $\smash{\tilde{k}=\tilde{f}_s\tilde{\kappa}_s=\tilde{f}_\phi\tilde{\kappa}_\phi}$, from the solution of which the numerical integration can be continued.

Varying the energy with respect to geometric variables, as we have done above, obviates the problem of elastic compatibility. This is the question---independent of the problem of incompatibility of the intrinsic configuration~\cite{goriely} that we have discussed when setting up the geometry of the intrinsic configuration---whether a deformation exists that produces a given set of strains and that provides one of the F\"oppl--von K\'arm\'an equations of plate theory~\cite{audoly}. In this context, this discussion of the numerical approach to solving Eqs.~\eqref{eq:dNTB} shows explicitly how they give rise to a compatible configuration, and therefore how they avoid the problem of elastic compatibility.
\section{\uppercase{Neo-Hookean Relations as the Thin Shell Limit of General Constitutive Relations}}\label{appC}
In this final Appendix, we show that the effective two-dimensional constitutive relations resulting from Eq.~\eqref{eq:I1A2},
\begin{align}
e=C\varepsilon^2\left[\bigl(\tr{\hat{\mat{E}}}\bigr)^2+\tr{\hat{\mat{E}}^2}\right]+O\bigl(\varepsilon^3\bigr),\label{eq:e0C}
\end{align}
are general and therefore do not only apply to the incompressible neo-Hookean three-dimensional constitutive relations assumed in Eq.~\eqref{eq:E}. To prove this, we consider, following Ref.~\cite{dervaux09}, incompressible isotropic energy densities expressible as a general power series
\begin{align}
e=\dfrac{1}{2}\sum_{m=0}^\infty{\sum_{n=0}^\infty}{C_{mn}\left(\mathcal{I}_1-3\right)^m\left(\mathcal{I}_2-3\right)^n}, \label{eq:edensC}
\end{align}
where $\mathcal{I}_1=\tr{\tens{C}}$ and $\mathcal{I}_2=\bigl(\mathcal{I}_{\smash{1}}^2-\tr{\tens{C}^2}\bigr)/2$ are the first two invariants of the Cauchy--Green tensor $\tens{C}=\tens{F}^\top\tens{F}$. We may set $C_{00}=0$ without loss of generality. The requirement that $e\geq 0$ for small, linearly elastic deformations~\cite{ogden} then leads to $C_{10}+C_{01}\geq0$. For $C_{10}+C_{01}=0$, the material has no linear elastic response (i.e. zero bulk modulus); we do not consider that case, and therefore assume that $C_{10}+C_{01}>0$. 

Using a result of Ref.~\cite{dervaux09} and the notation of Appendix~\ref{appA}, the Cauchy stress tensor for this material is
\begin{subequations}\label{eq:TCC}
\begin{align}
\tens{T}&=2\left[e_{,\mathcal{I}_1}\tens{F}+e_{,\mathcal{I}_2}\left(\mathcal{I}_1\tens{F}-\tens{FC}\right)\right]\tens{F}^\top-P\tens{I},
\end{align}
and hence the morphoelastic Piola--Kirchhoff tensor introduced in Eq.~\eqref{eq:P} is
\begin{align}
\tens{P}=\tens{T}\btilde{\tens{F}}^{-\top}=2\left[e_{,\mathcal{I}_1}\tens{F}+e_{,\mathcal{I}_2}\left(\mathcal{I}_1\tens{F}-\tens{FC}\right)\right]\bigl(\tens{F^0}\bigr)^{-\top}-P\btilde{\tens{F}}^{-\top}. \label{eq:QC}
\end{align}
\end{subequations}
In Eqs.~\eqref{eq:TCC}, $\btilde{\tens{F}}$, $\tens{F^0}$, and $\tens{F}=\btilde{\tens{F}}\bigl(\tens{F^0}\bigr)^{-1}$ are given by Eqs.~\eqref{eq:Fg2}, \eqref{eq:F0bm}, and \eqref{eq:FBB0}, respectively, and $P=P_{(0)}+O(\varepsilon)$ is pressure. (We now use an uppercase letter to denote pressure to emphasise that it is scaled differently to Appendix~\ref{appA}; in the notation used there, $P=Cp$.)
\paragraph*{Expansion and partial solution at order $O(1)$.}
From the leading-order expansion of $\tens{F}$ in Eq.~\eqref{eq:F0A} and using Corollary~\ref{cor2} and $\tilde{\mat{g}}=\mat{g^0}+O(\varepsilon)$ from Eq.~\eqref{eq:uscale}, we compute
\begin{align}
\tens{F}^\top=\left(\begin{array}{c|c}
\bigl(\mat{g^0}\bigr)^{-1}\mat{B}^\top\mat{g^0}&\bigl(\mat{g^0}\bigr)^{-1}\mat{w}\\[0.5mm]\hline
\mat{v}^\top\mat{g^0}&c
\end{array}\right)+O(\varepsilon), 
\end{align}
in which $\mat{B},\mat{v},\mat{w},c$ are given by Eqs.~\eqref{eq:Fdefs}, and thence
\begin{align}
\tens{C}&=\left(\begin{array}{c|c}
\!\bigl(\mat{g^0}\bigr)^{-1}\mat{B}^\top\mat{g^0}\mat{B}+\bigl(\mat{g^0}\bigr)^{-1}\mat{w}\mat{w}^\top&\bigl(\mat{g^0}\bigr)^{-1}\mat{B}^\top\mat{g^0}\mat{v}+c\bigl(\mat{g^0}\bigr)^{-1}\mat{w}\!\\[0.5mm]
\hline
\mat{v}^\top\mat{g^0}\mat{B}+c\mat{w}^\top&\mat{v}^\top\mat{g^0}\mat{v}+c^2
\end{array}\right)\nonumber\\
&\qquad+O(\varepsilon), 
\end{align}
In particular,
\begin{align}
\mathcal{I}_1=\tr{\left(\bigl(\mat{g^0}\bigr)^{-1}\mat{B}^\top\mat{g^0}\mat{B}\right)}+\mat{w}^\top\bigl(\mat{g^0}\bigr)^{-1}\mat{w}+\mat{v}^\top\mat{g^0}\mat{v}+c^2+O(\varepsilon). 
\end{align}
Since the incompressibility condition is independent of the constitutive relations, its leading-order expansion~\eqref{eq:detF0A} still holds true. Using this and the leading-order expansions~\eqref{eq:F0invT} and \eqref{eq:FginvT} and writing $e_{,\mathcal{I}_1}=E_1+O(\varepsilon)$, $e_{,\mathcal{I}_2}=E_2+O(\varepsilon)$, Eq.~\eqref{eq:QC} yields
\begin{align}
\tens{P}&=\dfrac{1}{{\zeta^0}_{,\zeta}}\left(\begin{array}{c|c}
\!O(1)&\!\begin{array}{l}
2\left\{E_1+E_2\left[\tr{\left(\bigl(\mat{g^0}\bigr)^{-1}\mat{B}^\top\mat{g^0}\mat{B}\right)}\right.\right.\\
\hspace{25mm}+\left.\left.\mat{w}^\top\bigl(\mat{g^0}\bigr)^{-1}\mat{w}\right]\right\}\mat{v}\\
\phantom{.}-2E_2\left(\mat{B}\bigl(\mat{g^0}\bigr)^{-1}\mat{B}^\top\mat{g^0}\mat{v}+c\mat{B}\bigl(\mat{g^0}\bigr)^{-1}\mat{w}\right)\\
\phantom{.}+P_{(0)}(\det{\mat{B}})\bigl(\mat{g^0}\bigr)^{-1}\mat{B}^{-\top}\mat{w}
\end{array}\!\\[10.5mm]
\hline
\!O(1)&\!\begin{array}{l}
\vphantom{A^{A^{A^{A^A}}}}2c\left[E_1+E_2\tr{\left(\bigl(\mat{g^0}\bigr)^{-1}\mat{B}^\top\mat{g^0}\mat{B}\right)}\right]\\\phantom{.}-2E_2\mat{w}^\top\bigl(\mat{g^0}\bigr)^{-1}\mat{B}^\top\mat{g^0}\mat{v}-P_{(0)}\det{\mat{B}}
\end{array}
\end{array}\right)\nonumber\\&\qquad+O(\varepsilon). \label{eq:Q0C}
\end{align}
Writing $\tens{P}=\tens{P_{(0)}}+\varepsilon\tens{P_{(1)}}+O\bigl(\varepsilon^2\bigr)$, the leading-order boundary condition is $\tens{P_{(0)}}\vec{n}\equiv\vec{0}$, similary to Appendix~\ref{appA}. Hence, from Eqs.~\eqref{eq:detF0A} and \eqref{eq:Q0C}, the leading-order problem is
\begin{subequations}\label{eq:eqsC}
\begin{align}
&c-\mat{w}^\top\mat{B}^{-1}\mat{v}=(\det{\mat{B}})^{-1},
\end{align}
\begin{align}
&2\left\{E_1+E_2\left[\tr{\left(\bigl(\mat{g^0}\bigr)^{-1}\mat{B}^\top\mat{g^0}\mat{B}\right)}+\mat{w}^\top\bigl(\mat{g^0}\bigr)^{-1}\mat{w}\right]\right\}\mat{v}\nonumber\\
&\qquad -2E_2\left(\mat{B}\bigl(\mat{g^0}\bigr)^{-1}\mat{B}^\top\mat{g^0}\mat{v}+c\mat{B}\bigl(\mat{g^0}\bigr)^{-1}\mat{w}\right)\nonumber\\
&\qquad+P_{(0)}(\det{\mat{B}})\bigl(\mat{g^0}\bigr)^{-1}\mat{B}^{-\top}\mat{w}=\mat{0},\\
&2c\left[E_1+E_2\tr{\left(\bigl(\mat{g^0}\bigr)^{-1}\mat{B}^\top\mat{g^0}\mat{B}\right)}\right]-2E_2\mat{w}^\top\bigl(\mat{g^0}\bigr)^{-1}\mat{B}^\top\mat{g^0}\mat{v}\nonumber\\
&\qquad-P_{(0)}\det{\mat{B}}=0.
\end{align}
\end{subequations}
These equations have a trivial solution
\begin{align}
&Z_{(0)}\equiv Z^0,&&\mat{S_{(0)}}\equiv\mat{0},&P_{(0)}=C_{10}+2C_{01},\label{eq:solC}
\end{align}
analogous to the leading-order solution found in Appendix~\ref{appA} and for which, from Eqs.~\eqref{eq:Fdefs}, $\mat{B}=\mat{I}$, $\mat{v}=\mat{w}=\mat{0}$, $c=1$, and hence $\tens{C}=\tens{I}+O(\varepsilon)$, so that $\mathcal{I}_1=\mathcal{I}_2=3+O(\varepsilon)$ and thus $E_1=C_{10}/2$, $E_2=C_{01}/2$ from Eq.~\eqref{eq:edensC}. We were not however able to show that this is the only solution of the nonlinear first-order differential equations for $Z_{(0)}$, $\mat{S_{(0)}}$ as functions of $Z^0$ provided by Eqs.~\eqref{eq:eqsC} that satisfies the conditions $Z_{(0)}=0$, $\smash{\mat{S_{(0)}}}=\mat{0}$ on the midsurface $\smash{Z^0}=0$. In this respect, our solution of the leading-order problem remains partial.

Our failure to solve Eqs.~\eqref{eq:eqsC} emphasises once again that what distinguishes these problems of large bending deformations from classical problems in elastic shell theories is the fact that the leading-order problem for large bending deformations is not trivial. In fact, were a second solution of Eqs.~\eqref{eq:eqsC} to exist, global energy considerations would select the solution; this would open a new can of worms in the analysis.

\paragraph*{Expansion at order $O(\varepsilon)$.} At this stage, we take Eqs.~\eqref{eq:solC} as the solution of the leading-order problem~\eqref{eq:eqsC} and proceed thence.  In particular, the deformation gradient still has an expansion of the form in Eq.~\eqref{eq:F2A}. Hence Eq.~\eqref{eq:FT2A} still holds true, and we compute
\begin{subequations}
\begin{align}
\tens{C}&=\left(\begin{array}{c|c}
\!\mat{I}+\varepsilon\left(2\mat{E}\!+\!\mat{B_{(1)}}\!+\!\bigl(\mat{g^0}\bigr)^{-1}\mat{B^\top_{\smash{\mat{(1)}}}}\mat{g^0}\right)&\varepsilon\left(\mat{v_{(1)}}\!+\!\bigl(\mat{g^0}\bigr)^{-1}\mat{w_{(1)}}\right)\!\\[1.5mm]
\hline
\varepsilon\left(\mat{v^\top_{\smash{\mat{(1)}}}}\mat{g^0}+\mat{w^\top_{\smash{\mat{(1)}}}}\right)&1+2\varepsilon c_{(1)}
\end{array}\right)\nonumber\\
&\qquad+O\bigl(\varepsilon^2\bigr),\label{eq:CC}\\
\tens{C}^2&=\left(\begin{array}{c|c}
\!\mat{I}+2\varepsilon\left(2\mat{E}\!+\!\mat{B_{(1)}}\!+\!\bigl(\mat{g^0}\bigr)^{-1}\mat{B^\top_{\smash{\mat{(1)}}}}\mat{g^0}\right)&O(\varepsilon)\\[1.5mm]
\hline
O(\varepsilon)&1+4\varepsilon c_{(1)}\!
\end{array}\right)
+O\bigl(\varepsilon^2\bigr),
\end{align}
\end{subequations}
whence
\begin{subequations}
\begin{align}
\mathcal{I}_1&=3+\varepsilon\bigl[2\bigl(\tr{\mat{E}}+\tr{\mat{B_{(1)}}}+c_{(1)}\bigr)\bigr]+O\bigl(\varepsilon^2\bigr),\\
\mathcal{I}_2&=3+\varepsilon\bigl[4\bigl(\tr{\mat{E}}+\tr{\mat{B_{(1)}}}+c_{(1)}\bigr)\bigr]+O\bigl(\varepsilon^2\bigr). 
\end{align}
\end{subequations}
The incompressibility condition being independent of the constitutive relations, Eq.~\eqref{eq:detF2A} and hence the first of Eqs.~\eqref{eq:c1A} still hold. The latter implies $\mathcal{I}_1=\mathcal{I}_2=3+O\bigl(\varepsilon^2\bigr)$. Thus
\begin{align}
e=\dfrac{1}{2}\bigl[C_{10}\left(\mathcal{I}_1-3\right)+C_{01}\left(\mathcal{I}_2-3\right)\bigr]+O\bigl(\varepsilon^4\bigr), \label{eq:eC}
\end{align}
and, in particular, $e_{,\mathcal{I}_1}=C_{10}/2+O\bigl(\varepsilon^2\bigr)$, ${e_{,\mathcal{I}_2}=C_{01}/2+O\bigl(\varepsilon^2\bigr)}$. In this way, the constitutive relations have reduced, up to smaller corrections, to those of a Mooney--Rivlin solid~\cite{goriely}. Moreover, Eq.~\eqref{eq:F02A} and hence Eqs.~\eqref{eq:F02invT} and \eqref{eq:Fg2invT} still hold. Since $P=C_{10}+2C_{01}+O(\varepsilon)$, it follows that
\begin{widetext}
\phantom{.}\vspace{-5mm}\begin{align}
\tens{P}&= \left(\begin{array}{c|c}
O(\varepsilon)\vphantom{\left(\dfrac{A}{A}\right)^2}&\varepsilon\dfrac{C_{10}+C_{01}}{c_{(0)}^0}\left(\mat{v_{(1)}}+\bigl(\mat{g^0}\bigr)^{-1}\mat{w_{(1)}}\right)+O\bigl(\varepsilon^2\bigr)\\[4mm]
\hline
O(\varepsilon)&O(\varepsilon)
\end{array}\right)\quad\Longrightarrow\quad 
\tens{P_{(0)}}=\tens{O},\;\tens{P_{(1)}}\vec{n}=\left(\begin{array}{c}
\dfrac{C_{10}+C_{01}}{c^0_{(0)}}\left(\mat{v_{(1)}}+\bigl(\mat{g^0}\bigr)^{-1}\mat{w_{(1)}}\right)\\[4mm]
\hline
O(1)
\end{array}\right).
\end{align}
Similarly to Appendix~\ref{appA}, the boundary conditions now imply $\tens{P_{(1)}}\vec{n}\equiv\vec{0}$, so, noting that $\smash{c^0_{\smash{(0)}}}>0$ and $\smash{C_{10}+C_{01}}>0$, the second of Eqs.~\eqref{eq:c1A} also still holds.
\paragraph*{Expansion at order $O\bigl(\varepsilon^2\bigr)$.} Since the expansion~\eqref{eq:detF2A} of the incompressibility condition still holds, Eqs.~\eqref{eq:c1A} still imply Eq.~\eqref{eq:c2A} and hence Eq.~\eqref{eq:I1A1}. Meanwhile, Eqs.~\eqref{eq:c1A} and \eqref{eq:CC} show that the off-diagonal terms in Eq.~\eqref{eq:CA} are in fact of order $O\bigl(\varepsilon^2\bigr)$, so it follows from Eq.~\eqref{eq:CA} that
\begin{align}
\tr{\tens{C}^2}&=\tr{\Bigl\{\mat{I}+\varepsilon\left(2\mat{E}+\mat{B_{(1)}}+\bigl(\mat{g^0}\bigr)^{-1}\mat{B^\top_{\smash{\mat{(1)}}}}\mat{g^0}\right)+\varepsilon^2\left[2\left(\mat{EB_{(1)}}+\bigl(\mat{g^0}\bigr)^{-1}\mat{B^\top_{\smash{\mat{(1)}}}}\mat{g^0E}\right)\right.+\left.\mat{B_{(2)}}+\bigl(\mat{g^0}\bigr)^{-1}\mat{B^\top_{\smash{\mat{(2)}}}}\mat{g^0}+\bigl(\mat{g^0}\bigr)^{-1}\mat{B^\top_{\smash{\mat{(1)}}}}\mat{g^0B_{(1)}}\right.}\nonumber\\
&\qquad+\left.\bigl(\mat{g^0}\bigr)^{-1}\mat{w_{(1)}}\mat{w^\top_{\smash{\mat{(1)}}}}\right]+O\bigl(\varepsilon^3\bigr)\Bigr\}^2+\bigl[1+2\varepsilon c_{(1)}+\varepsilon^2\bigl(2c_{(2)}+c_{\smash{(1)}}^2+\mat{v^\top_{\smash{\mat{(1)}}}}\mat{g^0v_{(1)}}\bigr)+O\bigl(\varepsilon^3\bigr) \bigr]^2+O\bigl(\varepsilon^4\bigr)\nonumber\\
&=3+4\varepsilon\left(c_{(1)}+\tr{\mat{E}}+\tr{\mat{B_{(1)}}}\right)+2\varepsilon^2\left[2\tr{\mat{E}^2}+4\tr{\left(\mat{EB_{(1)}}\right)}+4\tr{\left(\mat{E}\bigl(\mat{g^0}\bigr)^{-1}\mat{B^\top_{\smash{\mat{(1)}}}}\mat{g^0}\right)}+\tr{\mat{B}_{\smash{\mat{1}}}^2}+2\tr{\left(\mat{B_{(1)}}\bigl(\mat{g^0}\bigr)^{-1}\mat{B^\top_{\smash{\mat{(1)}}}}\mat{g^0}\right)}\right.\nonumber\\
&\qquad+\left.2\tr{\mat{B_{(2)}}}+\mat{w^\top_{\smash{\mat{(1)}}}}\bigl(\mat{g^0}\bigr)^{-1}\mat{w_{(1)}}+3c_{(1)}^2+2c_{(2)}+\mat{v^\top_{\smash{\mat{(1)}}}}\mat{g^0v_{(1)}}\right]+O\bigl(\varepsilon^3\bigr)\nonumber\\
&=3+4\varepsilon^2\left\{2\left(\tr{\mat{E}}+\tr{\mat{B_{(1)}}}\right)^2+2\tr{\mat{E}^2}+\tr{\mat{B}^2_{\mat{(1)}}}+\tr{\left(\bigl(\mat{g^0}\bigr)^{-1}\mat{B^\top_{\smash{\mat{(1)}}}}\mat{g^0}\mat{B_{(1)}}\right)}+2\left[\tr{\left(\mat{EB_{(1)}}\right)}+\tr{\left(\mat{E}\bigl(\mat{g^0}\bigr)^{-1}\mat{B^\top_{\smash{\mat{(1)}}}}\mat{g^0}\right)}\right]\right\}+O\bigl(\varepsilon^3\bigr),\label{eq:I2C}
\end{align}

\vspace{-3mm}
\end{widetext}
using Eqs.~\eqref{eq:c1c2A} and Lemma~\ref{lemma4}, similarly to the calculations leading up to Eq.~\eqref{eq:I1A1}. 

Finally, if we write $\mathcal{I}_1=3+\varepsilon^2\smash{I_{(2)}}+O\bigl(\varepsilon^3\bigr)$ using Eq.~\eqref{eq:I1A1}, then Eq.~\eqref{eq:I2C} shows that $\smash{\tr{\tens{C}^2}=3+4\varepsilon^2\smash{I_{(2)}}+O\bigl(\varepsilon^3\bigr)}$. These expansions imply that $\smash{\mathcal{I}_2=3+\varepsilon^2\smash{I_{(2)}}+O\bigl(\varepsilon^3\bigr)}$. Equivalently, ${\mathcal{I}_2-3=\mathcal{I}_1-3+O\bigl(\varepsilon^3\bigr)}$. Hence, from Eq.~\eqref{eq:eC}, 
\begin{align}
e=\dfrac{C}{2}(\mathcal{I}_1-3)+O\bigl(\varepsilon^3\bigr),\quad\text{with }C=C_{10}+C_{01}>0. 
\end{align}
Up to smaller corrections, these are the neo-Hookean constitutive relations assumed in Eq.~\eqref{eq:E} and throughout Section~\ref{sec:model} and Appendix~\ref{appA}, and which, as shown there, indeed reduce at order $O\bigl(\varepsilon^2\bigr)$ to the effective two-dimensional constitutive relations in Eq.~\eqref{eq:e0C}. 
Assuming that the trivial solution~\eqref{eq:solC} of the leading-order problem defined by Eqs.~\eqref{eq:eqsC} is unique, this proves our claim in Section~\ref{sec:concl}, that these effective two-dimensional constitutive relations are general.\vspace{1mm}
\bibliography{shells}
\end{document}